\documentclass{article}

\usepackage[utf8]{inputenc}
\usepackage{pdfsync}

\def\showauthornotes{0}
\def\showkeys{0}

\def\confversion{0}
\def\widemargin{0}

\usepackage{xspace,enumerate}
\usepackage{amsmath,amssymb}
\usepackage{amsthm}
\ifnum\confversion=0
	\usepackage[toc,page]{appendix}
\fi
\usepackage{thmtools}
\usepackage{thm-restate}
\usepackage{graphicx}
\usepackage{boxedminipage}
\usepackage{csquotes}
\usepackage{makecell}
\usepackage{tabularx}
\usepackage{relsize}
\usepackage{scalerel}
\usepackage[dvipsnames]{xcolor} %% Load before TikZ
\usepackage{tikz}

\ifnum\showkeys=1
\usepackage[color]{showkeys}
\fi

\definecolor{darkred}{rgb}{0.5,0,0}
\definecolor{darkgreen}{rgb}{0,0.35,0}
\definecolor{darkblue}{rgb}{0,0,0.55}

\ifnum\confversion=1
\usepackage[pdfstartview=FitH,pdfpagemode=UseNone,colorlinks=false,linkcolor=darkblue,filecolor=darkred,citecolor=darkgreen,urlcolor=darkred,pagebackref,bookmarks=false]{hyperref}
\else
\usepackage[pdfstartview=FitH,pdfpagemode=UseNone,colorlinks,linkcolor=NavyBlue,filecolor=blue,citecolor=OliveGreen,urlcolor=NavyBlue,pagebackref]{hyperref}
\fi

\usepackage[capitalise,nameinlink]{cleveref}
\usepackage[T1]{fontenc}

\usepackage{mathtools,dsfont}

\usepackage{mathpazo}
\usepackage{microtype}
\ifnum\widemargin=0
\usepackage[top=1in, bottom=1in, left=1in, right=1in]{geometry}
\else
\usepackage[top=1in, bottom=1in, left=1.25in, right=1.25in]{geometry}
\fi
\ifnum\showauthornotes=1
\newcommand{\Authornote}[3]{{\sf\small\color{#3}{[#1: #2]}}}
\newcommand{\Authorcomment}[2]{{\sf \small\color{gray}{[#1: #2]}}}
\newcommand{\Authorfnote}[2]{\footnote{\color{red}{#1: #2}}}
\else
\newcommand{\Authornote}[3]{}
\newcommand{\Authorcomment}[2]{}
\newcommand{\Authorfnote}[2]{}
\fi

\declaretheorem[numberwithin=section]{theorem}
\declaretheorem[sibling=theorem]{lemma}
\declaretheorem[sibling=theorem]{claim}

\declaretheorem[sibling=theorem]{corollary}

\theoremstyle{definition}
\declaretheorem[sibling=theorem]{definition}
\declaretheorem[sibling=theorem]{remark}

\newtheorem{algo}[theorem]{Algorithm}

\newenvironment{algorithm}[3]
        {\noindent\begin{boxedminipage}{\textwidth}\begin{algo}[#1]\ \par
        {\begin{tabular}{r l}
        \textbf{Input} & #2\\
        \textbf{Output} & #3
        \end{tabular}\par\enskip}}
        {\end{algo}\end{boxedminipage}}

\def\FullBox{\hbox{\vrule width 6pt height 6pt depth 0pt}}

\def\qed{\ifmmode\qquad\FullBox\else{\unskip\nobreak\hfil
\penalty50\hskip1em\null\nobreak\hfil\FullBox
\parfillskip=0pt\finalhyphendemerits=0\endgraf}\fi}

\def\qedsketch{\ifmmode\Box\else{\unskip\nobreak\hfil
\penalty50\hskip1em\null\nobreak\hfil$\Box$
\parfillskip=0pt\finalhyphendemerits=0\endgraf}\fi}

\def\to{\rightarrow}
\def\eps{\varepsilon}
\def\epsilon{\varepsilon}

\def\eps{\epsilon}

\def\phi{\varphi}
\def\cal{\mathcal}

\def\implies{\Rightarrow}
\newcommand{\sub}{\ensuremath{\subseteq}}

\newcommand{\defeq}{:=}

\renewcommand{\bar}{\overline} 

\newcommand{\given}{\;\ifnum\currentgrouptype=16 \middle\fi \vert\;}
\newcommand{\ie}{i.e.,\xspace}

\newcommand{\etal}{et al.\xspace}

\newcommand{\mper}{\,.}
\newcommand{\mcom}{\,,}

\usepackage{nicefrac}
\let\nfrac=\nicefrac

\DeclarePairedDelimiter\parens{\lparen}{\rparen}
\DeclarePairedDelimiter\abs{\lvert}{\rvert}
\DeclarePairedDelimiter\norm{\lVert}{\rVert}

\DeclarePairedDelimiter\braces{\lbrace}{\rbrace}
\DeclarePairedDelimiter\brackets{\lbrack}{\rbrack}

\DeclarePairedDelimiterXPP\lnorm[1]{}\lVert\rVert{_2}{#1}

\newcommand{\inparen}[1]{\left(#1\right)}             %\inparen{x+y}  is (x+y)
\newcommand{\inbraces}[1]{\left\{#1\right\}}           %\inbrace{x+y}  is {x+y}
\newcommand{\insquare}[1]{\left[#1\right]}             %\insquare{x+y}  is [x+y]
\newcommand{\one}{{\mathbf{1}}}

\newcommand{\Esymb}{\mathbb{E}}
\newcommand{\Psymb}{\mathbb{P}}

\newcommand{\Varsymb}{\mathrm{Var}}
\DeclareMathOperator*{\ExpOp}{\Esymb}

\makeatletter
\def\Pr#1{%
    \ProbabilityRender{\Psymb}{#1}%
}

\def\Ex#1{%
    \ProbabilityRender{\Esymb}{#1}%
}

\def\tildeEx#1{%
    \ProbabilityRender{\widetilde{\Esymb}}{#1}%
}

\def\condPE#1#2{%
	\@ifnextchar\bgroup
	{\ConditionalProbabilityRender{\widetilde{\Esymb}}{#1}{#2}}
	{\ProbabilityRender{\widetilde{\Esymb}}{#1 \given #2}}
}

\def\tildeVar#1{
	\ProbabilityRender{\widetilde{\Varsymb}}{#1}
}

\def\tildeCov#1#2{
	\ProbabilityRender{\tildecov}{#1,#2}
}

\def\ConditionalProbabilityRender#1#2#3#4{
	\renderwithdist{#1}{#2}{#3 \given #4}	
}

\def\ProbabilityRender#1#2{%fancy probability command
  \@ifnextchar\bgroup%
  {\renderwithdist{#1}{#2}}
   {\singlervrender{#1}{#2}}
}
\def\singlervrender#1#2{%
   \ensuremath{\mathchoice
       {{#1}\insquare{ #2 }}
       {{#1}\insquare{ #2 }}
       {{#1}\insquare{ #2 }}
       {{#1}\insquare{ #2 }}
   }
}
\def\renderwithdist#1#2#3{%
   \@ifnextchar\bgroup
   {\superfancyrender{#1}{#2}{#3}}
   {\ensuremath{\mathchoice
      {\underset{#2}{#1}\insquare{ #3 }}
      {{#1}_{#2}\insquare{ #3 }}
      {{#1}_{#2}\insquare{ #3 }}
      {{#1}_{#2}\insquare{ #3 }}
     }
   }
}
\def\superfancyrender#1#2#3#4#5{
   \ensuremath{\mathchoice
      {\underset{#1}{{#1}}\left#4 #3 \right#5}
      {{#1}_{#2}#4 #3 #5}
      {{#1}_{#2}#4 #3 #5}
      {{#1}_{#2}#4 #3 #5}
   }
}
\makeatother

\newcommand\restrict[1]{\raisebox{-.5ex}{$|$}_{#1}}

\def\li{{ \ell }}
\def\ri{{ r }}
\def\zee{{ \mathbf{Z} }}
\def\hl{ h_{\ell} }
\newcommand{\hil}[1]{h_{#1,\ell}}
\def\gl{ \erase{g}_{\ell} }

\def\fjl{ f_{j,\ell} }
\def\sl{ s_{\ell} }
\def\Tau{{\mathrm{T}}}
\newcommand{\erase}[1]{\overline{#1}}
\newcommand{\ip}[2] {\ensuremath{\left\langle #1 , #2 \right\rangle}}

\def\fd{{ \mathsf{F} }}
\def\dd{{ \mathsf{D} }}

\def\AELC{{ \calC_{\mathrm{AEL}} }}
\def\AEL{\mathrm{AEL}}
\newcommand{\tildeE}{\widetilde{\mathbb{E}}}

\newcommand{\pcod}[2]{\widetilde{\ExpOp}^{(#1)}\left[ #2 \right]}

\newcommand\inv[1]{#1\raisebox{1.15ex}{$\scriptscriptstyle-1$}}
\newcommand\scale[2]{\vstretch{#1}{\hstretch{#1}{#2}}}
\newcommand\gj{\overline{g\raisebox{0.5ex}{\scale{.65}{(j)}}
}}
\newcommand{\T}{\intercal} % Transpose symbol

\def\indi#1{{\one \inbraces{ #1 } }}

\DeclareMathOperator{\tildecov}{\widetilde{\operatorname {Cov}}}

\newcommand{\R}{{\mathbb R}}

\newcommand{\N}{{\mathbb{N}}}

\newcommand{\F}{{\mathbb F}}

\newcommand{\im}{\mathrm{im}}
\newcommand{\inn}{\mathrm{in}}
\newcommand{\out}{\mathrm{out}}
\newcommand{\dec}{\mathrm{dec}}
\newcommand{\dis}{{\Delta}}

\newcommand{\deffont}{\sf}

\newcommand{\calC}{{\cal C}}
\newcommand{\calD}{{\cal D}}

\newcommand{\calH}{{\cal H}}

\newcommand{\calL}{{\cal L}}

\newcommand{\calT}{{\cal T}}
\newcommand{\calS}{{\cal S}}
\newcommand{\cC}{{\cal C}}

\def\cG{\mathcal G}

\def\DECODE{Decode-from-distributions}

\title{Explicit Codes approaching Generalized Singleton Bound using Expanders}
\author{
Fernando Granha Jeronimo\thanks{{\tt University of Illinois Urbana-Champaign}. {\tt granha@illinois.edu}. }
\and
Tushant Mittal\thanks{{\tt Stanford University}. {\tt tushant@stanford.edu}. Partly supported by NSF grant CCF-2326685.}
\and
Shashank Srivastava\thanks{{\tt DIMACS (Rutgers) \& Institute for Advanced Study}. {\tt shashanks@ias.edu}. Partly supported by the NSF grant CCF-2326685.}
\and
Madhur Tulsiani\thanks{{\tt Toyota Technological Institute at Chicago}. {\tt madhurt@ttic.edu}. Supported by the NSF grant CCF-2326685.} 
}

\allowdisplaybreaks
\begin{document}

\date{}
\maketitle

\thispagestyle{empty}

\begin{abstract}
We construct a new family of explicit codes that are list decodable to capacity and achieve an optimal list size of $O(\nfrac{1}{\eps})$. In contrast to existing explicit constructions of codes achieving list decoding capacity, our arguments do not rely on algebraic structure but utilize simple combinatorial properties of expander graphs. 
\medskip

% Our construction is based on a celebrated distance amplification procedure due to Alon, Edmonds, and Luby [FOCS'95], and our result can be interpreted as a ``local-to-global'' phenomenon for (a slight strengthening of) the generalized Singleton bound. 
%
Our construction is based on a celebrated distance amplification procedure due to Alon, Edmonds, and Luby [FOCS'95], which transforms any high-rate code into one with near-optimal rate-distance tradeoff. We generalize it to show that the same procedure can be used to transform any high-rate code into one that achieves list decoding capacity. Our proof can be interpreted as a ``local-to-global’' phenomenon for (a slight strengthening of) the generalized Singleton bound.
Using this construction,  for every $R, \eps \in (0,1)$ and $k \in \N^+$, we obtain an \emph{explicit} family of codes $\mathcal{C} \subseteq \Sigma^n$, with rate $R$ such that,\begin{itemize}
	\item They achieve the $\eps$-relaxed generalized Singleton bound: for any $g \in  \Sigma^n$ and any list $\calH$ of at most $k$ codewords, we have, 
\[ \Ex{h \in \calH}{\Delta(g,h)} ~\geq~ \frac{\abs{\calH}-1}{\abs{\calH}} \cdot (1 - R - \eps) \mper
\]

	\item The alphabet size is a constant depending only on $\eps$ and $k$.
	\item They can be list decoded up to radius $\frac{k-1}{k}(1-R-\eps)$, in time $n^{O_{k,\eps}(1)}$.
\end{itemize}
\medskip

As a corollary of our result, we also obtain the first explicit construction of LDPC codes achieving list decoding capacity, and in fact arbitrarily close to the generalized Singleton bound.

\medskip

% The only known explicit construction of codes achieving the generalized Singleton bound are folded Reed--Solomon code and univariate multiplicity codes [Chen, Zhang 2024], both of which require an alphabet size growing with $n$. 

%the alphabet size $\abs{\Sigma}$ of the code $\mathcal{C}$ can be taken to be $2^{(k^k/\eps)^{O(1)}}$.
%
%For any choice $R$, we give an explicit construction of codes with rate $R$, such that  for every fixed list of code $k$, with a constant alphabet size depending only on $\eps$ and $k$. (stuff about ldpc/ additive/linear). Decoding in polynomial time.
\end{abstract}

% \newpage

% \pagenumbering{roman}
% \tableofcontents
% \clearpage

\newpage
\pagenumbering{arabic}
\setcounter{page}{1}

\section{Introduction}
An error-correcting code $\calC \subseteq \Sigma^n$ is a collection of strings (codewords) over a finite alphabet $\Sigma$, such that every pair of codewords is well separated. 
Two important parameters associated with a code $\calC$ are its distance $\delta(\calC)$ and rate $\rho(\calC)$, which measure respectively the worst-case error correction capability of a code, and the amount of redundancy present in the code. These are defined as
\[
\delta(\calC) 
~=~ \min_{\underset{f\neq g}{f,g \in \calC}} \frac{1}{n} \cdot \abs{\inbraces{i \in [n] ~|~ f_i \neq g_i}}
\qquad \text{and} \qquad
\rho(\calC) ~\defeq~ \frac{1}{n} \cdot \log_{\abs{\Sigma}}\abs{\calC} \mper
\]

As may be expected, the separation between codewords (distance) involves a tradeoff with the number of codewords (rate). This is captuted by the well-known Singleton bound, (the asymptotic version of) which implies that for any code $\calC$, we must have $\delta(\calC) \leq 1 - \rho(\calC)$. 
The study of the achievable tradeoffs between rate $\rho$ and distance $\delta$ is a fundamental problem in the theory of error-correcting codes, with the tradeoffs given by the Singleton bound being optimal in the regime of large alphabets $\Sigma$. 

Several code families are known to achieve the Singleton bound including the celebrated Reed-Solomon codes~\cite{RS60}, multiplicity codes~\cite{KSY14, Kop15}, and many others~\cite{GRS23}. 
While the above codes achieving the Singleton bound often require an alphabet size growing with $n$, a particularly relevant construction by Alon, Edmonds, and Luby~\cite{AEL95} also yields near-optimal tradeoffs with a \emph{constant} alphabet size. 
In particular, this construction gives $\delta \geq 1 - \rho - \eps$ using an alphabet of size $2^{1/\eps^{O(1)}}$. 

Moreover, the above code constructions with near-optimal distances also enable highly efficient algorithms for the task of unique decoding, which is the task of recovering a codeword uniquely from, say a $\beta$ fraction of adversarial errors. 
For a code with distance $\delta$, we must have $\beta < \delta/2$, and it is indeed possible to achieve this algorithmically using the structure of several code constructions above. 
However, unique recovery is no longer possible when $\beta \geq \delta/2$, since the corruptions may be consistent with more than one codewords.  

\vspace{-5 pt}
\paragraph{List decoding capacity.}
The notion of list decoding defined by Elias~\cite{E57} and Wozencraft~\cite{W58} allows the output to be a small \emph{list} of codewords, when the fraction of errors is higher than the unique-decoding radius of $\delta/2$.
In this setting, both the \emph{combinatorial} question of bounding the number of codewords in the list, and the \emph{algorithmic} one of finding the list, are quite non-trivial and have been the subject of extensive research. Algorithms for list decoding were considered in the foundational works of Guruswami and Sudan~\cite{Sudan97, GS98} who gave the first list decoding algorithms for decoding Reed-Solomon codes up to a threshold known as the Johnson bound, where the list size is known to be bounded for any code.
Since then, there has been tremendous progress on both the combinatorial and algorithmic aspects of list decoding, and we refer the reader to excellent references by Guruswami, Rudra, and Sudan~\cite{Guruswami06, GRS23} for more details.
The constructions and algorithms for list-decodable codes have also had numerous applications in the areas of complexity theory~\cite{Sudan00, Trevisan05} and pseudorandomness~\cite{GUV09, Vadhan12}, leading to several deep and beautiful connections between these areas.

In terms of combinatorial bounds for list decoding \emph{beyond} the Johnson bound, volume arguments can be used to show that for the lists to be subexponential in $n$, we must have $\rho \leq 1 - H_q(\beta)$, where $q = \abs{\Sigma}$ and $H_q(x)$ denotes the $q$-ary entropy function (which approximately equals $x$ for large $q$). This bound is called the list decoding capacity, and is known to be achievable using random codes~\cite{GRS23} and random linear codes~\cite{GHK11}.

A breakthrough result of Guruswami and Rudra~\cite{GR06, GR08}, building on the work of Parvaresh and Vardy~\cite{PV05}, gave the first explicit construction of codes achieving list decoding capacity. 
Their construction was based on folded Reed-Solomon codes, and yielded a list of size $n^{O(1/\eps)}$ up to an error radius of $1 - \rho - \eps$, using an alphabet of size $n^{O(1/\eps^2)}$. 
Since then, there has been a series of exciting works on reducing the list size and alphabet size, combining folded Reed-Solomon codes with pseudorandom objects like \textit{subspace evasive sets}~\cite{Gur11, DL12,GW13}, expanders~\cite{KRZSW18} and \textit{subspace designs}~\cite{GK16, GRZ21}. 
A series of recent works have also led to significantly better bounds on the list sizes for folded Reed-Solomon codes via a deeper understanding of their structure, starting with beautiful work of Kopparty \etal~\cite{KRZSW18,KRSW23} and Tamo~\cite{Tamo24} yielding list sizes $(1/\eps)^{O(1/\eps)}$, followed by a bound of $O(1/\eps^2)$ by Srivastava~\cite{Sri25}, and culminating in an asymptotically optimal bound of $O(1/\eps)$ by Chen and Zhang~\cite{CZ24}.

In addition to folded Reed-Solomon codes and their variants, there has also been several other explicit constructions of codes achieving list decoding capacity using algebraic-geometric codes~\cite{Gur09, GX13, GRZ21, GX22}, multiplicity codes~\cite{Kop15}, polynomial ideal codes~\cite{BHKS23} and subcodes of Reed-Solomon codes~\cite{GX13, BST24}. 
The algebraic constructions above also come with efficient list decoding algorithms using polynomial interpolation methods, which have been a cornerstone of techniques in the area of list decoding. 
However, the explicit construction known to achieve a list size of $O(1/\eps)$ at radii $\eps$-close to list decoding capacity, are folded Reed-solomon and multiplicity codes, via the very recent result of Chen and Zhang~\cite{CZ24}. In fact, they show that these codes which in fact achieve two additional strengthenings of the notion of list decoding, as discussed below.

\vspace{-5 pt}
\paragraph{Average-radius list decoding.}
A stronger form of combinatorial list decoding was considered by Guruswami and Narayanan~\cite{GN14}, which shows that the list size at radius $\beta$ is bounded by $k$ by considering a strengthening of the contrapositive. 
Note that the bound on the list size is true if for any $g \in \Sigma^n$ and $h_1, \ldots, h_{k} \in \calC$, we have $\max_{i \in [k]} \Delta(g,h_i) > \beta$, where $\Delta(g,h)$ denotes the fractional Hamming distance as before.
The notion of average-radius list decoding requires establishing that in fact $\Ex{i \in [k]}{\Delta(g,h_i)} > \beta$ which also bounds the list size for a stronger reason. 

Starting with the beautiful works of Rudra and Wootters~\cite{Wootters13, RW14} which introduced techniques from high-dimensional probability, there have been a sequence of results that random linear codes, and random puncturings of structured codes, can in fact achieve list decoding capacity~\cite{FKS22, GST23, GLSTW24}, with several of the results even applicable for the average-radius list decoding~\cite{BGM23, GZ23, AGL24}.  

\vspace{-5 pt}
\paragraph{Generalized Singleton bounds.}
Shangguan and Tamo~\cite{ST20} (and independently Roth~\cite{Rot24}) considered a generalization of the Singleton bound in the context of combinatorial list decoding for Reed-Solomon codes, showing (in the asymptotic version) that for any code $\calC$, if the list size at error-radius $\beta$ is always bounded by $k-1$, then one must have $\beta \leq \frac{k-1}{k} \cdot (1 - \rho(\calC))$. 
While $k=2$ recovers the Singleton bound on the unique-decoding radius and hence the distance, the above bound also shows that lists at radius $1 - \rho - \eps$ must have size $(1-\rho - \eps)/\eps$. which they conjectured to be optimal.
The generalized Singleton bound also gives a different strengthening of the notion of list decoding, this time with a more fine-grained tradeoff for every list size.

Shangguan and Tamo~\cite{ST20}  also give an explicit construction of a Reed-Solomon code achieving the generalized Singleton bound for $k=3$, with alphabet size $q = 2^{2^n}$. 
Recent works on random puncturings of Reed-Solomon and related codes show that these in fact achieve the \emph{average-distance version} of (a relaxed) generalized Singleton bound, satisfying for any $g \in \Sigma^n$ and $h_1, \ldots, h_{k} \in \calC$,
\[
\Ex{i \in [k]}{\Delta(g,h_i)} ~\geq~ \frac{k-1}{k} \cdot (1 - \rho - \eps) \mper
\]
In a sequence of works, the alphabet size requirement for the above random constructions has been brought down from exponential in $n$~\cite{BGM23} to polynomial~\cite{GZ23} to linear in $n$~\cite{AGL24}.
It was also shown by \cite{BDG24} that exponential alphabet size is \emph{necessary} for the exact version of the Singleton bound, and thus, the relaxation of the bound by $\eps$ is required for smaller alphabet sizes.
It was also proved by Alrabiah, Guruswami, and Li~\cite{AGL24} that random linear codes with alphabet size $2^{1/\eps^{O(1)}}$ can achieve the average-radius relaxed generalized Singleton bound for each fixed $k$, and that a dependence of $2^{O_k(1/\eps)}$~\cite{AGL24:alphabet} is necessary even for the relaxed version.

Very recently, Chen and Zhang~\cite{CZ24} proved that folded Reed-Solomon codes satisfy the average-radius (relaxed) generalized Singleton bound.
Their construction is explicit and algorithmic, and achieves the $\eps$-relaxed bound with alphabet size $n^{O(k/\eps)}$.

\subsection{The Alon-Edmonds-Luby (AEL) construction and our results}
%
% We show that the expander based distance amplification procedure of Alon, Edmonds and Luby~\cite{AEL95} also gives an elementary construction of linear codes satisfying the average-radius ($\eps$-relaxed) generalized Singleton bound, for every fixed list size $k$, with alphabet $2^{(k^k/\eps)^{O(1)}}$.
%
We give an explicit construction of codes satisfying the average-radius ($\eps$-relaxed)
generalized Singleton bound, for every fixed list size $k-1$, with a constant alphabet size
depending only on $\eps$ and $k$ (we denote the list size by $k-1$ since it makes it easier to state the average-distance inequality for $k$ points).
%
%\snote{Maybe this is a good place to stress that we actually give a procedure to start from any high rate code and do this. And then, when instantiated appropriately with LDPC codes,} 
%\mnote{Did now in the abstract and remarks. Let me know if you want to also stress it here.}
%
Our codes also satisfy the low-density parity check (LDPC) property \ie they can described as kernels of sparse (parity-check) matrices with constant sparsity in each row. While random LDPC codes are known to achieve list decoding capacity~\cite{MosheiffRRSW19}, no explicit constructions of such codes were previously known.
\begin{theorem}[Informal version of  \cref{cor:ael_instantiation}]\label{thm:result-intro}
For every $\rho, \eps \in (0,1)$ and $k \in \N$, there exist an \emph{explicit} family of codes $\mathcal{C} \subseteq \Sigma^n$, such that $\rho(\calC) \geq \rho$ and
for any $g \in  \Sigma^n$ and any $\calH \subseteq \mathcal{C}$ with $\abs{\calH} \leq k$ 
\[
\Ex{h \in \calH}{\Delta(g,h)} ~\geq~ \frac{\abs{\calH}-1}{\abs{\calH}} \cdot (1 - \rho - \eps) \mper
\]
Moreover,  the code $\mathcal{C}$ has alphabet size $2^{(k^k/\eps)^{O(1)}}$ and  is characterized by parity checks of size $O((k^k/\eps)^{O(1)})$.
%
% the alphabet size $\abs{\Sigma}$ of the code $\mathcal{C}$ can be taken to be $2^{(k^k/\eps)^{O(1)}}$ and $\mathcal{C}$ is characterized by parity checks of size at most $O((k^k/\eps)^{O(1)})$.
%
\end{theorem} 
By choosing $k = O(1/\eps)$ in  the above theorem, we also get codes which are average-radius list
decodable up to capacity, with the list size at error radius $1 - \rho - \eps$ bounded by $O(1/\eps)$.
\begin{corollary}
For every $\rho, \eps \in (0,1)$, there exist an \emph{explicit} family of codes $\mathcal{C} \subseteq \Sigma^n$, such that $\rho(\calC) \geq \rho$ and
for any $g \in  \Sigma^n$, the list $\calL(g,1-\rho-\eps)$ at radius $1-\rho - \eps$ satisfies
\[
\abs{\calL(g,1-\rho-\eps)} ~=~ \abs{\inbraces{h \in \cC ~\mid~ \Delta(g,h) \leq 1 - \rho - \eps}} ~=~
O(1/\eps) \mper
\]
Moreover,  the code $\mathcal{C}$ has alphabet size  $2^{2^{(1/\eps)^{O(1)}}}$ and  is characterized by parity checks of size $2^{(1/\eps)^{O(1)}}$.
%
%
% Moreover, the alphabet size $\abs{\Sigma}$ of the code $\mathcal{C}$ can be taken to be $2^{2^{(1/\eps)^{O(1)}}}$ and $\mathcal{C}$ is characterized by parity checks of size at most $2^{(1/\eps)^{O(1)}}$.
\end{corollary}

Our construction is elementary in nature, and relies on a celebrated distance amplification technique of Alon, Edmonds, and Luby~\cite{AEL95} (AEL) based on expander graphs.
We derive the above theorem as a consequence of a more general result about a ``local-to-global'' amplification for average-radius list decodability via the AEL construction.
We note that in contrast to previous algebraic constructions of codes achieving list decoding
capacity, our arguments do not need to rely on polynomial interpolation or other algebraic
phenomena, and can be proved using simple combinatorial and spectral properties of expander graphs.

Before describing our results in detail, we recall how the AEL construction can be used to produce an infinite family of codes arbitrarily close to the Singleton bound for distance.
\vspace{-5 pt}
\paragraph{The AEL construction.}
Given a balanced $d$-regular bipartitite graph $G=(L,R,E)$ with $\abs{L} = \abs{R} = n$, we call it an $(n,d,\lambda)$-expander if the second singular value of the (normalized) biadjacency matrix is at most $\lambda$. 
\begin{figure}[htb]
\begin{center}	
\begin{tikzpicture}[scale = 0.7]
\begin{scope}
\draw[] (0,0) ellipse (1.5cm and 3cm);
\node[below] at (0,-3.1) {$L$};
\draw[] (6,0) ellipse (1.5cm and 3cm);
\node[below] at (6,-3.1) {$R$};

%\rotatebox{90}{$\,=$}\cC_\inn \ni \varphi(f^*(\ell))=  

\node[fill,circle,red] (a1) at (0,2) {};
\node[fill,circle,red] (a2) at (0,1) {};
\node[left,darkred] at (-0.2,1) {$\cC_\inn \ni \parens[\big]{f(e_1),f(e_1), f(e_3)} = f_{\ell}  $};
%\node[left] at (-1.9,0.55) {$\Large \rotatebox{90}{\,=} $};
%\node[left] at (-1.2,0.2) {$\cC_\inn \ni \varphi(f^*(\ell))  $};
\node[fill,circle,red] (a3) at (0,0) {};
\node[fill,circle,red] (a4) at (0,-1) {};
\node[fill,circle,red] (a5) at (0,-2) {};
%label={$(f_{\ell}(e_2),f_{\ell'}(e_3), f_{\ell''}(e_2))$}
\node[fill,circle,blue] (b1) at (6,2) {};
\node[fill,circle,blue] (b2) at (6,1) {};
\node[fill,circle,blue] (b3) at (6,0) {};
\node[below,darkblue] at (10,0.45) {{$f_r = \parens[\Big]{f(e_2),f(e_4), f(e_5)} \in \Sigma_\inn^d$}};
\node[fill,circle,blue] (b4) at (6,-1) {};
\node[fill,circle,blue] (b5) at (6,-2) {};

\draw[](a2)--node[pos=0.45,sloped, above] {\small{$f(e_1)$}}(b1);
\draw[](a2)--node[pos=0.45,sloped, above] {\small $f(e_2)$}(b3);
\draw[](a2)--node[pos=0.45,sloped, above] {\small $f(e_3)$}(b5);
\draw[] (a4)--node[pos=0.17,sloped, above] {\small{$f(e_4)$}}(b3);
\draw[] (a5)--node[pos=0.17,sloped, above] {\small{$f(e_5)$}}(b3);

%\node[below] at (2.8,-3.8) {Illustration of the AEL procedure};
\end{scope}
\end{tikzpicture}
\vspace{-10 pt}
\end{center}
\caption{Illustration of the AEL procedure}
\label{fig:ael}
\end{figure}
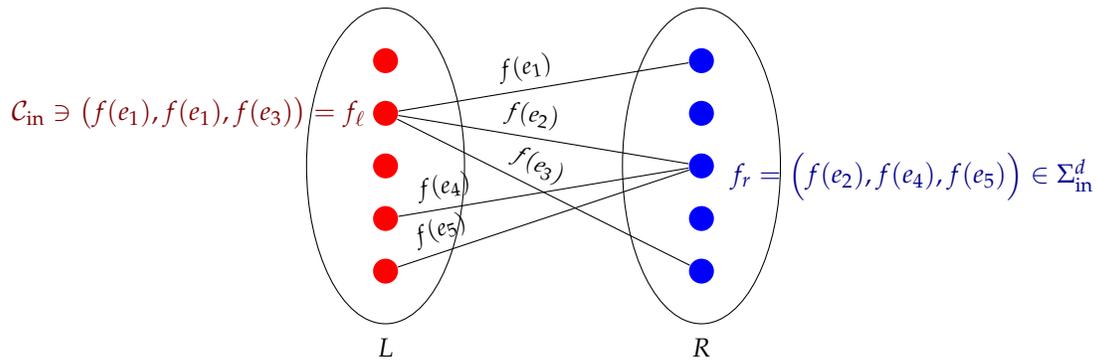

The AEL construction works with such an expanding graph, an ``inner code'' $\calC_{\inn} \subseteq \Sigma_{\inn}^d$ and an ``outer code'' $\calC_{\out} \subseteq \Sigma_{\out}^n$, with $\Sigma_{\out} \subseteq \Sigma_{\inn}^d$. 
We view a codeword of $\calC_{\out}$ as being written on the left vertices (formally, $\calC_{\out} \subseteq \Sigma_{\out}^L$. The symbol at each vertex $\ell \in L$ is then viewed as an element of $\Sigma_{\inn}^d$, which gives a labeling of the edges (with an ordering fixed in advance) and thus an element of $\Sigma_{\inn}^E$. 
Each right vertex $r \in R$ the
n collects the $d$ symbols from the incident edges, forming a symbol in $\Sigma_{\inn}^d$.
The procedure can be viewed as using $\calC_{\inn}$ to map $\calC_{\out}$ to a new code $\AELC \subseteq (\Sigma_{\inn}^d)^R \cong (\Sigma_{\inn}^d)^n$.
As illustrated in \cref{fig:ael}, we can think of symbols being on the edges of $G$, which form a codeword in $\AELC$ when grouped from the right. We will need to consider both the fraction of left and right vertices in which two codewords differ, and will denote these distances using $\Delta_L(\cdot, \cdot)$ and $\Delta_R(\cdot , \cdot)$ respectively.

An easy computation shows that $\rho(\AELC) \geq \rho(\calC_{\out}) \cdot \rho(\calC_{\inn})$. Moreover, an application of the expander mixing lemma can be used to show that 
\[
\delta(\AELC) ~\geq~ \delta(\calC_{\inn}) - \frac{\lambda}{\delta(\calC_{\out})} \mper
\]
Thus, when $\rho(\calC_{\out}) \geq 1 - \eps$ and $\lambda \leq \eps \cdot \delta(\calC_{\out})$, this gives a ``local-to-global'' amplification of the rate-distance tradeoffs obtained by the code $\calC_{\inn}$:  we get $\rho(\AELC) \geq \rho(\calC_{\inn}) - \eps$ and $\delta(\AELC) \geq \delta(\calC_{\inn}) - \eps$. 
In particular, instantiating the construction with a ``constant-sized'' code $\calC_{\inn}$ achieving the Singleton bound, yields an infinite family of codes (corresponding to a family of expander graphs) which achieve a tradeoff $\eps$-close to the Singleton bound. 

In addition to giving a local-to-global amplification, the AEL procedure is also highly versatile and has been adapted for a number of applications, particularly because it also preserves structural properties of the outer code $\calC_{\out}$. 
It is easy to see $\calC_{\out}$ and $\calC_{\inn}$ are both linear over $\F_q$, then so is the code $\AELC$. 
It also preserves the property of being LDPC, being locally testable or locally correctable~\cite{GKORZS17, KMRZS16}, linear time unique decodability \cite{GI05}, and also the duality properties needed for quantum codes~\cite{BGG24}. 

Hemenway and Wootters~\cite{HW15} also used the AEL procedure to construct linear-time decodable codes achieving list decoding capacity for \textit{erasures}. We refer the reader to the excellent discussion in \cite{KMRZS16} on the various applications of this procedure.
\vspace{-5 pt}
\paragraph{Local-to-global results for generalized Singleton bound.}
In the subsequent discussion, we will only consider relaxed versions of the generalized Singleton bound (and will not mention the qualifier explicitly), since it is necessary for having alphabet sizes to be subexponential in $n$.
Similar to the rate-distance tradeoffs discussed above, we show that the AEL procedure also gives a ``local-to-global'' amplification for (a slight strengthening of) the property of satisfying the average-radius generalized Singleton bound. 

For a code $\calC \subseteq \Sigma^n$ with distance $\delta$, and $S \subseteq n$ with $\abs{S}/n = s$, considered the ``erased'' code $\calC_S$, where all coordinates in $S$ are replaced by a special symbol $\bot$ and thus do not contribute to distances. The relative distance of the resulting code (which still has $n$ coordinates) is at least $\delta - s$. 
If the erased code $\calC_S$ is list decodable with list size $k$ at error radius $\beta = (k/(k+1)) \cdot (\delta - s - \eps)$, we would have for any $g$, and $h_1, \ldots, h_{k+1} \in \calC_S$
\[
\sum_{i \in [k+1]} \Delta(g,h_i) ~\geq~ k \cdot (\delta - s - \eps) \mper
\]
We say that the code $(\delta, k_0, \eps)$ average-radius list decodable \emph{with erasures} if the above inequality is true for \emph{all erased codes} $\calC_S$ and all $k < k_0$. 
Note that this property can equivalently be stated by simply erasing symbols in coordinates from $S$ in the center $g \in \Sigma^n$ to form the erased word $\erase{g} \in (\Sigma \cup \{\bot\})^n$ and not counting the $\bot$ symbols towards distances from codewords in $\calC$ (which is how we define it in \cref{def:gen_singleton}). 
We show that the AEL procedure gives a local-to-global amplification for the above property.
\begin{theorem}[Restatement of \cref{thm:main_technical_avg}]\label{thm:main-intro}
Let $k\geq 1$ be an integer and let $\eps > 0$. Let $\AELC$ be a
code obtained using the AEL construction using  $(G, \calC_{\out}, \calC_{\inn})$, where $\calC_{\inn}$ is $(\delta_0, k, \eps/2)$ average-radius list decodable with erasures, and $G$ is a $(n,d,\lambda)$-expander for $\lambda \leq \frac{\delta_{\out}}{6{k}^{k}} \cdot \eps$. 
Then, $\AELC$ is $(\delta_{0}, k,\eps)$ average-radius list decodable with erasures.
\end{theorem}
A reader may notice that instead of the erased codes $\calC_S$, we could have equivalently considered the punctured codes obtained by \emph{restricting} to coordinates in $\bar{S}$, which would yield the same inequality up to normalization factors of $1-s$. 
This view can be used to show that several known results for random codes and random linear codes also yield average radius list-decodability with erasures (via a simple union bound over puncturings). 
This proves the existence of the inner codes $\calC_{\inn}$, satisfying the necessary conditions of \cref{thm:main-intro} with $\delta_0 = 1 - \rho$, which can then be used to construct codes $\AELC$ satisfying the generalized Singleton bound. 
Moreover, if one wants even the inner codes to be fully explicit, they can also be chosen to be folded Reed-Solomon codes, using the results of Chen and Zhang~\cite{CZ24}.
Instantiating \cref{thm:main-intro} with any of the above choices for $\calC_{\inn}$, and appropriate families of graphs $G$ and codes $\calC_{\out}$, immediately yields the family of codes promised in \cref{thm:result-intro}.
%
%
% \mnote{will remove the following}
% %
% \begin{corollary}[Informal version of \cref{cor:ael_instantiation}]
% %
% For every $\rho, \eps \in (0,1)$ and $k \in \N$, there exist explicit inner codes $\calC_{\inn}$ and an infinite family of explicit codes $\AELC \subseteq (\F_q^d)^n$ obtained via the AEL construction satisfying $\rho(\AELC) \geq \rho$, and 
% for any $g \in  (\F_q^d)^n$ and any $\calH \subseteq \AELC$ with $\abs{\calH} \leq k$ 
% \[
% \sum_{h \in \calH} \Delta(g,h) ~\geq~ (\abs{\calH}-1) \cdot (1 - \rho - \eps) \mper
% \]
% Moreover, the alphabet size $q^d$ of the code $\AELC$ can be taken to be $2^{O(k^{3k}/\eps^7)}$.
% %
% \end{corollary}

%
\vspace{-10 pt}
\paragraph{Remarks.}
Given the versatility of the AEL framework, we can also obtain additional structural properties for the codes resulting from our constructions.
\begin{itemize}
\item[-] As discussed earlier, if both $\calC_{\inn}$ and $\calC_{\out}$ are $\F_q$-linear, then so
  is the resulting code $\AELC$. Note that the alphabet size for $\AELC$ is $q^d$ and thus, we
  obtain linearity only over a \emph{subfield}, as is also the case for other capacity-achieving
  codes such as folded Reed-Solomon.
\item[-] Our construction only needs the outer code $\calC_{\out}$ to have high rate and (arbitrarily small) constant distance. In particular, we do not need the outer code to be list decodable for our structural or algorithmic results. Moreover, starting from outer codes with additional properties (such as being LDPC), the AEL construction can often inherit these properties.
\item[-] If the code $\calC_{\out}$ is LDPC with parity checks of size at most $w$, then $\AELC$ is LDPC with parity checks of size at most $w \cdot d$ \emph{over the subfield} $\F_q$. 
Instantiating with an LDPC code $\calC_{\out}$, we thus also get an explicit construction of LDPC codes achieving list decoding capacity, with parity checks of size $2^{(1/\eps)^{O(1)}}$, where $\eps$ is the (arbitrarily small) gap to capacity. Examples of such codes were only known via random constructions~\cite{MosheiffRRSW19}.
\item[-] We note that while we state our results in terms of the average-radius generalized Singleton bound, the AEL construction can also be used to obtain local-to-global results directly for list decoding capacity~\cite{Sri24:thesis}, relying only on a bound on the list size for the inner code.
However, the statement in terms of the average-radius bounds leads to a simpler proof and stronger guarantees, and also directly leads to list decoding algorithms via the ``proofs to algorithms''
paradigm~\cite{FKP19}.
\end{itemize}
\vspace{-10 pt}
\paragraph{Algorithms.}
The proofs for our theorems rely on simple spectral inequalities for expander graphs, which can then
be turned into decoding algorithms, using the Sum-of-Squares (SoS) hierarchy of semidefinite
programs.  
Our algorithms rely on the ``proofs to algorithms'' framework~\cite{FKP19} adapted to the setting of codes~\cite{JST23}.
This framework was used in \cite{JST23} to obtain list decoding algorithms for AEL codes, and in
this work we show how this can be extended to obtain list decoding algorithms matching the
generalized Singleton bound. 
%
%\fnote{Perhaps, it is worth highlighting more this point somewhere. From a combinatorial code construction achieving certain parameters, it is not always immediate that efficient decoding can be done up to them.}
%
\begin{theorem}[Informal version of  \cref{cor:algo-main}]\label{thm:algo-intro}
For every $\rho, \eps \in (0,1)$ and $k \in \N$, there exist an \emph{explicit} family of codes $\mathcal{C} \subseteq \Sigma^n$, such that $\rho(\calC) \geq \rho$ and
for any $g \in  \Sigma^n$ and any $\calH \subseteq \mathcal{C}$ with $\abs{\calH} \leq k$ 
\[
\Ex{h \in \calH}{\Delta(g,h)} ~\geq~ \frac{\abs{\calH}-1}{\abs{\calH}} \cdot (1 - \rho - \eps) \mper
\]
The code  $\mathcal{C}$ has alphabet size  $2^{(k^k/\eps)^{O(1)}}$, and can be list decoded from
radius $\frac{k-1}{k}(1-\rho-\eps)$ deterministically in time $n^{2^{(k^k/\eps)^{O(1)}}}$ with a
list of size at most $k-1$.
\end{theorem} 
Broadly speaking, the SoS framework requires one to prove properties of an object by writing the
object as a set of formal variables, and expressing any desired inequalities as sums-of-squares of
low-degree polynomials in these formal variables.
The properties of the SoS hierarchy then allow one to write a convex relaxation to search for such an object,
and to deduce these properties also hold for ``convex proxy'' for the object obtained as a solution
to the relaxation.

We obtain the above theorem, by reasoning about a convex proxy for (real embeddings of) the entire list $\calH =
\inbraces{h_1, \ldots, h_k}$, replacing codewords by vectors of formal variables $\inbraces{\zee_1,
  \ldots, \zee_k}$.
The spectral expansion of a graph $G$ can be expressed in the form $A_G - u_0
u_0^{\T} \preceq \lambda \cdot I$, where $A_G$ is the normalized adjacency matrix and $u_0$ is the
(known) top eigenvector.
A spectral inequality relying on expansion can often be reduced to understanding a quadratic form
$\ip{\zee}{(\lambda \cdot I - A_G + u_0 u_0^{\T}) \zee}$ which can be viewed as a sum-of-squares
expression in the variables $\zee$ since the quadratic form $\ip{\zee}{M\zee}$ of a positive
semidefinite matrix $M = \sum_{\sigma_i} v_i v_i^{\T}$ is $\sum_{i} \sigma_i \ip{v_i}{\zee}^2$.

The key contribution of our work is the ``proof'' part of this
proofs-to-algorithms implementation.
Adapting the framework from \cite{JST23} to our setting does require some nontrivial technical ideas (discussed in the proof overview
below) but a reader interested in understanding the code construction and properties can safely skip the algorithmic aspects of our result (everything after \cref{sec:inner-code}) on a first reading.

\subsection{An overview of proofs and techniques}
We now give a brief overview of the proof of \cref{thm:main-intro}. 
Consider $h \in \AELC$ and $g \in (\Sigma_{\inn}^d)^n$, where we measure the distance
$\Delta_R(g,h)$ as the fraction of right vertices $r \in R$ on which $g$ and $h$ differ.
Since the goal is to deduce inequalities about ``global'' distances of the form $\Delta_R(g,h)$
using ``local'' inequalities for the inner code $\calC_{\inn}$, we also consider the local views
$g$ and $h$ from each vertex $\ell \in L$.
Formally, we can also view $g, h \in \Sigma_{\inn}^E$ as labeling the edges, which defines for each
$\ell \in L$ local views $g_{\ell}, h_{\ell} \in \Sigma_{\inn}^d$.  For $g_{\ell}$ and $h_{\ell}$ we meaure local
distances $\Delta(g_{\ell}, h_{\ell})$ as a fraction of $d$ edges incident on $\ell$.

Recall that the local distance inequalities show that for any $\ell \in L$, $g_{\ell} \in
\Sigma_{\inn}^d$ and distinct $h_{1,\ell}, \ldots, h_{k,\ell}$
\[
\sum_{i \in [k]} \Delta(g_{\ell}, h_{i,\ell}) ~\geq~ (k-1) \cdot (\delta_0 - \eps) \mcom
\]
and our goal is to deduce a similar inequality for the global distances (possibly with $\eps'$ instead of
$\eps$).
Also, note that by taking
$k=2$ and $g_{\ell} = h_{1,\ell}$, we also get that $\Delta(h_{1,\ell},h_{2,\ell}) \geq \delta_0 -
\eps$, and we can treat $\delta_0$ essentially as the distance of $\calC_{\inn}$, and by the AEL
bound, also for the code $\AELC$.
In the argument below, we will treat $\eps$ as an arbitrarily small error.

\vspace{-10 pt}
\paragraph{Sampling bound.}
A first observation, which is basically the distance proof of AEL,  is that the local distances
\emph{for most} $\ell \in L$ provide a lower bound on the global distances. 
Indeed, let $\Delta(g_{\ell}, h_{\ell}) \geq \theta$ for a set $L' \subseteq L$ of size (say) $\eps
\cdot n$.
Each $\ell \in L'$ thus has $\theta \cdot d$ ``error edges'' incident to it, on which $g_{\ell}$ and
$h_{\ell}$ don't match.
Let $R' \subseteq R$ be the right endpoints of these error edges and note that each $r \in R'$ must be a vertex
where $g, h$ differ.
The expander mixing lemma (\cref{lem:eml}) then shows these that large (greater than $\theta$)
local distances even at $\eps$ fraction of $\ell \in L$ imply large global distance
\[
\abs{L'} \cdot (\theta \cdot d) ~\leq~ \abs{E(L',R')} ~\leq~ \frac{d}{n} \cdot \abs{L'} \cdot
\abs{R'} + \lambda \cdot n \cdot d
\quad \implies \quad
\Delta_R(g,h) ~\geq~ \frac{\abs{R'}}{n} ~\geq~ \theta - \lambda \cdot \frac{n}{\abs{L'}} \mper
\]
By contrapositive, for (say) $\lambda \leq \eps^2$, we have $\Pr{\ell \in L}{\Delta(g_{\ell}, h_{\ell}) \leq
  \Delta_R(g,h) + \eps} \geq 1 - \eps$.
For a given $g \in (\Sigma_{\inn}^d)^R$ and codewords $h_1, \ldots, h_k \in \AELC$ and, we can say
using a union bound that
\[
\Pr{\ell \in L}{\forall i \in [k]~~ \Delta(g_{\ell}, h_{i,\ell}) \leq \Delta_R(g,h_i) + \eps} ~\geq~ 1 -
k\eps \mper
\]
\vspace{-20 pt}
\paragraph{A wrong proof.}
Applying the sampling bound for any ``good'' $\ell$ where all local distances $\Delta(g_{\ell},h_{i,\ell})$
lower bound the global distances $\Delta_R(g,h_i)$, one could guess that we are
already done, since $g_{\ell} \in \Sigma_{\inn}^d$ and the local codewords $h_{1,\ell}, \ldots,
h_{k,\ell}$ satisfy
\[
\sum_{i \in [k]} \Delta(g_{\ell}, h_{i,\ell}) ~\geq~ (k - 1) \cdot (\delta_0 - \eps) 
~\implies~
\sum_{i \in [k]} \Delta_R(g,h_i) ~\geq~ (k-1) \cdot (\delta_0 - \eps) - k \cdot \eps
\]
Of course the ``bug'' in the above proof is that the local projections $h_{1, \ell}, \ldots,
h_{k,\ell}$ are not necessarily distinct even if the codewords $h_1, \ldots, h_{\ell}$ are, and thus we
cannot apply a local distance inequality for $k$ \emph{distinct} codewords.
However, we can get around this issue by carefully examining the consequences of the local
projections being equal for two codewords.

\vspace{-5 pt}
\paragraph{Partitions and erasures.} 
We can think of each left vertex as inducing a partition $\tau_{\ell} =
(\calH_{1}, \ldots, \calH_{p_{\ell}})$ of the set $\calH$ of codewords, where $h, h'$ are in the
same part if and only if $h_{\ell} = h_{\ell}'$.
Since the total number of partitions is at most $k^k$, we fix a
partition $\tau^* = (\calH_{1}, \ldots, \calH_{p})$ which occurs for at least $k^{-k}$ fraction of
$\ell \in L$, forming (say) the set $L^*$, which we still think of as somewhat large if
$\lambda \ll k^{-k}$.

Consider the special case when $\tau^* = (\{h_1\}, \ldots, \{h_{k-2}\}, \{h_{k-1}, h_k\})$ \ie
for all $\ell \in L^*$, it is only the last two codewords that have the same local projection (note
that the projections still depend on $\ell$, just the partition is fixed).
Let $\{f_{1,\ell}, \ldots, f_{k-1, \ell}\}$ denote the set of local projections (for $\ell \in L^*$)
where we have $f_{k-1, \ell} = h_{k-1,\ell} = h_{k,\ell}$ for all $\ell \in L^*$.
The local distance inequality then gives 
\[
\sum_{i \in [k-1]} \Delta(g_{\ell}, f_{i,\ell}) ~\geq~ (k-2) \cdot (\delta_0 - \eps)
~\implies~
\sum_{i \in [k-2]} \Delta_R(g_{\ell},h_{i}) + \Delta(g_{\ell}, f_{k-1, \ell}) ~\geq~ (k-2) \cdot
(\delta_0 - 2\eps)
\]
They key difference is now we no longer think of the term $\Delta(g_{\ell},f_{k-1,\ell})$ as a lower bound
on $\Delta(g,h_{k-1, \ell})$ and $\Delta(g,h_{k,\ell})$, but on the fraction of positions where both
$h_{k-1}$ and $h_k$ \emph{simultaneuosly} differ from $g$.
This can again be seen by applying the expander mixing lemma argument above to the set 
$L' = \{\ell \in L ~|~ \Delta(g_{\ell}, f_{k-1, \ell}) \geq \theta\}$. 
Each $\ell \in L'$ then has $\theta \cdot d$
``simultaneous error edges'',  with the right endpoints being error locations for both $h_{k-1}$ and $h_{k}$.

Now consider the ``erased code'' $\calC_S^{\AEL}$ where we erase the set $S \subseteq R$ of size (say) $s
\cdot n$ of locations where $g_{r} \neq h_{k-1, r}$ and $g_r \neq h_{k,r}$, and don't count the
erased locations towards the distance. 
Since the distance of $\AELC$ is at least $\delta_0 - \eps$, the distance of $\calC_S^{\AEL}$ is at least
$\delta_0 - s - \eps$. Also, $\Delta_S(g,h_{k-1}) = \Delta_R(g,h_{k-1}) - s$ and similarly for
$h_k$, since they both do differ from $g$ in the erased locations. Triangle inequality then gives
\[
\Delta_S(g,h_{k-1}) + \Delta_S(g,h_k) ~\geq~ \delta_0 - s ~\implies~ \Delta_R(g,h_{k-1}) +
\Delta_R(g,h_{k}) ~\geq~ \delta_0 -\eps + s
\]
Adding the two inequalities above, and using the sampling lower bound $\Delta(g_{\ell},f_{k-1,\ell}) \leq s +
\eps$, then proves the average distance inequality for $h_1, \ldots, h_k$.
\begin{figure}[h]
\centering
\includegraphics[width=0.8\textwidth]{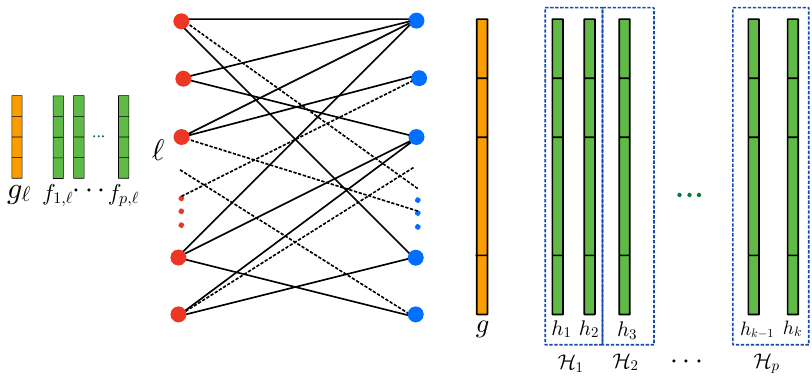}
\vspace{-10 pt}
\caption{Local projections and partitions}
\vspace{-10 pt}
\label{fig:partition}
\end{figure}
\vspace{-5 pt}
\paragraph{Induction.}
The above idea of using erasures for common error locations, can also be extended to arbitrary
partitions $\tau^* = (\calH_1, \ldots, \calH_p)$. 
We now apply the local distance inequalities for codewords $f_{1,\ell}, \ldots, f_{p,\ell}$ corresponding to the different parts, at a ``good'' vertex $\ell \in L^*$ where the sampling bound holds. This gives
\[
\sum_{j \in [p]} \Delta(g_{\ell}, f_{j, \ell}) ~\geq~ (p-1) \cdot (\delta_0 - \eps)
\quad \implies \quad 
\sum_{j \in [p]} s_j ~\geq~ (p-1) \cdot (\delta_0 - \eps) - p \eps \mcom
\]
where $s_j$ now denotes the fraction of simultaneous errors for all $h \in \calH_j$.

For each part $\calH_j$ we also consider an ``erased'' code with $s_j$ fraction of erasures, corresponding to common error locations of all codewords in $\calH_j$. When $\abs{\calH_j} > 2$, the triangle inequality needs to be replaced by an average distance inequality (for global codewords in) the erased code. 
To arrange this, we prove our results by induction on $k$, which gives that (assuming $\abs{\calH_j} \leq k-1$ for all parts), 
\[
\forall j \in [p] \quad \sum_{h \in \calH_j} \Delta_R(g,h) ~\geq~ \inparen{\abs{\calH_j} - 1} \cdot (\delta_0 - \eps) + s_j \mper
\]
Adding the above inequalities and the lower bound on $\sum_j s_j$ then yields the result, since $\sum_j \abs{\calH_j} = k$. (Actually, we need a slight strenghtening of the inequalities above to correctly carry out the induction step, for which the details can be found in \cref{sec:avg-singleton}.)

% which means that such inequalities are
% available for all $\abs{\calH_j} \leq k-1$.

Finally, we also need to rule out the trivial partition with only one part of size $k$, since we needed $\abs{\calH_j} \leq k-1$ for all $j \in [p]$ to use induction. 
However, this cannot happen as this would mean that $h_{1,\ell} = \cdots = h_{k,\ell}$ for (essentially) all $\ell$. But $h_1, \ldots, h_k$ are distinct codewords in $\AELC$ and thus their
projections must differ in at least $\delta_{\out}$ fraction of left vertices.
\vspace{-5 pt}
\paragraph{Discussion.}
As mentioned earlier, our proof techniques rely on elementary combinatorial and spectral arguments,
instead of polynomial interpolation, which has been the workhorse for list size bounds in all
previous explicit constructions achieving list decoding capacity. 
As compared to previous applications of AEL for list decoding, the key difference here is a
fine-grained analysis of the structure of the local lists. 
While many previous analyses treated the local lists as a black box and relied on stronger (list
recovery) properties of the outer code for algorithmic applications, the idea of keeping track of
their structure, such as in terms of common error locations, may also be useful for other applications.

\vspace{-5 pt}
\paragraph{Connections to interleaved codes.}
The idea of using error locations of some parts of a code as erasures for the remaining code was already used in the work of Gopalan, Guruswami, and Raghavendra \cite{GGR09}. However, there the focus was to show that when starting with an already good list decodable code, its list size does not degrade too much upon interleaving. The interleaving operation can be seen as the dense analog of AEL, but since it does not change the blocklength, the results in \cite{GGR09} did not give an infinite family of codes starting from a single (inner) code as we do.

In more detail, an order-$t$ interleaving of a code $\calC$ of blocklength $n$ can be seen as applying the AEL procedure with $\calC$ as the inner code and $\calC^t$ as the trivial outer code, on the complete bipartite graph $K(t,n)$. By replacing all uses of the expander mixing lemma in the proof by $E(S,T) = |S|\cdot |T|$ (or equivalently, $\lambda=0$), one can verify that if the code $\calC$ achieves generalized Singleton bound with erasures, so does its order-$t$ interleaving. This also improves upon the list size that one could obtain via a black-box application of \cite{GGR09} on a code that lies close to the generalized Singleton bound.

The idea of using error locations as erasures can also be used to prove weak non-trivial bounds on list size of interleaved codes in general, and follows the same structure as the proof of Schwartz-Zippel lemma. In particular, both the list codewords for the interleaving of \emph{arbitrary} codes and the roots of a multivariate polynomial are contained in sets with a common structure, which immediately gives list size upper bounds. Such results were known before for the interleaving of Reed-Solomon codes, which can be derived using multivariate interpolation combined with the Schwartz-Zippel lemma \cite{CS03, PV05, GX13}. Our argument using erasures is however applicable to general codes and uses only on their distance, which is surprising because for non-algebraic codes the analog of multivariate interpolation is not clear at all. We believe this combinatorial handle on interleaved codes sheds light on interpolation-based decoders, and may find further applications. The interested reader is referred to \cite{Sri24:thesis} for details.

\vspace{-5 pt}
\paragraph{Outline of our algorithm.}
As discussed earlier, our algorithm for list decoding a given $g \in \Sigma^n$ searches for ``convex
proxies'' given by the SoS hierarchy, for the entire list of codewords in a ball around $g$.
We consider vectors $\zee_1, \ldots, \zee_k$ corresponding to (real embeddings of) the codewords
and the solutions to the SoS relaxation can be described as objects $\tildeEx{\zee_1}, \ldots,
\tildeEx{\zee_k}$ containing values corresponding to all low-degree monomials in the variables,
which we refer to as ``pseudocodewords''.
One can also extend $\tildeEx{\cdot}$ by linearity to form a linear operator which yields values for
all polynomials. The constraints of the SoS hierarchy imply that $\tildeEx{P(\zee)^2} \geq 0$ for all
low-degree polynomials $P$, and thus all inequalities that can be expressed as sums-of-squares of
low-degree polynomials are satisfied by the solution.

Our algorithm simply searches for the largest $k'$ for which there exists a tuple $(\tildeEx{\zee_1}, \ldots,
\tildeEx{\zee_{k'}})$ of pseudocodewords satisfying two conditions:
\begin{enumerate}
\item For all $i \neq j$, $\zee_i$ and $\zee_j$ are far (just as codewords are):
  $\tildeEx{\Delta_L(\zee_i, \zee_j)} \geq \delta'$.
\item For all $i$, $\zee_i$ is within the required error radius of $g$ (just as list elements are):
  $\tildeEx{\Delta_R(g,\zee_i)} \leq \beta$.
\end{enumerate}
Note that we use above that the distance $\Delta_L$ and $\Delta_R$ can be expressed as low-degree
polynomials in the variables $\inbraces{\zee_i}_{i \in [k']}$ (which is easy diffcult to show).
The analysis for the algorithm now consists of two parts. 

We first show that for a maximum $k'$ where the convex relaxation for $(\tildeEx{\zee_1}, \ldots,
\tildeEx{\zee_{k'}})$ is feasible, but the one for $k'+1$ codewords is not, any \emph{true} codeword
$h \in \AELC$ with $\Delta_R(g,h) \leq \beta$ must (essentially) satisfy $\tildeEx{\Delta_L(\zee_i, h)}
\leq \delta'$ for some $i \in [k]$. 
This is because if $h$ is far from all $\zee_i$s, we can augment the solution to a $k'+1$ tuple by
adding (low-degree monomials in) $h$ as $\tildeEx{\zee_{k'+1}}$. 
Since SoS is a relaxation, true codewords are also valid pseudocodewords, and thus $k'$ cannot
be maximum.
Given this ``covering'' property for the maximum $k'$ and by choosing $\delta'$ small enough, we can
obtain any $h$ in the list by ``unique-decoding'' the corresponding codeword of $\calC_{\out}$ from
the vector $\tildeEx{\zee_i}$ for some $i \in [k]$.
This argument is presented in \cref{sec:sos_algo}.

The second part of the proof consists in showing that \emph{there exists a maximum} $k'$. 
This follows from extending the average distance inequality to pseudocodewords, using the fact that
the proof can be expressed as sums-of-squares, and is proved in \cref{sec:sos_proof}. 
Formally, we get that for any $k'$-tuple of pseudocodewords as above, we must have
\[
\sum_{i \in [k']} \tildeEx{\Delta_R(g, \zee_i)} ~\geq~ (k' - 1) \cdot (\delta_0 - \eps) \mper
\]
However, for $\beta \leq ((k-1)/k) \cdot (\delta_0 - \eps)$, this contradicts condition (2) above
unless $k' \leq k$. 
Using codes with $\delta_0 \geq 1 - \rho - \eps$ and taking $k = O(1/\eps)$, we can use the above
algorithm to decode the code $\AELC$ arbitrarily close to list decoding capacity.

\section{Preliminaries}\label{sec:prelims}

\subsection{Expanders and Codes}

%\begin{definition}[$(\beta,\gamma)$-expander]
%	Let $\beta,\gamma\in (0,1)$. A $d$-regular bipartite graph $G =(L,R,E)$ with $|L|=|R|=n$ is a $(\beta,\gamma)$-expander if the following is true for every $\alpha \in (\gamma,1)$ and for every $S\sub L$ with $|S|\geq \beta\cdot n$: if for every vertex in $S$, $\alpha d$ edges among its neighborhood are colored red, then at least $(\alpha-\gamma)n$ vertices in $R$ have one or more red edges incident on it.
%\end{definition}
%
%\begin{proposition}
%	The complete bipartite graph $K_{n,n}$ is a $(\frac{1}{n},0)$-expander.
%\end{proposition}
%
%\begin{proposition}
%	The bipartite spectral expander with second largest normalized singular value $\lambda$ is a $(\beta,\gamma)$-expander if $\lambda \leq \beta \gamma$.
%\end{proposition}

For a bipartite graph $G=(L,R,E)$, let $L$ be the set of left vertices, and $R$ be the set of right vertices. Let $A_G$ denote the $L\times R$ biadjacency matrix, and $\sigma_2(A_G)$ be its second largest singular value.

\begin{definition}[$(n,d,\lambda)$-expander]
A $d$-regular bipartite graph $G(L,R,E)$ with $|L|=|R|=n$ is said to be an $(n,d,\lambda)$--expander if $\sigma_2(A_G) ~\leq~ \lambda \cdot d$.
\end{definition}

Infinite families of $(n,d,\lambda)$--expanders, with growing $n$ as $d$ and $\lambda$ are constant, can be derived based on double covers of Ramanujan graphs of~\cite{LPS88} as long as $\lambda \geq \frac{2\sqrt{d-1}}{d}$.

\begin{lemma}[Expander Mixing Lemma]\label{lem:eml}
	Given an $(n,d,\lambda)$-expander $G=(L,R,E)$ and functions $f: L \rightarrow \R$ and $g: R \rightarrow \R$, the following well--known property is a simple consequence of definition of $(n,d,\lambda)$--expanders:
	\[ \Big\vert{\Ex{(\li,\ri) \sim E}{ f(\li) \cdot g(\ri)} - \Ex{\li\sim L}{f(\li)} \cdot \Ex{\ri\sim R}{g(\ri)}}
                \Big\vert ~\leq~ \lambda \cdot \norm{f}_2\norm{g}_2 \mper
	\]
	When $f = \mathbb{1}_S,g = \mathbb{1}_T$ are indicators of sets, denote by $E(S,T)$ the number of edges between $S$ and $T$.  Then,
	\[ \Big\vert{ {E(S,T)} -\frac{d\cdot |S||T|}{n}}
                \Big\vert ~\leq~ \lambda \cdot d\cdot  \sqrt{|S||T|} ~\leq~ \lambda \cdot d\cdot n \mper
	\]	
\end{lemma}

\paragraph{Notation and Edge orderings.} 
For a bipartite graph $G=(L,R,E)$, we index the left set by $\li$ and the right set by $\ri$. For a vertex $\li \in L$, 
we denote the set of edges incident to it by $N(\li)$ (left neighborhood), and the set of edges
incident to  $\ri\in R$ is denoted by $N(\ri)$(right neighborhood). 
We use $\li \sim \ri$ to denote that the vertex $\li \in L$ is adjacent to the vertex $\ri \in R$, that is, $(\li,\ri)\in E$. 

Fix an arbitrary ordering of the edges. Then there are bijections between the sets $E$, $L \times
[d]$, and $R \times [d]$, 
given by taking $(\li,i)$ to be the $i^{th}$ edge incident on $\li$, and similarly for $R \times [d]$.
Henceforth, we will implicitly assume such an ordering of the edges is fixed, and use the resulting bijections.

\paragraph{Codes.} Now we define the most basic notions associated with a code. 
% We will only work with linear codes but give the general definition here.
\begin{definition}[Fractional Hamming Distance]
	Let $\Sigma$ be a finite alphabet and let $f,g\in \Sigma^n$. Then the (fractional) distance
        between $f,g$ is defined as \[ \dis(f,g) = \Ex{i\in [n]}{ \indi{f_i \neq g_i}} \mper \]
\end{definition}

\begin{definition}[Code, distance and rate]
	A code $\calC$ of block length $n$, distance $\delta$ and rate $\rho$ over the alphabet size $\Sigma$ is a set $\calC \subseteq \Sigma^n$ such that,
	\[
	\rho = \frac{\log_{|\Sigma|} |\calC|}{n} \;\;\text{and }\; \delta = \min_{h_1,h_2\in \calC \colon h_1 \ne h_2} \dis(h_1,h_2).
	\]	
%	\begin{enumerate}[(i)]
%		\item $\rho = \frac{\log_{|\Sigma|} |\calC|}{n}$
%		\item $\delta = \min_{h_1,h_2\in \calC} \dis(h_1,h_2)$
%	\end{enumerate}
	Such codes are succinctly represented as $[n,\rho,\delta]_\Sigma$.
	We say $\calC$ is a \emph{linear code} if $\Sigma$ can be identified with a finite field $\F_q$, and $\calC$ is a linear subspace of $\F_q^n$.
\end{definition}

% Sometimes, we will view $\calC$ as a bijection $\calC: [|\calC|] \rightarrow [q]^n$, or $\calC: [q^{\rho \cdot n}] \rightarrow [q]^n$.

%\snote{Define list of codewords $\calL(g,\delta)$ here.}
%\tnote{Should we?}

\subsection{Alon-Edmonds-Luby distance amplification}\label{sec:AEL_prelims}

Alon, Bruck, Naor, Naor and Roth \cite{ABNNR92} introduced a graph-based distance amplification scheme which was generalized by Alon, Edmonds and Luby \cite{AEL95}, and used by Guruswami and Indyk \cite{GI05} to design linear-time unique decodable near-MDS codes.

The scheme is a three-step process involving an outer code  ($\calC_\out$), an inner code ($\calC_\inn$), and a bipartite expander $G$: (i) concatenate the outer code $\cC_\out$ with inner code $\cC_\inn$ (ii) shuffle the symbols of concatenated code via edges on a $d$-regular bipartite expander graph $G$, and (iii) collect $d$-symbols on the right vertices and fold them back to produce the final code, $\AELC$. We now formally define this procedure.

Fix an $(n,d,\lambda)$-expander $G = (L,R,E)$. Let $\calC_{\out}$ be an $[n,r_{\out},\delta_{\out}]_{\Sigma_\out}$ code and let $\calC_{\inn}$ be a $[d,r_{\inn},\delta_{\inn}]_{\Sigma_\inn}$ code with $\abs{\Sigma_\out} = |\calC_\inn|$, and a bijection $\phi: \Sigma_\out\rightarrow \calC_\inn\subseteq \Sigma_\inn^d$,  realizing this.  
%\fnote{The convention $[d,r_{\inn}, \delta_{\inn}]_{\Sigma_\inn}$ with dimension/rate before distance seems more common. Not sure if it is worth changing.}\tnote{agreed and done.}

A codeword $f$ of $\AELC$ technically belongs to the space $(\Sigma_\inn^d)^R$ but by our choice of parameters, $(\Sigma_\inn^d)^R$ is in bijection with the set $\Sigma_\inn^E$ and one can just \emph{fold} or \emph{unfold} the symbols to move from one space to the other. Choosing $f$ to be in $\Sigma_\inn^E$ allows us to talk about $f$ viewed from left vertex set $L$ or right vertex set $R$. Formally, if $f_\ell \defeq f\restrict{N(\ell)} \in \cC_\inn$, we have $ \inv{\varphi}(f_\ell) \in \Sigma_\out$. 

%We refer to  which is the view of $f$ from left after converting the symbols.

%For $f: L \to \Sigma_\out$, define $f^*: E\to \Sigma_\inn$ such that $f^*\restrict{N(\ell)} := \phi(f(\ell))$
%
% 
%
%%we have ${\Sigma_\out}^L \cong \cC_\inn^L \cong {\Sigma_\inn}^E \cong ({\Sigma_\inn^d})^R$, 
%
%
%$f \in \calC_{\out} \subseteq {\Sigma_\out}^L$
%define $f_\ell:= f\restrict{N(\ell)}$.  

\begin{definition}[AEL Codes]

Given inner code $\cC_\inn$ an outer code $\cC_\out$, a map $\varphi$, and a graph $G$ as above, the AEL code is defined as,
\[
\AELC ~=~ \braces[\big]{f: R\to\Sigma_\inn^d \mid \forall \ell \in L, f_\ell \in \cC_\inn,\;\text{and } \parens[\big]{\inv{\varphi}(f_\ell)}_\ell \in \cC_\out }.
\]
%	\[
%		f^{\AEL} (\ri) ~=~ \left( f^{*}_{\calC_{0}}(e_1),f^{*}_{\calC_{0}}(e_2),\cdots,f^{*}_{\calC_{0}}(e_d) \right)
%	\]
%	where $e_1,e_2,\cdots,e_d$ are the $d$ edges incident on $\ri$.
%	The AEL code $\calC^{\AEL}_{\calC_{0}}(\calC_{1}) \subseteq [q_0^d]^{n}$ is defined as 
%	\[
%		\calC^{\AEL}_{\calC_{0}}(\calC_{1}) ~=~ \{ f^{\AEL}_{\calC_{0}} \mid f \in \calC_{1}\}
%	\]
%	When clear from context, we will omit $\calC_{0},\calC_1$ in the above notation to call the AEL code $\AELC$.
\end{definition}

One can also think of the code procedurally as, starting with $f^* \in \cC_\out$, and defining $f : E\to \Sigma_\inn$ as the function such that $f_\ell := \varphi(f^*(\ell))$. This function $f$ can then be folded on the right vertices to obtain an AEL codeword.

\paragraph{Distance metrics for AEL Codes.}
Let $f,g \in \AELC$. As we have seen, one can view these as functions on $L$, or on $R$. Associated with each of these viewpoints, we can define the following two distance functions:
\begin{align*}
	\dis_L(f,g) &\defeq \Ex{\li\in L}{ \indi{ f_{\ell} \neq g_{\ri}}} \\
%	\dis(f,g) &\defeq \Ex{e\in E}{\indi{ f_e \neq g_e}} = \Ex{\li \in L}{\dis(f_{N(\li)},g_{N(\li)})} = \Ex{\ri \in R}{\dis(f_{N(\ri)},g_{N(\ri)})}\\
	\dis_R(f,g) &\defeq \Ex{\ri\in R}{ \indi{ f_{\ri} \neq g_{\ri}}} .
\end{align*}

Note that for the AEL code, the distance metric is the right distance $\Delta_R(\cdot, \cdot)$. Alon, Edmonds, and Luby, proved the following result, which shows that the construction can be used to amplify the distance to $\delta_0$ by choosing $\lambda$ sufficiently small.
\begin{theorem}[\cite{AEL95}]\label{thm:ael_distance}
	Let $f,g \in \AELC$ be distinct. Then, $\dis_R(f,g)~\geq~ \delta_{\inn}-
        \frac{\lambda}{\dis_L(f,g)} \geq~ \delta_{\inn} - \frac{\lambda}{\delta_{\out}} $.
%        , which also implies $\Delta(\calC^{\AEL})        ~\geq~ \delta_{\inn} - \frac{\lambda}{\delta_{\out}}$.
\end{theorem}
We wait to give a proof as we will prove a much more general version of this claim in the next section. For the rest of the paper, the alphabet of the AEL code, $\Sigma_\inn^d$ will be denoted simply as $\Sigma$.

\section{Average radius generalized Singleton bound}\label{sec:avg-singleton}
We now prove that AEL codes satisfy the average-radius generalized Singleton bound. 
We will actually prove a more general statement involving a received word where some of the coordinates are erasures, marked with a special symbol $\bot$. We will denote such partially erased received words as $\erase{g}$ instead of $g$ to mark the distinction.
For $\erase{g} \in (\Sigma \cup \{\bot\})^n$ and $f \in \Sigma^n$, we define the distance using only the non-erasure coordinates as
\[
\Delta(\erase{g},h) ~\defeq~ 
\frac{1}{n} \cdot \abs{\inbraces{i \in [n] ~\mid~ \erase{g}_i \in \Sigma ~~\text{and}~~ h_i \neq \erase{g}_i}} \mper
\]
Note that if $s$ denotes the fraction of erasures in $\erase{g} \in (\Sigma \cup \{\bot\})^n$ and $h_1, h_2$ are two codewords from a code with distance $\delta$, then the triangle inequality implies $\Delta(\erase{g},h_1)+\Delta(\erase{g},h_2) \geq \delta - s$. 
The following definition generalizes this to any set of (at most) $k$ distinct codewords.

\begin{definition}[Average-radius list decodable with erasures]\label{def:gen_singleton}
A code $\calC \sub \Sigma^n$ is $(\delta, k,\eps)$ {\deffont average-radius
  list decodable with erasures} if for any $\erase{g} \in (\Sigma \cup \{\bot\})^n$ with (say) $s$ fraction of erasures, and for any set of codewords $\calH \sub \calC$ with $|\calH| \leq k$, it holds that 
\[
\sum_{h \in \calH} \Delta(\erase{g},h) ~\geq~ (\abs{\calH} - 1) \cdot (\delta - s - \eps) \mper
\]
\end{definition}
Note that the above definition also implies a lower bound on the distance of the code $\calC$, since for
any two distinct codewords $h_1, h_2$, we can take $g = h_1$ to get $\Delta(h_1, h_2) \geq \delta - \eps$.
Moreover, a code $\calC$ satisfying the above definition (even with $s=0$) must have the property that an open ball around any $g \in \Sigma^n$ of radius $(\frac{k-1}{k}) \cdot (\delta - \eps)$ contains at most $k-1$ codewords. For $k = 1/\eps$, this yields a list size of $1/\eps$ at radius $\delta - 2\eps$. 
%\tnote{should be open ball here? there can be k codewords with distance exactly $k-1/k(\delta-\eps)$ from $g$}

We show that the stronger property above of being average-radius list decodable \emph{with erasures}
interacts nicely with the AEL construction, which yields a local-to-global result for this property
\ie if the (constant-sized) inner code $\calC_{\inn} \subseteq \Sigma_{\inn}^d$ used in the AEL
construction is $(\delta, k,\eps)$ average-radius list decodable with erasures, then so is the resulting (global) code $\AELC \subseteq (\Sigma_{\inn}^d)^n$.

\begin{theorem}\label{thm:main_technical_avg}
Let $k\geq 1$ be an integer and let $\eps > 0$. Let $\AELC$ be a
code obtained using the AEL construction using  $(G, \calC_{\out}, \calC_{\inn})$, where $\calC_{\inn}$ is $(\delta_0, k_0, \eps/2)$ average-radius list decodable with erasures, and $G$ is a $(n,d,\lambda)$-expander for $\lambda \leq \frac{\delta_{\out}}{6{k_0}^{k_0}} \cdot \eps$. 
Then, $\AELC$ is $(\delta_{0}, k_0,\eps)$ average-radius list decodable with erasures.
\end{theorem}
%
% It follows from \cref{thm:ael_distance} that $\Delta(\AELC) \geq \delta_{\inn} - \lambda/\delta_{\out}
%   \geq \delta_{\inn} - \eps$. We will denote the quantity $\delta_{\inn} - \eps$ simply as $\Delta$ in the rest of the proof.
%
We need to prove that for any collection $\calH = \inbraces{ h_1,\cdots , h_k} \subseteq \AELC$, and any $\erase{g} \in (\Sigma_{\inn}^d \cup \{\bot\})^R$ with fraction of erasures (say) $s$, we must have 
\[
\sum_{h \in \calH} \Delta_R(\erase{g}, h) ~\geq~ (k-1) \cdot (\delta_0 - s - \eps) \mper
\]
We will prove this by induction on the size $k$ of the collection $\calH$. 
Note that the case $k=1$ is trivial since distances are non-negative. 
Before proceeding to the induction step, we first need to understand the ``local views" of codewords $h \in \calH$ from each vertex $\ell \in L$. We refer to them as local projections and develop some inequalities for them below.

%For the inductive step, we assume that the conclusion of the theorem is true for any $\erase{g}$ and any $\calH' \subseteq \AELC$ with $\abs{\calH'} \leq k-1$. \tnote{I feel there is little value in having this para here.}
%

\vspace{-5 pt}
\paragraph{Local projections and induced partitions.}
To use the inequalities for the local codes $\calC_{\inn}$, we will need to consider the ``local projections" of codewords $h \in \calH$ for each vertex $\ell \in L$, which are codewords in $\calC_{\inn} \subseteq \Sigma_{\inn}^d$. 
For a vertex $\ell \in L$, and $h \in \calH$, let $\hl$ be the local codeword in $\calC_{\inn}$ given by the values in $h$ for the edges incident on $\ell$.
We know that the codewords in $\calH$ are pairwise distinct, however, this need not be true for their local projections $\hil{1},\cdots \hil{k}$. 
We say that a left vertex $\ell$ \textit{induces a partition} $\tau_{\ell} = (\calH_{1,\ell}, \ldots, \calH_{p_{\ell},\ell})$ of $\calH$, wherein $h_i, h_j$ are in the same part if and only if $\hil{i} = \hil{j}$. 
Since the number of partitions if bounded by $k^k$, many left vertices must induce the same partition. We will need the additional fact that this must be a non-trivial partition with the number of parts $p \geq 2$.

\begin{claim}\label{lem:type_arg}
For any set of codewords $\calH = \inbraces{ h_1,\cdots , h_k}$, there exists a partition $\tau^*$
of it with at least 2 non-empty parts, and a set $L^*\subseteq L$ such that $\tau_{\ell} = \tau^*$
for all $\ell \in L^*$, and $|L^*| \geq \delta_\out \cdot n / {k^k}$.
\end{claim}
\begin{proof}
Let $L'$ be the set of vertices in $L$ which induces a non-trivial partition, \ie for which not all local projections of the codewords in $\calH$ are identical:
\[
L' ~=~ \{\ell \in L \mid \exists\, i,j \in [k], \; \hil{i} \neq \hil{j} \}.
\]
Since the codewords in $\calH$ are distinct, we know that $|L'| \geq \delta_{\out} \cdot n$. 
The total number of partitions of $\calH$ is at most $k^k$, and thus among the vertices in $L'$, there must be at least $\nfrac{|L'|}{k^k}$ many vertices that induce a common partition of $\calH$, which by definition of $L'$ will be non-trivial.	
\end{proof}

Throughout the proof, we will work with one such fixed partition $\tau^* = (\calH_1, \ldots, \calH_p)$, and the corresponding set $L^*$. Fix an $\ell \in L^*$ and $j\in [p]$. 
By definition, the local codeword $\hl$ is the same for all $h\in \calH_j$. We denote this common codeword by $\fjl$ and the local projection of $\erase{g}$ by $\gl \subseteq (\Sigma_{\inn} \cup \{\bot\})^d$, where if $\erase{g}_r = \bot$ for a vertex $r \in R$, we take projection to be $\bot$ on all edges incident on $r$. 
Let $\sl$ denote the fraction of erasures ($\bot$ symbols) in $\gl$. We need the following inequality for local projections.
\begin{claim}\label{claim:local-bound}
Let $\tau^*$ and $L^*$ be as above, and  $\gl, \sl$ and $\fjl$ be defined as above for each $\ell \in L^*$. Then,
\[
\sum_{j\in [p]} \Delta(\gl, \fjl) ~\geq~ (p-1) \cdot (\delta_{0} - \sl - \eps/2) \mper
\]
\end{claim}
\begin{proof}
The claim follows from the fact that the inner code is $(\delta_0, k_0,\eps/2)$ average-radius list decodable with erasures and $p < k \leq k_0$.
\end{proof}

We will also need the following claim regarding the local erasure fractions $\sl$.
\begin{claim}[Sampling bound for erasures]\label{claim:sampling_erasure}
For the set $L^*$ defined as above with $\abs{L^*} \geq \delta_{\out} \cdot n / k^k$, and local erasure fractions $\sl$,
\[
		\Ex{\ell\in L^*}{\sl} ~\leq~ s + \frac{\eps}{6} \mper
\]
\end{claim}
\begin{proof}
Let $S = \inbraces{r \in R ~|~ \erase{g}_r = \bot}$ be the set of erasure vertices on the right with $\abs{S} = s \cdot n$. 
Then, applying the expander mixing lemma yields
\[
\sum_{\ell \in L^*} \sl \cdot d ~=~ \abs{E(L^*,S)} ~\leq~ \frac{d}{n} \cdot \abs{L^*} \cdot \abs{S} + \lambda \cdot d \cdot n 
\quad\implies\quad
\Ex{\ell\in L^*}{\sl} ~\leq~ s + \lambda \cdot \frac{n}{\abs{L^*}} \mper
\]
Using $\abs{L^*} \geq \delta_{\out} \cdot n / k^k$ and $\lambda \leq (\delta_{\out}/k^k) \cdot (\eps/6)$ then proves the claim. 
\end{proof}

\vspace{-5 pt}
\paragraph{Completing the induction step.} 
Recall that using the induction hypothesis, we can say that for any $\erase{g'}$ (with say $s'$ fraction of erasures) and any $\calH' \subseteq \AELC$ with $\abs{\calH'} \leq k-1$, we must have $\sum_{h \in \calH'} \Delta_R(\erase{g'},h) \geq (\abs{\calH'}-1) \cdot (\Delta - s' - \eps)$.
Since the partition $\tau^* = (\calH_1, \ldots, \calH_p)$ is nontrivial, the cardinality $\abs{\calH_j}$ of each part is at most $k-1$ and we can claim by induction with the given $\erase{g}$ that
\[
\forall j \in [p] \qquad \sum_{h \in \calH_j} \Delta_R(\erase{g}, h) \geq (\abs{\calH_j}-1) \cdot (\Delta - s - \eps) \mper
\]
The following key lemma yields a strengthening of this bound by applying the induction hypothesis with a \emph{different} center $\gj$ for each $\calH_j$.

\begin{lemma}[Inductive bound on distances]\label{lemma:inductive}
Let the partition $\tau^* = (\calH_1, \ldots, \calH_p)$ and the set $L^*$ be as above, and let the local projections $\gl$ and $\fjl$ be also as defined above. If the code $\AELC$ is $(\delta_0,k-1,\eps)$ average-radius list decodable with erasures, then for every $j\in [p]$,
\[			
\sum_{h\in \calH_j} \Delta_R(\erase{g},h) ~\geq~  (|\calH_j|-1) \cdot \inparen{\Delta - s - \eps} + \Ex{\ell \in L^*}{\Delta(\gl, \fjl)} - \frac{\eps}{6} \mper
\]
\end{lemma}
\begin{proof}
By definition of $\fjl$, we have that for all $h \in \calH_j$, $\hl = \fjl$ for all $\ell \in L^*$. Thus, if $\gl$ and $\fjl$ differ on edge $(\ell,r)$ with $\erase{g} \neq \bot$, then $r$ is a \emph{common error location} for all $h \in \calH_j$. We define the set 
\[
S_j ~\defeq~ \inbraces{r \in R ~\mid~ \erase{g}_r \neq \bot ~\text{and}~ \exists \ell \in L^*, e = (\ell,r) ~\text{such that}~ (\gl)_e \neq (\fjl)_e } \mper
\]
Let $s_j = \abs{S_j}/n$ and let $\gj$ be obtained from $\erase{g}$ by replacing symbols in $S_j$ by $\bot$. 
The total fraction of erasures in $\gj$ is $(\abs{S}+\abs{S_j})/n = s+s_j$. Also, $\Delta_R\parens[\big]{\gj,h} = \Delta_R(\erase{g},h) - s_j$ for all $h \in \calH_j$, since all vertices in $S_j$ are known to be error locations which are erased in $\gj$. Applying the inductive hypothesis with $\gj$ now gives
\begin{align*}
&\sum_{h \in \calH_j} (\Delta_R(\erase{g},h) - s_j)
~=~ \sum_{h \in \calH_j} \Delta_R\parens[\big]{\gj,h} 
~\geq~ (\abs{\calH_j}-1) \cdot (\Delta - s - s_j - \eps) \\
\implies~~
&\sum_{h \in \calH_j} \Delta_R(\erase{g},h) ~\geq~ (\abs{\calH_j}-1) \cdot (\Delta - s - \eps) + s_j \mper
\end{align*}
To obtain a bound on $s_j$, we again use expander mixing lemma to deduce
\[
\sum_{\ell \in L^*} \Delta(\gl,\fjl) \cdot d ~=~ \abs{E(L^*, S_j)} ~\leq~ \frac{d}{n} \cdot \abs{L^*} \cdot \abs{S_j} + \lambda \cdot d \cdot n 
\quad \implies \quad
\Ex{\ell \in L^*}{\Delta(\gl, \fjl)} ~\leq~ s_j + \lambda \cdot \frac {n}{\abs{L^*}} 
\mper 
\]
Using $\abs{L^*} \geq \delta_{\out} \cdot n / k^k$ and $\lambda \leq (\delta_{\out}/k^k) \cdot (\eps/6)$ gives the required bound.
\end{proof}

We can now prove the induction step for the set $\calH = \inbraces{h_1, \ldots, h_k}$.
\begin{proof}[Proof of \cref{thm:main_technical_avg}]
The proof, as mentioned earlier, is by induction on $k$.
%Note that the case $k=1$ is trivial since distances are non-negative. For the inductive step, we assume that the conclusion of the theorem is true for any $\erase{g}$ and any $\calH' \subseteq \AELC$ with $\abs{\calH'} \leq k-1$. \tnote{Added the para from before here.}

Let $L^*$ and $\tau^* = (\calH_1, \ldots, \calH_p)$ be as above. We use the induction hypothesis to apply the bound from \cref{lemma:inductive} to each part $\calH_j$ which has size at most $k-1$ as $\tau^*$ is non-trivial. This gives,
\begin{align*}
\sum_{h\in \calH} \Delta_R(\erase{g},h) 
~=~ \sum_{j\in [p]} \sum_{h\in \calH_j} \Delta_R(\erase{g},h) 
&~\geq~ \sum_{j\in [p]} \inparen{(|\calH_j|-1) \cdot \inparen{\Delta - s - \eps} + \Ex{\ell \in L^*}{\Delta(\gl, \fjl)} - \frac{\eps}{6}}\\
&~=~  (k-p) \cdot \inparen{\Delta - s - \eps} + \sum_{j\in [p]} \Ex{\ell \in L^*}{\Delta(\gl, \fjl)} - \frac{p\eps}{6}.
\end{align*}
Using local distance inequality from \cref{claim:local-bound} and the sampling bound from \cref{claim:sampling_erasure}, we can bound the second term as
\[
\sum_{j\in [p]} \Ex{\ell \in L^*}{\Delta(\gl, \fjl)} 
~\geq~ (p-1) \cdot \parens[\Big]{\delta_{\inn} - \Ex{\ell \in L^*}{s_{\ell}} - \frac{\eps}{2}}
~\geq~ (p-1) \cdot  \parens[\Big]{\delta_{\inn} - s - \frac{2\eps}{3}} \mper
\]
Combining the above bounds and using $\delta_{\inn} \geq \Delta$ gives,
\begin{align*}
\sum_{h\in \calH} \Delta_R(\erase{g},h) 
&~\geq~ (k-p) \cdot (\Delta - s - \eps) + (p-1) \cdot  \parens[\Big]{\Delta - s - \frac{2\eps}{3}} - \frac{p\eps}{6} \\
&~=~ (k-1) \cdot (\Delta - s - \eps) + \frac{(p-1)\eps}{3} - \frac{p\eps}{6} \mcom
\end{align*}
which completes the proof since $p \geq 2$.
\end{proof}
As we will prove in \cref{sec:inner-code}, it is easy to observe using known results by Alrabiah, Guruswami and Li~\cite{AGL24} that a random
linear code satisfies \cref{def:gen_singleton} with high probability, and can thus be used as the
inner code $\calC_{\inn}$. 
Since the inner code is a constant-sized object, we can search over all linear codes in $(\F_q)^d$
of dimension $\rho \cdot d$ for a given rate $\rho$, and the code $\AELC$ then yields an explicit
construction of codes achieving the generalized Singleton bound.
Moreover, if the inner code is required to be fully explicit, it can also be obtained from folded Reed-Solomon codes, using the results by Chen and Zhang~\cite{CZ24}.
\begin{corollary}\label{cor:ael_instantiation}
For every $\rho, \eps \in (0,1)$ and $k \in \N$, there exist explicit inner codes $\calC_{\inn}$ and an infinite family of explicit codes $\AELC \subseteq (\F_q^d)^n$ obtained via the AEL construction that satisfy: 
\begin{enumerate}
\item $\rho(\AELC) \geq \rho$.
\item For any $g \in  (\F_q^d)^n$ and any $\calH \subseteq \AELC$ with $\abs{\calH} \leq k$ that
\[
\sum_{h \in \calH} \Delta(g,h) ~\geq~ (\abs{\calH}-1) \cdot (1 - \rho - \eps) \mper
\]
\item The alphabet size $q^d$ of the code $\AELC$ can be taken to be $2^{O(k^{3k}/\eps^9)}$.
\item $\AELC$ is characterized by parity checks of size $O(k^{2k}/\eps^{11})$ over the field $\F_q$.
\end{enumerate}
\end{corollary}
\begin{proof}
Let $d = O(k^{2k}/\eps^8)$ be such that there exist explicit families of $(n,d,\lambda)$-expander graphs (for arbitrarily large $n$) with $\lambda \leq \eps^4/(2^{18} k^k)$. 
Let $\calC_{\inn} \subseteq \F_q^d$ be a code given by \cref{cor:random-code}, with rate $\rho_{\inn} = \rho + \eps/4$, which is $(1 - \rho, k, \eps/2)$ average-radius list decodable with erasures. Note that the alphabet size $q$ for $\calC_{\inn}$ can be taken to be $2^{O(k + 1/\eps)}$.
Finally, let $\calC_{\out} \subseteq (\F_q^{\rho_{\inn} \cdot d})^n$ be an outer (linear) code with rate $\rho_{\out} = 1 - \eps/4$ and distance (say) $\delta_{\out} = \eps^3/2^{15}$. Explicit families of such codes can be also be obtained (for example) via expander-based Tanner code constructions (see Theorem 11.4.6 and Corollary 11.4.8 in \cite{GRS23}). Using Tanner codes also gives that $\calC_{\out}$ has parity checks of size at most $O(1/\eps^3)$.

Given the above parameters, we have 
$\rho(\AELC) ~\geq~ \rho_{\out} \cdot \rho_{\inn} ~=~ (1-\eps/4) \cdot (\rho + \eps/4) ~\geq~ \rho$.
Since $\lambda \leq \eps \cdot \delta_{\out}/(6k^k)$ and $\calC_{\inn}$ is $(1-\rho, k, \eps/2)$ average-radius list decodable with erasures, we can use \cref{thm:main_technical_avg} to conclude that $\AELC$ is $(1-\rho, k, \eps)$ average-radius list decodable (with erasures) which yields the second condition. 
Since $\calC_{\out}$ has parity checks of $O(1/\eps^3)$, $\calC_{\inn} \subseteq \F_q^d$, and each symbol of $\AELC$ is a function of at most $d$ symbols from $\calC_{\out}$ (encoded via $\calC_{\inn}$), $\AELC$ can be taken to have parity checks of size at most $O(d \cdot (1/\eps^3)) = O(k^{2k}/\eps^{11})$.
Finally, we note that the alphabet size of the code $\AELC$ is $q^d = \exp\inparen{O((k + 1/\eps) \cdot (k^{2k}/\eps^8))} = \exp\inparen{k^{3k}/\eps^9}$, which proves the claim. 
\end{proof}

\vspace{-5 pt}
\paragraph{A weaker consequence of \cref{def:gen_singleton}.}
We also state the following consequence of \cref{def:gen_singleton}, which still yields a
strengthening of the average distance inequality and the generalized Singleton bound (when
instantiated with the appropriate code), and may be of independent interest. 
Note that this statement also yields a corollary of \cref{thm:main_technical_avg} which is simply in
terms of center $g \in \Sigma^n$ with no erasure symbols, and yields an advantage over the distance
inequality, in terms of the error locations which are \emph{common} to all the codewords $h_1, \ldots, h_k$.
\begin{lemma}\label{lemma:common-error-bound}
Let $\calC \subseteq \Sigma^n$ be $(\delta_0, k_0, \eps)$ average radius list-decodable with
erasures. Then, for any $g \in \Sigma^n$, any $k \leq k_0$ and $h_1, \ldots, h_k \in \calC$, we have
that
\[ \sum_{i \in [k]} \Delta(g,h_i) ~\geq~ (k-1) \cdot (\delta_0 - \eps) + \Ex{r \in [n]}{\prod_{i \in
    [k]}\indi{h_{i,r} \neq g_r}} \mper
\]
\end{lemma}
\begin{proof}
Define $\erase{g} \in \Sigma^n$ as 
\[
\erase{g}_r ~=~ 
\begin{cases}
\bot & \text{if} ~g_r \neq h_{i,r} ~\forall i \in [k] \\
g_r &\text{otherwise} \mper
\end{cases}
\]
Note that the fraction of erasure symbols is $s = \Ex{r \in [n]}{\prod_{i \in [k]}\indi{h_{i,r} \neq
    g_r}}$ and $\Delta(\erase{g},h_i) = \Delta(g,h) - s$ for all $i \in [k]$. Applying
\cref{def:gen_singleton} with $\erase{g}$ gives
\[
\sum_{i \in [k]} (\Delta(g,h_i) - s) 
~=~ \sum_{i \in [k]} \Delta(\erase{g},h_i) 
~\geq~ (k-1) \cdot (\delta_0 - s - \eps) \mcom
\]
and rearranging proves the claim.
\end{proof}
\begin{remark}
A reader might notice that the definition of the $\erase{g}$ is the same as used in the proof of
\cref{lemma:inductive}. 
In fact, it is easy to see that the consequence \cref{lemma:common-error-bound} can directly be
proved via induction using the same proof as \cref{thm:main_technical_avg}, which avoids using a
$\erase{g}$ with erasures as part of the induction (although one still needs the list decodability
with erasures for the inner code $\calC_{\inn}$).
While we chose to prove the stronger local-to-global statement as \cref{thm:main_technical_avg}
above, for the algorithmic application we will only prove an algorithmic analogue of
\cref{lemma:common-error-bound} to avoid technical issues with keeping track of arbitrary erasure patterns.
\end{remark}
%

%!TEX root=main.tex

%%% Local Variables:
%%% mode: latex
%%% TeX-master: "main"
%%% End:

\section{Inner Codes meeting generalized Singleton Bound}\label{sec:inner-code}

In this section, we will look at two constructions of inner codes that are average-radius list decodable with erasures, the property we need to instantiate our construction. These will be a random linear code, and folded Reed-Solomon codes.

In the literature, the property of being average-radius list decodable is usually defined without reference to any erasures (\ie $s=0$ in \cref{def:gen_singleton}). Formally, we say that $\calC \subseteq \Sigma^n$ is $(\delta,k,\eps)$ {\deffont average-radius list decodable} if for all $g \in \Sigma^n$ and all $\calH \subseteq \calC$ with $\abs{\calH} \leq k$, we have
\[
\sum_{h \in \calH}{\Delta(g,h)} ~\geq~ \inparen{\abs{\calH}-1} \cdot (\delta - \eps) \mper
\]
The other way of looking at erasures is via puncturings, and we will now see that if a code and its puncturings are average-radius list decodable in the above sense, then it is average-radius list decodable with erasures.

\vspace{-5 pt}
\paragraph{Erasures and Puncturings.}
Let $C \subseteq \F_q^n$ be a code, $S\subseteq [n]$, and denote $s:= \nfrac{|S|}{n}$. Let $\rho =
\rho(C)$ denote the rate of $C$. Define $C_S  \subseteq \F_q^{(1-s)n}$ as the punctured code obtained by removing the coordinates in $S$. 
\begin{claim}\label{claim:puncture}
If for each $S\subseteq [n]$ with  $s \leq 1-\rho$, $C_S$  is $(1-\rho,L, \frac{\varepsilon}{1-s} )$ average-radius list decodable, then $C$ is  $(1-\rho, L, \eps )$ average-radius list decodable with erasures.	
\end{claim}
\begin{proof}
	
Let $\erase{g}$ have erasures in a set $S\subseteq [n]$, and denote by $g_S$, the punctured vectored with these erasure locations removed. Similarly for any list of $L$ codewords $\calH \subseteq C$, denote the punctured list by $\calH_S$. If the code $C_S$ is $(L, \frac{\eps}{1-s} )$ average-radius list decodable (without erasures), then, 

\[
\sum_{h_S \in \calH_S} \Delta(g_S,h_S) ~\geq~ (L - 1) \cdot \parens[\Big]{1 - \rho(C_S) - \frac{\eps}{1-s} }\mper
\]
Observe that if $C$ is a rate $\rho$ code, then $\rho(C_S) \leq \frac{\rho}{1-s}$ (in fact,
$\rho(C_S) = \rho$ if the distance of $C_S$ is greater than 0, but we only need the one-sided inequality). Also $\Delta(\erase{g},h) \geq \Delta(g_S,h_S) \cdot (1-s)$. Multiplying the entire equation by $(1-s)$ and plugging this in, we get,
\[
\sum_{h_S \in \calH_S} \Delta(\erase{g},h) ~\geq~ (L - 1) \cdot \parens[\Big]{1-s - \rho - \eps }\mper \qedhere
\]% which is the requirement in \cref{def:gen_singleton}. 
%Observe that $\Delta(C_S) \geq \frac{\Delta-s}{1-s}$, and $\Delta_S(\erase{g},h) = \Delta_S(g_S,h_S) \cdot (1-s)$. Plugging this in the above equation yields the requirement in \cref{def:gen_singleton}.  
\end{proof}

%. Note that a code $C$ is average-radius list decodable with erasures (\cref{def:gen_singleton}) if $C_S$ is $(L,\eta(1-|S|/n))$\tnote{check this $\eta$ normalization thing} average-radius list decodable (without any erasures) for all subsets of size at most $\Delta n$, the set of such $S$ captures the possible set of erasures. 

\vspace{-5 pt}
\paragraph{Random Linear Codes.} Let $\cG_{n,\rho n,q}$ be the uniform distribution over $n\times \rho n$-matrices over $\F_q$, \ie where each entry of the matrix is picked uniformly at random from $\F_q$. Let $C = \im(G)$ where $G\sim \cG_{n,\rho n,q}$. Then, for any fixed $S\subseteq [n]$,  $C_S = \im(G_S) $ where $G_S\sim \cG_{n-|S|,\rho n,q}$. 

\begin{theorem}[{\cite[Thm. 1.3]{AGL24}}]
Fix an integer $L \geq 1$, $q \geq 2\cdot 2^{10L/\epsilon}$, and $\rho,\epsilon \in (0,1)$.
Then for sufficiently large $n$, a random linear code $C = \im(G)$ where $G\sim \cG_{n,\,\rho n, q}$, is $(1-\rho, L,\epsilon)$ average-radius list decodable with probablity at least $1- \kappa$ , where, 
\[\kappa ~=~  \parens[\Big]{\frac{c_{L,\epsilon}}{q}}^{\lfloor \frac{\epsilon n}{2} \rfloor}, \;\text{ for } c_{L,\epsilon} < 2\cdot 2^{10L/\epsilon}  .\] 
%\[\tau ~\leq~ 2+^{(L+2)n} \cdot {n \choose r} \cdot2^{(L+1)r}\cdot \parens[\bigg]{\frac{L}{q}}^r, \;\text{ for } r = \Big\lfloor \frac{\epsilon n}{2} \Big\rfloor .\] 
\end{theorem}
\begin{proof}
	Their definition of average-radius list decodable gives an inequality for exactly $\ell$ distinct codewords. Thus, we obtain this bound by taking a union over their bounds for $\ell = 1, \cdots, L$. 
\end{proof}

\begin{corollary}\label{cor:random-code}
Let $C$ be a random linear code as generated above for $q \geq 2^{\nfrac{(10L+2)}{\epsilon}}$. Then,
with probability at least $(1-2^{-n/3})$, $C$ is $(1-\rho,L,\eps)$ average-radius list decodable with erasures.\end{corollary}
%\begin{corollary}
%Let $C$ be a random linear code as generated above for $q \gg 2^{10L/\epsilon}\cdot 2^{\frac{3(1-\gamma)}{\epsilon}}$ where $\gamma \in (0,1)$. Then, with high probability, for all $S \subseteq [n]$ of size at most $|S| \leq \gamma n$, the punctured code $C_S$ is $ (L, \epsilon)$ average-radius list decodable. 
%\end{corollary}
\begin{proof}
We will show that with high probability, for all $S \subseteq [n]$ of size at most $|S| \leq (1-\rho)\cdot n$, the punctured code $C_S$ is $ (1-\rho(C_S), L, \epsilon)$ average-radius list decodable. 

For a fixed $S$ of fractional size  $s \leq 1-\rho$, we have that $C_S$ is not
$(1-\rho(C_S),L,\frac{\epsilon}{1-s})$ average-radius list-decodable with probablity at most, 
\[
\leq \parens[\Big]{\frac{c_{L,\frac{\epsilon}{1-s}}}{q}}^{\Big\lfloor \frac{ \frac{\epsilon}{1-s}
      (1- s) n}{2}\Big\rfloor} \leq \parens[\Big]{\frac{c_{L,\eps}}{q}}^{\lfloor\frac{ \eps
      n}{2}\rfloor} = \kappa .
\] 
There are at most $2^n$ choices of the subsets $S$, and thus by a union bound, the probability that
all the punctured codes are average-radius list-decodable is at least 
\[
1 - 2^{n} \kappa ~\geq~ 1 - 2^n \cdot 2^{(10L/\eps + 1) \cdot \eps n/2} \cdot q^{-\lfloor\eps
  n/2\rfloor } ~\geq~ 1 - 2^{-n/3} \mcom
\]
for $q \geq 2^{2/\epsilon}\cdot 2^{10L/\epsilon}$ and $n \geq 60L+12$. 	
\end{proof}
%	 where, 
%\[\tau ~\leq~ 2^{(L+2)n} \cdot {n \choose r} \cdot2^{(L+1)r}\cdot \parens[\bigg]{\frac{L}{q}}^r, \;\text{ for } r = \Big\lfloor \frac{\epsilon n}{2} \Big\rfloor .\] 

%Plugging in $\gamma = \Delta$, we get that 

Thus, for a large enough alphabet size, random linear codes are average-radius list-decodable with erasures.

\vspace{-5 pt}
\paragraph{Folded Reed--Solomon Codes.} A recent work of Chen and Zhang~\cite{CZ24} shows that explicit folded Reed--Solomon (RS) codes are also average-radius list-decodable. Let $\F_q[x]$ denote the set of polynomials with $\F_q$-coefficients, and $\F_q^*$ be the multiplicative cyclic group of non-zero elements. 

\begin{definition}[Folded RS Codes]
Fix $n,b >0$, $\rho \in (0,1)$, and $q \geq bn$. Let $\gamma$ be a generator of $\F_q^*$, and pick $\vec{\alpha} = (\alpha_1,\cdots, \alpha_n) \in \F_q^n$. For $f \in \F_q[x]$, let $\Gamma_i = (f(\alpha_i), f(\gamma\alpha_i), \cdots, f(\gamma^{b-1}\alpha_i))$
Then,
\[
\mathrm{FRS}^{b,\gamma}_{n,\rho}(\vec{\alpha}) = \braces{(\Gamma_1, \cdots, \Gamma_n) \mid \deg(f) < \rho b n} \subseteq (\F_q^{b})^n.
\]	
The code is called \textit{appropriate} if $\{ \gamma^i\alpha_j \mid 0\leq i \leq b-1, j \in [n] \}$ has size $bn$, \ie all values are distinct.
\end{definition}

\begin{theorem}[{\cite[Thm. 1.3]{CZ24}}]
For any integer $L \geq 1$ and $\epsilon \in (0,1)$, and $b \geq L/\epsilon$. Then, an appropriate folded Reed-Solomon code, $\mathrm{FRS}^{b,\gamma}_{n,\rho}(\vec{\alpha})$, is $(1- \rho,L,\varepsilon)$ average-radius list decodable where $\rho$ is the rate of this folded code.
 \end{theorem}
% $1- \frac{s\rho}{s-L+1}
%\tnote{Is plugging in $s = L/\epsilon$ okay?}

Note that puncturing the folded Reed--Solomon code is equivalent to the FRS code over a puncturing of $\vec{\alpha}$. Clearly, a puncturing of an appropriate folded Reed--Solomon code is also an appropriate Reed--Solomon code, and thus using \cref{claim:puncture}, one obtains:  

\begin{corollary}
	For any integer $L \geq 1$ and $\epsilon \in (0,1)$, set $b = L/\epsilon$. Then, an appropriate folded Reed-Solomon code, $\mathrm{FRS}^{s,\gamma}_{n,k}(\vec{\alpha})$, is $(1- \rho,L,\varepsilon)$ average-radius list decodable with erasures.
\end{corollary}

%\tnote{Should I change L to $k$? Not doing that as the papers we cite use $k$ for something else and it would be confusing when someone go reads the reference.}

%For a Reed--Solomon code, puncturing is the same as shrinking the evaluation set. Since the notion of \enquote{appropriate evaluation points} is closed under taking subsets, we immediately obtain that puncturings of the Folded Reed--Solomon code are also average-radius list-decodable. Therefore, by \cref{claim:puncture} 
%For a Reed--Solomon code, puncturing is the same as shrinking the evaluation set. Since the notion of \enquote{appropriate evaluation points} is closed under taking subsets, we immediately obtain that puncturings of the Folded Reed--Solomon code are also average-radius list-decodable.  

%!TEX root=main.tex

%%% Local Variables:
%%% mode: latex
%%% TeX-master: "main"
%%% End:

%\input{sos_prelims.tex}

\section{Decoding using SoS up to the generalized Singleton bound}

Until now, we have proved that when the inner and outer codes as well as the graph used in AEL amplification are suitably chosen, then the list of codewords around an arbitrary center (corrupted codeword) is of small size. 
In this section, we describe a polynomial time algorithm that takes as input the corrupted codeword, and outputs this list.

This algorithm is based on the Sum-of-Squares (SoS) hierarchy of semidefinite programs, which gives a systematic way of tightening convex relaxations.  SoS has been used before for decoding algorithms for codes constructed using spectral expanders in \cite{AJQST20, JQST20, RR23, JST23}. 
Among these, \cite{JST23} used the SoS hierarchy to give a list decoding algorithm for AEL up to the Johnson bound, yielding rate $R$ codes efficiently decodable up to $1-\sqrt{R}-\eps$ for any $\eps>0$. 
We will heavily rely on their framework but will improve the decoding radius to $1-R-\eps$ by proving an SoS analog of the generalized Singleton bound for appropriately instantiated AEL amplification.

Before going into the proof, we describe additional preliminaries for the SoS hierarchy. Readers familiar with the general terms and concepts can skip ahead to \cref{sec:sos_proof} where we define a specific SoS relaxation for the AEL code, and prove an SoS analog of the generalized Singleton bound. Finally, \cref{sec:sos_algo} describes the decoding algorithm in detail.

%\snote{Just wrote some text. Some of this will go to intro, some will get scattered. Mix above with the next para to produce something reasonable.}

\subsection{Additional Preliminaries: Sum-of-Squares Hierarchy}\label{sec:sos_prelims}

The sum-of-squares hierarchy of semidefinite programs (SDPs) provides a family of increasingly
powerful convex relaxations for several optimization problems. 
Each ``level" $t$ of the hierarchy is parametrized by a set of constraints corresponding to
polynomials of degree at most $t$ in the optimization variables. While the relaxations in the
hierarchy can be viewed as  semidefinite programs of size $n^{O(t)}$ \cite{BS14, FKP19}, 
it is often convenient to view the solution as a linear operator, called the ``pseudoexpectation" operator.
%
% It is well-known that such constrained pseudoexpectation operators of SoS-degree $t$ can be described as solutions to semidefinite programs of size $n^{O(t)}$ \cite{BS14, Laurent09}. This hierarchy of semidefinite programs for increasing $t$ is known as the SoS hierarchy.

%
%
%\fnote{Someone familiar with SoS will likely want to skip most of this paragraph (only use it as a reference as needed) and jump to the AEL part of the preliminaries.}
%\vspace{-5 pt}
%
\paragraph{Pseudoexpectations}

Let $t$ be a positive even integer and fix an alphabet $\Sigma$ of size $s$. Let $\zee = \{Z_{i,j}\}_{i\in[m],j\in\Sigma}$ be a collection of variables and $\R[\zee]^{\leq t}$ be the vector space of polynomials of degree at most $t$ in the variables $\zee$ (including the constants).

\begin{definition}[Constrained Pseudoexpectations]\label{def:constraints_on_sos}
Let $\calS = \inbraces{f_1 = 0, \ldots, f_m = 0, g_1 \geq 0, \ldots, g_r \geq 0}$ be a system of
polynomial constraints in $\zee$, with each polynomial in $\calS$ of degree at most $t$. We say $\tildeEx{\cdot}$ is a pseudoexpectation operator of SoS-degree $t$, over the variables $\zee$  respecting $\calS$, if it is a linear operator $ \tildeEx{\cdot}: \R[\zee]^{\leq t} \rightarrow \R$ such that:
	\begin{enumerate}
	\item $\tildeEx{1} = 1$.
    \item $\tildeEx{p^2} \geq 0$ if $p$ is a polynomial in $\zee = \{Z_{i,j}\}_{i\in [m],j\in \Sigma}$ of degree $\leq t/2$.
	\item $\tildeEx{p \cdot f_i} = 0$,  $\forall\, i \in [m]$ and $\forall\, p$ such that $\deg(p \cdot f_i) \leq t$.
	\item $\tildeEx{p^2 \cdot \prod_{i \in S} g_i} \geq 0$, $\forall\, S \subseteq [r]$ and $\forall\, p$ such that $\deg(p^2\cdot \prod_{i \in S} g_i) \leq t$.
	\end{enumerate}
\end{definition}
~

%\tnote{Suggestion -- directly define the constrained version. Attempting to make the connection with assignments a bit more explicit below.}

% An SoS solution of degree $t$, or a pseudoexpectation of SoS-degree $t$, over the variables $\zee$ is represented by a linear operator $ \tildeEx{\cdot}: \R[\zee]^{\leq t} \rightarrow \R$ such that:
%%
%\vspace{-5 pt}
%%
%\begin{enumerate}[(i)]
%    \item $\tildeEx{1} = 1$.
%    \item $\tildeEx{p^2} \geq 0$ if $p$ is a polynomial in $\zee = \{Z_{i,j}\}_{i\in [m],j\in [q]}$ of degree $\leq t/2$.
%\end{enumerate}
%%
%\vspace{-5 pt}
%%
%\tnote{Is the note needed? It is reiterating that it is a linear operator}
% Note that linearity implies $\tildeEx{p_1} + \tildeEx{p_2} = \tildeEx{p_1+p_2}$ and $\tildeEx{c\cdot
%  p_1} = c \cdot \tildeEx{p_1}$ for $c\in \R$, for $p_1, p_2 \in \R[\zee]^{\leq t}$.
%%
%This also allows for a succinct representation of $\tildeEx{\cdot}$ using any basis for $\R[\zee]^{\leq t}$.
%
%\tnote{ }

% to $m$ variables in alphabet $[q]$.

Let $\mu$ be a distribution over the set of assignments, $\Sigma^m$.  Define the following collection of random variables,  \[\zee = \braces[\big]{ \, Z_{i,j}  = \indi{ i \mapsto j}\, \mid \, i\in[m],\, j\in\Sigma } .\] Then, setting $\tildeEx{p(\zee)} = \Ex{\mu}{p(\zee)} $ for any polynomial $p(\cdot)$ defines an (unconstrained) pseudoexpectation operator. However, the reverse is not true when $t < m$, and there can be degree-$t$ pseudoexpectations that do not correspond to any genuine distribution, $\mu$. Therefore, the set of all pseudoexpectations should be seen as a relaxation for the set of all possible distributions over such assignments. The upshot of this relaxation is that it is possible to optimize over the set. Under certain conditions on the bit-complexity of solutions~\cite{OD16, RW17:sos}, one can optimize over the set of degree-$t$ pseudoexpectations in time $m^{O(t)}$ via SDPs.

\paragraph{Local constraints and local functions.}
Any constraint that involves at most $k$ variables from $\zee$, with $k\leq t$, can be written as a degree-$k$ polynomial, and such constraints may be enforced into the SoS solution.
%
% \paragraph{Canonical usage}
%
In particular, we will always consider the following canonical constraints on the variables $\zee$.
\ifnum\confversion=1
\begin{align*}
&Z_{i,j}^2 = Z_{i,j},\ \forall i\in[m],j\in[s] \\
\text{and} \quad &\sum_j Z_{i,j} = 1,\ \forall i\in[m] \mper
\end{align*}
\else
\[
Z_{i,j}^2 = Z_{i,j},\ \forall \,i\in[m],j\in\Sigma
\quad \text{and} \quad 
\sum_j Z_{i,j} = 1,\ \forall\, i\in[m] \mper
\]
\fi
% As shown in the previous section, an assignment $f:[m]\rightarrow [q]$ is encoded using $m$ characteristic vectors, and so we wish to impose the following local constraints on the variables $\zee$:
%
% \begin{enumerate}[(i)]
% 	\item $Z_{i,j}^2 = Z_{i,j},\ \forall i\in[m],j\in[q]$.
% 	\item $\sum_j Z_{i,j} = 1,\ \forall i\in[m]$.
% \end{enumerate}
%
% These constraints are enforced as described in \cref{def:constraints_on_sos}, and we will henceforth not explicitly mention it. \snote{See Madhur's comment.}
%
We will also consider additional constraints and corresponding polynomials, defined by ``local" functions. For any $f\in \Sigma^m$ and $M\sub [m]$, we use $f_M$ to denote the restriction $f|_M$, and $f_i$ to denote $f_{\{i\}}$ for convenience.
\begin{definition}[$k$-local function]
	A function $\mu: \Sigma^m \rightarrow \R$ is called $k$-local if there is a set $M\subseteq [m]$ of size $k$ such that $\mu(f)$ only depends on $\inbraces{f(i)}_{i\in M}$, or equivalently, $\mu(f)$ only depends on $f|_M$.
	
	If $\mu$ is $k$-local, we abuse notation to also use $\mu: \Sigma^M \rightarrow \R$ with $\mu(\alpha) = \mu(f)$ for any $f$ such that $f|_M=\alpha$. It will be clear from the input to the function $\mu$ whether we are using $\mu$ as a function on $\Sigma^m$ or $\Sigma^M$.
\end{definition}

Let $\mu:\Sigma^m\rightarrow \R$ be a $k$-local function that depends on coordinates $M\subseteq [m]$ with $|M|=k$. Then $\mu$ can be written as a degree-$k$ polynomial $P_{\mu}$ in $\zee$:
\[
	P_{\mu}(\zee) = \sum_{\alpha \in \Sigma^M} \parens[\Big]{\mu(\alpha) \cdot\prod_{i\in M} Z_{i,\alpha_i}}
\]

% To see how $p_{\mu}$ is related to the $k$-local function $\mu$, observe that $p_{\mu}\inparen{\zee = \inbraces{z^{(f)}_{i,j}}} = \mu(f)$.

With some abuse of notation, we let $\mu(\zee)$ denote $P_{\mu}(\zee)$. We will use such $k$-local
functions inside $\tildeEx{\cdot}$ freely without worrying about their polynomial
representation. For example, $\tildeEx{ \indi{\zee_{i} \neq j}}$ denotes $\tildeEx{ 1- Z_{i,j}}$. Likewise, sometimes we will say we set $\zee_i = j$ to mean that we set $Z_{i,j} = 1$ and $Z_{i,j'} = 0$ for all $j'\in \Sigma \backslash \{j\}$.

% $\tildeEx{\cdot}$ operator applied to the polynomial corresponding to the $1$-local function $\mu: [q]^m \rightarrow \R$ which is defined as: $\mu(f)$ is $1$ if $f_i = 0$ and $\mu(f) = 0$ otherwise.
% %
The notion of $k$-local functions can also be extended from real-valued functions to vector-valued functions straightforwardly.

\begin{definition}[Vector-valued local functions]
A function $\mu: \Sigma^m \rightarrow \R^N$ is $k$-local if the $N$ real valued functions corresponding to the $N$ coordinates are all $k$-local. Note that these different coordinate functions may depend on different sets of variables, as long as the number is at most $k$ for each of the functions.
\end{definition}
\paragraph{Local distribution view of SoS}

It will be convenient to use a shorthand for the function $\indi{\zee_A = \alpha}$, and we will use $\zee_{A,\alpha}$. Likewise, we use $\zee_{i,j}$ as a shorthand for the function $\indi{\zee_i = j}$. That is, henceforth,
\ifnum\confversion=1
\begin{align*}
	&\tildeEx{\zee_{A,\alpha}} = \tildeEx{\indi{\zee_A = \alpha}} = \tildeEx{ \prod_{a\in A}
                             Z_{a,\alpha_a}}
                             \\
	\text{and } \quad & \tildeEx{\zee_{i,j}} = \tildeEx{\indi{\zee_i = j}} = \tildeEx{ Z_{i,j}}.
\end{align*}
\else
\begin{align*}
	\tildeEx{\zee_{A,\alpha}} ~=~ \tildeEx{\indi{\zee_A = \alpha}} ~=~ \tildeE\brackets[\Big]{\prod_{a\in A}
                      Z_{a,\alpha_a}
                                    }
\qquad \text{and} \qquad
	\tildeEx{\zee_{i,j}} ~=~ \tildeEx{\indi{\zee_i = j}} = \tildeEx{ Z_{i,j}}
\end{align*}
\fi

% Note that for any degree-$t$ pseudoexpectation operator $\tildeEx{\cdot}$ with $t\geq 2$,
% \[
% 	\sum_{j\in [q]} \tildeEx{\zee_{i,j}} = \sum_{j} \tildeEx{Z_{i,j}} = \tildeEx{\sum_j Z_{i,j}}= 1
% \qquad \text{and} \qquad
% 	\tildeEx{\zee_{i,j}} = \tildeEx{Z_{i,j}} = \tildeEx{Z_{i,j}^2} \geq 0
% \]
% Thus, the real values $\inbraces{\tildeEx{\zee_{i,j}}}_{j\in [q]}$ define a distribution over $[q]$
% , which we will sometimes call local distribution for $\zee_i$. 

% In fact, this argument can be extended to define local distributions for $\zee_S$ for $|S|\leq
% t/2$. Let $S \subseteq [m]$ be such that $|S|=k\leq t/2$,
Note that for any $A \subseteq [m]$ with $\abs*{A} = k \leq t/2$,
\ifnum\confversion=1
\begin{gather*}
	\sum_{ \alpha \in \Sigma^{k}} \tildeEx{\zee_{A,\alpha}} = 
% \sum_{ \alpha \in [q]^{k}} \tildeEx{ \prod_{s\in S} Z_{s,\alpha_s}} 
 \tildeEx{ \prod_{a\in A} \inparen{ \sum_{j\in \Sigma} Z_{a,j}} } = 1
\\
	\tildeEx{\zee_{A,\alpha}} = \tildeEx{ \prod_{a\in A} Z_{a,\alpha_a}} = \tildeEx{ \prod_{a\in
            A} Z^2_{a,\alpha_a}} \geq 0 \mper
\end{gather*}
\else
\[
	\sum_{ \alpha \in \Sigma^{k}} \tildeEx{\zee_{A,\alpha}} = 
% \sum_{ \alpha \in [q]^{k}} \tildeEx{ \prod_{s\in S} Z_{s,\alpha_s}} 
 \tildeE\brackets[\bigg]{ \prod_{a\in A} \parens[\bigg]{\sum_{j\in \Sigma} Z_{a,j} } } = 1
\qquad \text{and} \qquad
	\tildeEx{\zee_{A,\alpha}} = \tildeE\brackets[\bigg]{ \prod_{a\in A} Z_{a,\alpha_a}} = \tildeE\brackets[\bigg]{ \prod_{a\in
            A} Z^2_{a,\alpha_a}} \geq 0 \mper
\]
\fi

Thus, the values $\inbraces{\tildeEx{\zee_{A, \alpha}}}_{\alpha\in \Sigma^A}$ define a distribution
over $\Sigma^k$. We call this the local distribution for $\zee_A$, or simply for $A$.
% which we think of as the local distribution for $\zee_S$.
%
% Given a distribution $\calD$ over assignments in $[q]^m$, the local distribution induced by $\PExp^{(\calD)}[\cdot]$ on $\zee_S$ is precisely the marginal distribution induced by $\calD$ for the set $S$. In this view, degree-$t$ pseudoexpectations give us consistent marginal distributions over sets of size at most $t/2$ that may not correspond to any global distribution, and allow us to optimize over this set in time $m^{\calO(t)}$.
%
Let $\mu: \Sigma^m \rightarrow\R$ be a $k$-local function for $k\leq t/2$, depending on $M \subseteq
[m]$. Then,
\ifnum\confversion=1
\begin{align*}
	\tildeEx{\mu(\zee)} 
~=~& \tildeEx{\sum_{\alpha\in \Sigma^M} \inparen{\mu(\alpha) \cdot\prod_{i\in M} Z_{i,\alpha_i}}}\\
~=~& \sum_{\alpha\in \Sigma^M} \mu(\alpha) \cdot \tildeEx{\prod_{i\in M} Z_{i,\alpha_i}}\\
~=~& \sum_{\alpha\in \Sigma^M} \mu(\alpha) \cdot \tildeEx{\zee_{M,\alpha}}
\end{align*}
\else
\begin{align*}
	\tildeEx{\mu(\zee)} 
~=~ \tildeE\brackets[\bigg]{\sum_{\alpha\in \Sigma^M} \parens[\bigg]{\mu(\alpha) \cdot\prod_{i\in M} Z_{i,\alpha_i}}}
~=~ \sum_{\alpha\in \Sigma^M} \mu(\alpha) \cdot \tildeE\brackets[\Big]{\prod_{i\in M} Z_{i,\alpha_i}}
~=~ \sum_{\alpha\in \Sigma^M} \mu(\alpha) \cdot \tildeEx{\zee_{M,\alpha}}
\end{align*}
\fi

That is, $\tildeEx{\mu(\zee)}$ can be seen as the expected value of the function $\mu$ under the local distribution for $M$.

\begin{claim}\label{claim:sos_domination}
	Let $\tildeEx{\cdot}$ be a degree-$t$ pseudoexpectation. For $k \leq t/2$, let $\mu_1,\mu_2$
        be two $k$-local functions on $\Sigma^m$, depending on the same set of coordinates $M$, and
        $\mu_1(\alpha) \leq \mu_2(\alpha) ~~\forall \alpha \in \Sigma^M$. Then $\tildeEx{\mu_1(\zee)} \leq \tildeEx{\mu_2(\zee)}$.
%
 % Suppose that for any $\alpha\in [q]^M$, $\mu_1(\alpha) \leq \mu_2(\alpha)$. Then
 %        \[
 %        	\tildeEx{\mu_1(\zee)} \leq \tildeEx{\mu_2(\zee)}
 %        \]
\end{claim}

\begin{proof}
Let $\calD_M$ be the local distribution induced by $\tildeEx{\cdot}$ for $\zee_M$. Then
$\tildeEx{\mu_1(\zee)} = \Ex{\alpha \sim \calD_M}{\mu_1(\alpha)}$, and $\tildeEx{\mu_2(\zee)} =
\Ex{\alpha\sim \calD_M}{\mu_2(\alpha)}$, which implies $\tildeEx{\mu_1(\zee)} \leq \tildeEx{\mu_2(\zee)}$.
%
% Since $\mu_1(\alpha) \leq \mu_2(\alpha)$ for every $\alpha\in [q]^M$, 
% \[
% 	\Ex{\alpha \sim \calD_M}{\mu_1(\alpha)} \leq \Ex{\alpha\sim \calD_M}{\mu_2(\alpha)}
% \]
% and so,
% \[
% 	\tildeEx{\mu_1(\zee)} \leq \tildeEx{\mu_2(\zee)}
% \]
\end{proof}
The previous claim allows us to replace any local function inside $\tildeEx{\cdot}$ by another local function that dominates it. We will make extensive use of this fact.
\vspace{-5 pt}
%\tnote{Move covariance, conditioning stuff to appendix 1? It seems to disrupt the flow and we only need to cite it in section 6 for Lemma 6.1?}
\paragraph{Covariance for SoS solutions}
Given two sets $S,T \sub [m]$ with $|S|,|T|\leq k/4$, we can define the covariance between indicator random variables of $\zee_S$ and $\zee_T$ taking values $\alpha$ and $\beta$ respectively, according to the local distribution over $S \cup T$. This is formalized in the next definition.
\begin{definition}
Let $\tildeEx{\cdot}$ be a pseudodistribution operator of SoS-degree-$t$, and $S,T$ are two sets of
size at most $t/4$, and $\alpha\in \Sigma^S$, $\beta\in \Sigma^T$, we define the pseudo-covariance and
pseudo-variance,
\ifnum\confversion=1
\small
\begin{gather*}
\tildecov(\zee_{S,\alpha},\zee_{T,\beta}) 
= \tildeEx{ \zee_{S,\alpha} \cdot \zee_{T,\beta} } - \tildeEx{\zee_{S,\alpha}} \tildeEx{\zee_{T,\beta}} \\	
\tildeVar{\zee_{S,\alpha}} ~=~ \tildecov(\zee_{S,\alpha},\zee_{S,\alpha})
\end{gather*}
\normalsize
\else
\begin{align*}
\tildecov(\zee_{S,\alpha},\zee_{T,\beta}) 
~&=~ \tildeEx{ \zee_{S,\alpha} \cdot \zee_{T,\beta} } - \tildeEx{\zee_{S,\alpha}} \,\tildeEx{\zee_{T,\beta}}\\ 		
\tildeVar{\zee_{S,\alpha}} ~&=~ \tildecov(\zee_{S,\alpha},\zee_{S,\alpha})
\end{align*}
\fi
The above definition is extended to pseudo-covariance and pseudo-variance for pairs of sets $S,T$, 
as the sum of absolute value of pseudo-covariance for all pairs $\alpha,\beta$ :
\ifnum\confversion=1
\begin{gather*}
\tildecov(\zee_S,\zee_T) 
~=~ \sum_{\alpha\in \Sigma^S \atop \beta\in \Sigma^T} \abs*{ \tildecov(\zee_{S,\alpha},\zee_{T,\beta}) } \\
\tildeVar{\zee_S} ~=~ \sum_{\alpha\in \Sigma^S} \abs*{ \tildeVar{\zee_{S,\alpha} } }
\end{gather*}
\else
\begin{align*}
\tildecov(\zee_S,\zee_T) 
~&=~ \sum_{\alpha\in \Sigma^S, \beta\in \Sigma^T} \abs*{ \tildecov(\zee_{S,\alpha},\zee_{T,\beta}) }\\[3pt]
\tildeVar{\zee_S} ~&=~ \sum_{\alpha\in \Sigma^S} \abs*{ \tildeVar{\zee_{S,\alpha} } }
\end{align*}
\fi
\end{definition}

% These definitions can be extended to pseudocovariances and pseudo-variances for pairs of sets $S,T$ with $|S|,|T|\leq t/4$ as the sum of absolute value of pseudocovariance for all pairs $\alpha,\beta$.
% \begin{definition}
% Let $\tildeEx{\cdot}$ be a pseudodistribution operator of SoS-degree-$t$, and $S,T$ are two sets of size at most $t/4$, we define the pseudo-covariance between $\zee_S$ and $\zee_T$ as,
% 	\begin{align*}
% 		\tildecov(\zee_S,\zee_T) = \sum_{\alpha\in [q]^S,\beta\in [q]^T} \abs*{ \tildecov(\zee_{S,\alpha},\zee_{T,\beta}) }
% 	\end{align*}
% We also define the analogous pseudovariance as,
% 	\begin{align*}
% 		\tildeVar{\zee_S} = \sum_{\alpha\in [q]^S} \abs*{ \tildeVar{\zee_{S,\alpha} } }
% 	\end{align*}
% \end{definition}
%
We will need the fact that $\tildeVar{\zee_S}$ is bounded above by 1, since,
\ifnum\confversion=1
\begin{align*}
\tildeVar{\zee_S} 
&~=~ \sum_{\alpha} \abs*{\tildeVar{\zee_{S,\alpha}}} \\
&~=~ \sum_{\alpha}\inparen{
          \tildeEx{\zee_{S,\alpha}^2} - \tildeEx{\zee_{S,\alpha}}^2} \\
&~\leq~ \sum_{\alpha} \tildeEx{\zee_{S,\alpha}^2} \\
&~=~ \sum_{\alpha} \tildeEx{\zee_{S,\alpha}} 
~=~ 1.
\end{align*}
\else
\[
\tildeVar{\zee_S} 
~=~ \sum_{\alpha} \abs*{\tildeVar{\zee_{S,\alpha}}} 
~=~ \sum_{\alpha}\inparen{
          \tildeEx{\zee_{S,\alpha}^2} - \tildeEx{\zee_{S,\alpha}}^2} 
~\leq~ \sum_{\alpha} \tildeEx{\zee_{S,\alpha}^2} 
~=~ \sum_{\alpha} \tildeEx{\zee_{S,\alpha}} 
~=~ 1.
\]
\fi

% \begin{claim}\quad	$\tildeVar{\zee_S} \leq 1$.
% \end{claim}
% %
% \begin{proof}
% \begin{align*}
% 	\tildeVar{\zee_S} &= \sum_{\alpha} \abs*{\tildeVar{\zee_{S,\alpha}}} = \sum_{\alpha}\inparen{ \tildeEx{\zee_{S,\alpha}^2} - \tildeEx{\zee_{S,\alpha}}^2} \leq \sum_{\alpha} \tildeEx{\zee_{S,\alpha}^2} = \sum_{\alpha} \tildeEx{\zee_{S,\alpha}} = 1
% \end{align*}
% \end{proof}
%
\vspace{-5 pt}
\paragraph{Conditioning SoS solutions.}
We will also make use of conditioned pseudoexpectation operators, which may be defined in a way
similar to usual conditioning for true expectation operators, as long as the event we condition on
is local. 
The conditioned SoS solution is of a smaller degree but continues to respect the constraints that the original solution respects.

\begin{definition}[Conditioned SoS Solution] Let $F \subseteq \Sigma^m$ be subset (to be thought of as an event) such that $\one_F:\Sigma^m \rightarrow \{0,1\}$ is a $k$-local function. Then for every $t>2k$, we can condition a pseudoexpectation operator of SoS-degree $t$ on $F$ to obtain a new conditioned pseudoexpectation operator $\condPE{\cdot}{F}$ of SoS-degree $t-2k$, as long as $\tildeEx{\one^2_F(\zee)}>0$. The conditioned SoS solution is given by
\[
	\condPE{ p(\zee)}{F(\zee) } \defeq \frac{\tildeEx{p(\zee) \cdot \one^2_{F}(\zee)}}{\tildeEx{\one^2_{F}(\zee)}}
\]
where $p$ is any polynomial of degree at most $t-2k$.
\end{definition}

%\snote{Mention that constraints \emph{respected} by SoS remain true after conditioning.}

We can also define pseudocovariances and pseudo-variances for the conditioned SoS solutions.
\begin{definition}[Pseudocovariance]
	Let $F\sub \Sigma^m$ be an event such that $\one_F$ is $k$-local, and let $\tildeEx{\cdot}$ be a pseudoexpectation operator of degree $t$, with $t>2k$. Let $S,T$ be two sets of size at most $\frac{t-2k}{2}$ each. Then the pseudocovariance between $\zee_{S,\alpha}$ and $\zee_{T,\beta}$ for the solution conditioned on event $F$ is defined as,
\ifnum\confversion=1
\begin{multline*}
\tildecov(\zee_{S,\alpha},\zee_{T,\beta} \vert F) = \\
\tildeEx{ \zee_{S,\alpha} \zee_{T,\beta} \vert F} - \tildeEx{\zee_{S,\alpha} \vert F} \tildeEx{\zee_{T,\beta} \vert F}
\end{multline*}
\else
\begin{align*}
\tildecov(\zee_{S,\alpha},\zee_{T,\beta} \vert F) 
~=~ \tildeEx{ \zee_{S,\alpha} \zee_{T,\beta} \vert F} - \tildeEx{\zee_{S,\alpha} \vert F} ~ \tildeEx{\zee_{T,\beta} \vert F}
\end{align*}
\fi
\end{definition}

We also define the pseudocovariance between $\zee_{S,\alpha}$ and $\zee_{T,\beta}$ after
conditioning on a random assignment for some $\zee_V$ with $V\sub [m]$. 
Note that here the random assignment for $\zee_V$ is chosen according to the local distribution for
the set $V$.

\begin{definition}[Pseudocovariance for conditioned pseudoexpectation operators]
\ifnum\confversion=1
\begin{multline*}
\tildecov(\zee_{S,\alpha},\zee_{T,\beta} \vert \zee_V) 
= \\ \sum_{\gamma \in \Sigma^V}   \tildecov(\zee_{S,\alpha},\zee_{T,\beta} \vert \zee_V = \gamma) \cdot \tildeEx{\zee_{V,\gamma}}
	\end{multline*}
\else
\begin{align*}
\tildecov(\zee_{S,\alpha},\zee_{T,\beta} \vert \zee_V) 
~=~ \sum_{\gamma \in \Sigma^V}   \tildecov(\zee_{S,\alpha},\zee_{T,\beta} \vert \zee_V = \gamma) \cdot \tildeEx{\zee_{V,\gamma}}
\end{align*}
\fi
\end{definition}

And we likewise define $\tildeVar{\zee_{S,\alpha} \vert \zee_V}$, $\tildecov(\zee_S, \zee_T \vert \zee_V)$ and $\tildeVar{\zee_S \vert \zee_V}$.

%%% Local Variables:
%%% mode: latex
%%% TeX-master: "main"
%%% End:

\subsection{Pseudocodewords achieve the generalized Singleton bound}\label{sec:sos_proof}

We will show that the proof of list decodability established for integral codewords can be extended to SoS relaxations of these codewords, under certain average-case low covariance conditions. 
First, we describe what an SoS relaxation of an AEL codeword looks like. This relaxation will be such that true codewords are always feasible solutions, but there may be other feasible solutions. These feasible solutions are therefore called \textit{pseudocodewords}.

The main goal will be to show that for some large constant $t$ depending on $k$ and $\eps$, but independent of $n$, the degree-$t$ SoS relaxation is tight enough to decode up to radius $\frac{k-1}{k}(1-R-\eps)$. We will actually work with $k$-tuples of pseudocodewords, and the show that this tuple---after some random conditioning---satisfies a generalized Singleton bound just like a list of true (distinct) codewords.

\paragraph{SoS relaxations for AEL codewords}
Let $G(L,R,E)$ be the bipartite $(n,d,\lambda)$-expander on which the AEL code is defined, and let $\calC_\inn\subseteq \Sigma_\inn^d$ be the inner code. 

Before going into the details of the SoS relaxation for AEL codes, let us set up some convenient notation. For sets $S\sub [k]$ and $F \sub E$, we use $\zee_{S,F}$ to index their Cartesian product $\zee_{S \times F}$. Further, since we will be often dealing with the case when $F=N(\li)$ or $F=N(\ri)$, we will use $\zee_{S,\li}$ as a shorthand for $\zee_{S,N(\li)}$, and similarly use $\zee_{S,\ri}$ as a shorthand for $\zee_{S,N(\ri)}$. For example,
\[
	\tildeVar{\zee_{[k],\li}} = \tildeVar{ \zee_{[k] \times N(\li)}}.
\]

\begin{definition}[$k$-tuple of pseudocodewords]\label{def:k_tuple}
A \emph{$k$-tuple of psueocodewords of degree-$t$} is a pseudoexpectation operator $\tildeEx{\cdot}$ of SoS-degree-$t$ defined on the variables $\{Z_{i,e}\}_{i\in [k], e\in E}$ over the alphabet $\Sigma_{\inn}$ that respects the following constraints:
\[
	\forall i\in [k], \forall \li \in L,~~~ \zee_{i,N(\li)} \in \calC_{\inn}
\] 
\end{definition}

It is easy to see that any $k$ strings $f_1,\cdots ,f_k$ in $\calC_{\inn}^L$ can be used to define a $k$-tuple of pseudocodeword, by simply setting $\zee_{i,e}$ to be $(f_i)_e$ for all $i \in [k], e\in E$. Note that these strings need not be distinct, as written. However, the set of $k$-tuples of pseudocodewords is more general that just integral strings, and in particular, is convex. However, we can still generalize notions like distance for these objects.

\begin{definition}[Distances and pseudocodewords]
	For an $i\in [k]$, we define the left and right distances between the $i^{th}$ component of a $k$-tuple of pseudocodeword $\tildeEx{\cdot}$ of SoS-degree $t\geq 2d$ and any string $g \in \Sigma_{\inn}^E$ as
	\begin{align*}
		\tildeEx{\Delta_L(g, \zee_i)} &\defeq \Ex{\li \in L}{\tildeEx{\indi{g_{\li} \neq \zee_{i, \li}}}} \\
%	\dis(\tildeEx{\cdot},h) &\defeq \Ex{e}{\tildeEx{\indi{\zee_e \neq h_e}}} \\
		\tildeEx{\Delta_R(g,\zee_i)} &\defeq \Ex{\ri\in R}{\tildeEx{\indi{g_{\ri} \neq \zee_{i, \ri}}}}.
	\end{align*}
\end{definition}

The above can be seen as counting the fraction of errors between $\zee_i$ and $g$. We will also need another piece of notation which will make it easier to track the fraction of errors that are common to an entire set $S\sub [k]$.
\begin{align*}
	\tildeEx{\Delta_R(g,\zee_S)} \defeq \Ex{\ri \in R}{\tildeEx{\Pi_{i\in S} \indi{g_{\ri} \neq \zee_{i,\ri}}}}
\end{align*}

We will also need a notion of distance between two pseudocodewords within a $k$-tuple.
\begin{definition}[Distances between pseudocodewords in a $k$-tuple]
Let $\tildeEx{\cdot}$ be a $k$-tuple of pseudocodewords. Let $i,i' \in [k]$ be indices, then the distance between $i^{th}$ and $(i')^{th}$ pseudocodewords can be defined as
\begin{align*}
		\tildeEx{\Delta_L(\zee_i, \zee_{i'})} &\defeq \Ex{\li \in L}{\tildeEx{\indi{\zee_{i, \li} \neq \zee_{i', \li}}}} \\
%	\dis(\tildeEx{\cdot},h) &\defeq \Ex{e}{\tildeEx{\indi{\zee_e \neq h_e}}} \\
		\tildeEx{\Delta_R(\zee_i,\zee_{i'})} &\defeq \Ex{\ri\in R}{\tildeEx{\indi{\zee_{i,\ri} \neq \zee_{i', \ri}}}}.
	\end{align*}
\end{definition}

\paragraph{$\eta$-goodness} Instead of working with arbitrary pseudocodewords, we will work with those  that have small pseudocovariances across a typical edge in $G$. This key property, called $\eta$-goodness, has been used in prior works and we extend its definition from \cite{JST23} to $k$-tuple of pseudocodewords as follows:

\begin{definition}[$\eta$-good pseudocodeword]
    A $k$-tuple of pseudocodewords $\tildeEx{\cdot}$ of degree at least $2kd$ is called $\eta$-good if
    \[
        \Ex{\li,\ri}{\tildeCov{\zee_{[k]\times N(\li)}}{\zee_{[k]\times N(\ri)}}} \leq \eta.
    \]
\end{definition}

The key upshot of having this property is~\cref{lem:avg_corr} which we use in our proof.  Moreover, one can obtain such $\eta$-good pseudocodewords by randomly conditioning an SoS solution. We relegate the proof of these details to \cref{sec:appendix} as they are an adaptation of earlier proofs.

%
%\begin{definition}[AEL Pseudocodewords]
%	For $t\geq 2d$, we define a degree-$t$ AEL pseuocodeword to be a degree-$t$ pseudoexpectation operator $\tildeEx{\cdot}$ on $\zee$ respecting the following constraints:
%	\begin{align*}
%		\forall \li \in L,\, i\in [k_0]\quad \zee_{i, \li} \in \calC_\inn
%	\end{align*}
%\end{definition}
%
%Next we define the distances between a pseudocodeword and a codeword of $\AELC$. 

%\tnote{Moved the definition of eta-good here. Is the notation of Z consistent?}

The main result of this section is the following generalization of~\cref{lemma:common-error-bound} to $\eta$-good pseudocodewords.

\begin{theorem}\label{thm:sos_main}
	Let $k_0\geq 1$ be an integer and let $\eps > 0$. Let $\AELC$ be a
code obtained using the AEL construction using $(G,\calC_{\out}, \calC_{\inn})$, where $\calC_{\inn}$ is $(\delta_0, k_0, \eps/2)$ average-radius list decodable with erasures, and the graph $G$ is a $(n,d,\lambda)$-expander. 
	Let $\tildeEx{\cdot}$ be an $\eta$-good $k_0$-tuple of pseudocodewords that satisfies the following pairwise distance property on the left:
	\[
		\forall i,i'\in [k_0] \text{ with } i\neq i', \qquad \tildeEx{\Delta_L(\zee_i, \zee_{i'})} \geq \beta\mper
	\]
	Further, assume that $\lambda \leq \frac{\beta}{12k_0^{k_0}}\cdot \eps$ and $\eta \leq \frac{\beta}{12k_0^{k_0}}\cdot \eps$.	Then, for any $g \in (\Sigma_{\inn}^d)^R$,
%	, and any $K \sub [k_0]$ with $|K|=k$, 
	\begin{align*}
	\sum_{i\in [k_0]} \tildeEx{\Delta_R(g,\zee_i)} ~&\geq~ (k_0-1)(\delta_0-\eps) + \Ex{\ri \in R}{\tildeEx{ \Pi_{i\in [k_0]} \indi{g_{\ri} \neq \zee_{i,\ri}} }} \\
	~&=~ (k_0-1)(\delta_0 -\eps) + \tildeEx{\Delta_R(g,\zee_{[k_0]})}.
%		\Ex{i\in [k]}{\tildeEx{\Delta_R(g,\zee_i)}} \geq \frac{k-1}{k} \inparen{ \Delta - \Ex{\ri \in R}{\tildeEx{ \Pi_{i\in[k]} \indi{g_{\ri} \neq \zee_{i,\ri}} }} }
	\end{align*}
\end{theorem}
%
%\begin{proof}
%By induction on $k$. The case $k=1$ is trivial.
Just like in the proof of our main theorem in \cref{sec:avg-singleton}, we need preparatory claims about the existence of a nice partition and (a variant of) $L^*$. We start by proving analytic generalizations of \cref{lem:type_arg} and \cref{lemma:inductive}. Let $\Tau$ denote the set of all partitions\footnote{Note that now we work with partitions of $[k]$ rather a list of codewords $\calH$.} of $[k]$. For an $\li \in L$ and $\tau \in \Tau$, define the local function, 
	\[
		\Lambda_{\li, \tau}(\zee) \defeq \indi{ (\zee_{1, \li}, \zee_{2, \li}, \cdots , \zee_{k, \li}) \text{ induces partition }\tau}.
	\]
	By definition, for every $\li \in L$,
	\[
		\sum_{\tau \in \Tau} \Lambda_{\li, \tau}(\zee) = 1.
	\]
	Let $\tau_1 \in \Tau$ be the trivial partition with only 1 part. In the integral proof, we said that there must be a non-trivial partition that is induced on a significant fraction of left vertices, and used it to define the set $L^*$. For tuples of pseudocodewords, each left vertex will be inducing a distribution over all possible partitions, and we will argue that there is a partition $\tau^*$ that receives a significant mass among these distributions on average over all $\li \in L$. The indicator of this partition $\Lambda_{\li,\tau^*}(\zee)$ will then play the role of the indicator of the set $L^*$ in integral proof.
	\begin{claim}[Generalization of \cref{lem:type_arg}]\label{claim:best_partition}
		There exists a $\tau^*\in \Tau \setminus \{\tau_1\}$ such that 
		\[
			\tildeEx{\Ex{\li \in L}{\Lambda_{\li, \tau^*}(\zee)}} \geq \frac{\beta}{k^k}.
		\]
	\end{claim}
	\begin{proof}
		If $(\zee_{1, \li}, \zee_{2, \li}, \cdots , \zee_{k, \li})$ induces the partition $\tau_1$, then $\zee_{1,\li} = \zee_{2,\li}$. That is,
		\[
			\Lambda_{\li,\tau_1}(\zee) \leq \indi{\zee_{1,\li} = \zee_{2,\li}}.
		\]
		We bound the contribution from the trivial partition as 
		\[
			\tildeEx{\Ex{\li}{\Lambda_{\li,\tau_1}(\zee)}} \leq \tildeEx{\Ex{\li}{\indi{\zee_{1,\li} = \zee_{2,\li}}}} = \tildeEx{1-\Delta_L(\zee_1,\zee_2)} \leq 1-\beta
		\]
%		Since $\tildeEx{\Delta_L(\zee_1,\zee_2)} \geq \beta$,
%		\begin{align*}
%			&\tildeEx{\Ex{\li}{\indi{\zee_{1, \li} \neq \zee_{2, \li}}}} ~\geq~ \beta \\
%			\implies &~\tildeEx{\Ex{\li}{\indi{\zee_{1, \li} = \zee_{2,\li}}}} ~\leq~ 1-\beta \\
%			\implies &~\tildeEx{\Ex{\li}{\Lambda_{\li,\tau_1}(\zee)}} ~\leq~ 1-\beta.
%		\end{align*}
		The above shows that the partition $\tau_1$ cannot be too common. Next, we use $\sum_{\tau\in \Tau} \Lambda_{\li,\tau}(\zee) =1$ to show that:
		\begin{align*}
			1 = \tildeEx{\Ex{\li}{\sum_{\tau \in \Tau}{\Lambda_{\li,\tau}(\zee)}}} = \tildeEx{\Ex{\li}{\Lambda_{\li,\tau_1}(\zee)}} + \sum_{\tau \in \Tau \setminus \{\tau_1\}} \tildeEx{\Ex{\li}{\Lambda_{\li,\tau}(\zee)}} \\
			\implies\sum_{\tau \in \Tau \setminus \{\tau_1\}} \tildeEx{\Ex{\li}{\Lambda_{\li,\tau}(\zee)}} = 1 - \tildeEx{\Ex{\li}{\Lambda_{\li,\tau_1}(\zee)}} \geq 1-(1-\beta) = \beta.
		\end{align*}
		By averaging, we can use $|\Tau| \leq k^k$ to conclude that there is a $\tau^* \in \Tau \setminus \{\tau_1\}$ such that,
		\[
			\tildeEx{\Ex{\li}{\Lambda_{\li,\tau^*}(\zee)}} ~\geq~ \frac{1}{|\Tau|} \sum_{\tau \in \Tau \setminus \{\tau_1\}} \tildeEx{\Ex{\li}{\Lambda_{\li,\tau}(\zee)}} ~\geq~ \frac{\beta}{k^k}.\qedhere
		\]
	\end{proof}
	
	\paragraph{Capturing common error locations}
	Suppose the partition $\tau^*$ is given by $[k] = (K_1, \cdots, K_p)$ for some $1<p<k$.  Henceforth, we will be working with this fixed partition. 	We define a function, $\dd_{S,\ri}(g,\zee)$, which is an indicator of whether $\ri \in R$ is a common error location for the pseudocodewords indexed by $S\subseteq [k]$. 
%	Let us define two shorthands for some $S\subseteq [k]$:
	\begin{align*}
		\dd_{S,\ri}(g, \zee) &= \Pi_{i\in S} \indi{g_{\ri} \neq \zee_{i,\ri}} \\
		\dd_{S,e}(g, \zee) &= \dd_{S,\ri}(g, \zee), \text{ where } e=(\li,\ri).
	\end{align*}
Of course, if a vertex $\ri $ is a common error location for all $k$ pseudocodewords, then it is also a common error location for a subset, implying
	\[
		\dd_{S,\ri}(g,\zee) ~\geq~ \dd_{[k],\ri}(g,\zee).
	\]
	We will also need to track the error locations common to $S$ but not to $[k]$, so we define two additional shorthands:
	\begin{align*}
		\fd_{S,\ri}(g, \zee) ~&=~ \dd_{S,\ri}(g,\zee) ~-~ \dd_{[k],\ri}(g,\zee) \\
		\fd_{S,e}(g, \zee) ~&=~ \fd_{S,r}(g, \zee), \text{ where } e=(\li,\ri).
	\end{align*}

\paragraph{Key Claims}		
		We will now prove SoS versions of the three main claims from the integral proof; ~\cref{lemma:inductive}, \cref{claim:local-bound}, and \cref{claim:sampling_erasure}. 
	%To begin, let us define the SoS generalizations of the terms in the lemma. 

%Finally, we use that the local snapshots of common errors are upper bounded by the global common errors. This will be another AEL argument.

%	The following table illustrates how the above functions capture...  
%\tnote{Suppresing the dictionary for now.}
%		\begin{table}[h]
%\begin{center}
%  \begin{tabular}{l|c}
%   Codeword & SoS Generalization \\
%    \hline\\
%   $\Delta(g_\ell, \fjl)$ & $\tildeEx{\Ex{e\in N(\li)}{\fd_{K_j,e}(g,\zee)}}$ \\
%    $\indi{\ell\in L^*}$  & $\tildeEx{\Lambda_{\li,\tau^*}(\zee)}$\\
%   	$\abs{L^*}/n$ & $\tildeEx{\Ex{\li}{\Lambda_{\li,\tau^*}(\zee)}}$  \\[2pt]
%   	$\Ex{\ell \in L^*}{\Delta(\gl, \fjl)}$ & $\frac{\Ex{\li}{\tildeEx{\Lambda_{\li,\tau^*}(\zee) \cdot \Ex{e\in N(\li)}{\fd_{K_j,e}(g,\zee)}}} }{\Ex{\li}{\tildeEx{\Lambda_{\li,\tau^*}(\zee)}}}$
%  \end{tabular}
%  \caption{Dictionary between}
%\end{center}
%\end{table}

The first of these showed that errors observed on $L^*$ serve as a lower bound for global errors on the right. The analog of "averaging over $L^*$ can be carried out by reweighing according to the indicator function $\Lambda_{\li,\tau^*}$ for the $\tau^*$ partition.
	
\begin{lemma}[Generalization of \cref{claim:sampling_erasure}]\label{lemma:local_erasures_upper_bound} 
%For any functions $\Gamma(\ell, \zee), \Psi(\zee) $ of  pseudocodewords, the following holds for $\eta$-good pseudocodeword $\zee$.
%\[
%	\tildeEx{\Ex{\li \sim \ri}{\Gamma(\zee) \cdot \Psi(\zee)}} ~\leq~ \tildeEx{\Ex{\li}{\Gamma(\zee)}} \tildeEx{\Ex{\ri}{\Psi(\zee)}} + \lambda +\eta.			
%\]
Assume that $\lambda \leq \frac{\beta}{12k^{k}}\cdot \eps$ and $\eta \leq \frac{\beta}{12k^{k}}\cdot \eps$.
For the functions $\fd, \dd$, and any set $S\subseteq [k]$ we have,
\begin{align*}
	\frac{\tildeEx{\Ex{\li \sim \ri}{\Lambda_{\li,\tau^*}(\zee) \cdot \dd_{S,\ri}(g,\zee)}}}{\tildeEx{\Ex{\li}{\Lambda_{\li,\tau^*}(\zee)}}} ~\leq~ \tildeEx{\Ex{\ri}{\dd_{S,\ri}(g,\zee)}} + \frac{\eps}{6}.\\
		\frac{\tildeEx{\Ex{\li \sim \ri}{\Lambda_{\li,\tau^*}(\zee) \cdot \fd_{S,\ri}(g,\zee)}}}{\tildeEx{\Ex{\li}{\Lambda_{\li,\tau^*}(\zee)}}} ~\leq~ \tildeEx{\Ex{\ri}{\fd_{S,\ri}(g,\zee)}} + \frac{\eps}{6}.
		\end{align*}
	\end{lemma}
\begin{proof} The proof is based on an AEL-like argument and is identical for either case. The first step uses the expander mixing lemma for pseudocodewords (\cref{lem:eml_sos}), and the second utilizes the $\eta$-good property (\cref{lem:avg_corr}),
	\begin{align*}
%		\tildeEx{\Ex{\li}{\Lambda_{\li,\tau^*}(\zee) \cdot \Ex{e\in N(\li)}{D_{[k],e}(g,\zee)}}} &= 
		\tildeEx{\Ex{\li\sim \ri}{\Lambda_{\li,\tau^*}(\zee) \cdot \dd_{S,\ri}(g,\zee)}}	~&\leq~ \tildeEx{\Ex{\li , \ri}{\Lambda_{\li,\tau^*}(\zee) \cdot \dd_{S,\ri}(g,\zee)}} + \lambda \\
		~&\leq~ \tildeEx{\Ex{\li}{\Lambda_{\li,\tau^*}(\zee)}} \cdot \tildeEx{\Ex{\ri}{\dd_{S,\ri}(g,\zee)}} + \lambda + \eta.
	\end{align*}
	To obtain the final consequence we divide by $\tildeEx{\Ex{\li}{\Lambda_{\li,\tau^*}(\zee)}}$ and use~\cref{claim:best_partition}.
		\begin{align*}
		\frac{\tildeEx{\Ex{\li\sim \ri}{\Lambda_{\li,\tau^*}(\zee) \cdot \dd_{S,\ri}(g,\zee)}}}{\tildeEx{\Ex{\li}{\Lambda_{\li,\tau^*}(\zee)}}} ~&\leq~ \tildeEx{\Ex{\ri}{\dd_{S,\ri}(g,\zee)}} + \frac{\lambda + \eta}{\tildeEx{\Ex{\li}{\Lambda_{\li,\tau^*}(\zee)}}} \\
%		~&\leq  \tildeEx{\Ex{\ri}{D_{[k],\ri}(g,\zee)}} + \frac{\lambda + \eta}{(\beta/k^k)} \\
		~&\leq~  \tildeEx{\Ex{\ri}{\dd_{S,\ri}(g,\zee)}} + \frac{\eps}{6}. \qedhere
%		~&=~  \tildeEx{\Delta_R(g,\zee_{[k]})} + \frac{\eps}{6}.
	\end{align*}
\end{proof}

%	Finally, define $\lstar := \tildeEx{\Ex{\li \in L}{\Lambda_{\li, \tau^*}(\zee)}}$
%	\snote{Move these definitions outside theorem? $\fd$ is a macro, feel free to change.}

%	\begin{align*}
%	\Delta(\gl, \fjl) ~&\mapsto~~   \tildeEx{\Ex{e\in N(\li)}{\fd_{K_j,e}(g,\zee)}}\;,	\\
%	\indi{\ell\in L^*}  ~&\mapsto~~ \tildeEx{\Lambda_{\li,\tau^*}(\zee)}\;, \\
%	\abs{L^*}/n ~&\mapsto~~ \tildeEx{\Ex{\li}{\Lambda_{\li,\tau^*}(\zee)}}.
%	\end{align*}
	
%	It is easy to see that when $Z$ is an integral codeword, the two definitions coincide. Combining these, 
%	\begin{align*}
%	\Ex{\ell \in L^*}{\Delta(\gl, \fjl)} = \frac{\Ex{\ell}{\indi{\ell\in L^*} \cdot\Delta(\gl, \fjl)}}{|L^*|} ~&\mapsto~~   \frac{\Ex{\li}{\tildeEx{\Lambda_{\li,\tau^*}(\zee) \cdot \Ex{e\in N(\li)}{\fd_{K_j,e}(g,\zee)}}} }{\Ex{\li}{\tildeEx{\Lambda_{\li,\tau^*}(\zee)}}} .
%	\end{align*}
%	\snote{Maybe above is not a great idea since we want to discourage $\erase{g}$ in the SoS section.}
	
%	$\tildeEx{\Lambda_{\li,\tau^*}(\zee) \cdot \Ex{e\in N(\li)}{\indi{g_e \neq \zee_{i,e}} \cdot (1-D_{[k],e}(g,\zee))}} $
%	
%	$\Ex{\ell \in L^*}{\Delta(\gl, \fl_j)}$

The second lemma showed that the number of common error locations increases when only considering a subset of indices $K_j \sub [k]$, and this increase can be lower bounded in terms of the distance between $g_{\li}$ (actually, $\gl$) and the common local projections of $K_j$, averaged over $L^*$. The second term on the RHS in the lemma below is the analog of these errors between $\gl$ and common local projections $f_j$, averaged over $L^*$.
	\begin{lemma}[Generalization of  \cref{lemma:inductive}]\label{lemma:more_common_errors}
	For any part $K_j$ in the partition $\tau^*$, we have,
	%, and an arbitrary $i\in H_j$,
	\begin{align*}
		\tildeEx{\Delta_R(g,\zee_{K_j})}
%		&= \Ex{\ri}{\tildeEx{D_{K_j,\ri}(g,\zee)}}
%		&\geq \Ex{\ri}{\tildeEx{D_{[k],\ri}(g,\zee)}} + \frac{\Ex{\li}{\tildeEx{\Lambda_{\li,\tau^*}(\zee) \cdot \Ex{e\in N(\li)}{\fd_{K_j,e}(g,\zee)} }}}{\Ex{\li}{\tildeEx{\Lambda_{\li,\tau^*}(\zee)}}} - \frac{\eps}{6}\\
		~\geq~ \tildeEx{\Delta_R(g,\zee_{[k]})}  + \frac{\tildeEx{\Ex{\li \sim \ri}{\Lambda_{\li,\tau^*}(\zee) \cdot \fd_{K_j,\ri}(g,\zee)}}}{\tildeEx{\Ex{\li}{\Lambda_{\li,\tau^*}(\zee)}}} - \frac{\eps}{6} .
	\end{align*}
	\end{lemma}
	
	\begin{proof} By using the definition of the functions $\dd, \fd$ and some rearragement, we get,
	\begin{align*}
		\tildeEx{\Delta_R(g,\zee_{K_j})} ~&=~ \Ex{\ri}{\tildeEx{\dd_{K_j,\ri}(g,\zee)}} \\
		~&=~ \Ex{\ri}{\tildeEx{\dd_{[k],\ri}(g,\zee) \cdot \dd_{K_j,\ri}(g,\zee) + (1-\dd_{[k],\ri}(g,\zee)) \cdot \dd_{K_j,\ri}(g,\zee)}} \\
		~&=~ \Ex{\ri}{\tildeEx{\dd_{[k],\ri}(g,\zee)}} + \Ex{\ri}{\tildeEx{\fd_{K_j,\ri}(g,\zee)}}.
	\end{align*}
	
The proof finishes by replacing the second term by the bound from~\cref{lemma:local_erasures_upper_bound}.
%a typical AEL argument. The first step uses the expander mixing lemma (\cref{lem:eml_sos}), and the second utilizes the $\eta$-good property,
%	\begin{align*}
%		\tildeEx{\Ex{\li \sim \ri}{\Lambda_{\li, \tau^*}(\zee) \cdot \fd_{K_j,\ri}(g,\zee)}} \leq~&~ \tildeEx{\Ex{\li, \ri}{\Lambda_{\li, \tau^*}(\zee) \cdot \fd_{K_j,\ri}(g,\zee)}} + \lambda \\
%		\leq~&~ \tildeEx{\Ex{\li}{\Lambda_{\li, \tau^*}(\zee)}} \cdot \tildeEx{\Ex{\ri}{\fd_{K_j,\ri}(g,\zee)}} + \lambda + \eta\\
%		=~&~ \lstar \cdot \tildeEx{\Ex{\ri}{\fd_{K_j,\ri}(g,\zee)}} + \lambda + \eta.
%	\end{align*}
%	
%Dividing both sides by $\lstar$, we get 
%%\cref{eq:more_common_errors_lower_bound} and \cref{eq:more_common_errors_upper_bound}, we get
%	\begin{align*}
%		\tildeEx{\Ex{\ri}{\fd_{K_j,\ri}(g,\zee)}} &~\geq~ \frac{1}{\lstar}\cdot {\tildeEx{\Ex{\li \sim \ri}{\Lambda_{\li, \tau^*}(\zee) \cdot \fd_{K_j,\ri}(g,\zee)}}}- \frac{\lambda+\eta}{\lstar} \\
%		&~\geq~ \frac{1}{\lstar}\cdot {\tildeEx{\Ex{\li \sim \ri}{\Lambda_{\li, \tau^*}(\zee) \cdot \fd_{K_j,\ri}(g,\zee)}}}- \frac{\eps}{6}.\qedhere
%	\end{align*}
\end{proof}

	Finally, we state the local inequality needed from the inner code. Since this inequality is only valid for vertices in $L^*$, we multiply the required equation by indicator $\Lambda_{\li,\tau^*}(\zee)$ so that it is trivial when the indicator is 0. Subject to this indicator being 1, the inequality says that the generalized Singleton bound with erasures holds for the inner code. Note that this bound holds for each vertex in $L^*$ unlike previous sampling lemmas which only work in an average sense over $L^*$.
%		\tnote{Shashank:Check}
	\begin{lemma}[Generalization of \cref{claim:local-bound}]\label{lemma:local_avg_singleton}
		For every $\li \in L$ and for every $j\in [p]$,
		\[
			\Lambda_{\li, \tau^*}(\zee) \parens[\bigg]{ \sum_{j = 1}^p \Ex{e\in N(\li)}{\fd_{K_j, e}(g,\zee)}} ~\geq~ \Lambda_{\li, \tau^*}(\zee) \cdot (p-1) \parens[\Big]{ \delta_0 - \Ex{e\in N(\li)}{\dd_{[k],e}(g,\zee)}  - \frac{\eps}{2}}.
		\]
	\end{lemma}
	\begin{proof}
	Let $M_{\li} \sub E$ be the union of edge neighborhoods over vertices in $R$ that are adjacent to $\li$. That is, 
	\[
		M_{\li} = \bigcup_{\ri 
		\sim \li}  N(\ri)
	\]
	Let us consider the following two local functions that depend on $[k] \times M_{\li}$.
	\begin{align*}
		\mu_1(\cdot) ~&=~ \Lambda_{\li, \tau^*}(\cdot) \;\parens[\bigg]{\; \sum_{j\in [p]} \Ex{e\in N(\li)}{\fd_{K_j, e}(g,\cdot)}\,} \\
		\mu_2(\cdot) ~&=~ \Lambda_{\li, \tau^*}(\cdot) (p-1) \inparen{ \delta_0 - \Ex{e\in N(\li)}{\dd_{[k],e}(g,\cdot)}  - \frac{\eps}{2}}
	\end{align*}
	We will prove this inequality pointwise by showing that for any $\alpha \in \Sigma_{\inn}^{[k] \times M_{\li}}$, $\mu_1(\alpha) \geq \mu_2(\alpha)$.
	
	If $\alpha$ is such that $\Lambda_{\li,\tau^*}(\alpha) = 0$, then the inequality is trivially true. Henceforth, we assume that $\Lambda_{\li,\tau^*}(\alpha) = 1$. This means that $(\alpha_{1,\li},\alpha_{2,\li},\cdots ,\alpha_{k,\li})$ induces the partition $\tau^*$. 
	By definition of $K_j$, for any $i,i'\in K_j$, we get that $\alpha_{i,\li} = \alpha_{i',\li}$. Let us denote this common codeword in $\calC_{\inn}$ as $\fjl$.
	
	Let $g_{\li} \in \Sigma_{\inn}^{d}$ be the local projection of $g$ to the edge neighborhood of $\li$. For every $e\in N(\li)$ that satisfies $\dd_{[k],e}(g,\alpha) = 1$, we replace the corresponding coordinate in $g_{\li}$ by an erasure to obtain $\gl\in \inparen{ \Sigma_{\inn} \cup \{\bot\} }^{N(\li)}$. 
	The fraction of erasures in $g_{\li}$ is
	\[
		s_{\li} = \Ex{e\in N(\li)}{\dd_{[k],e}(g,\alpha)}.
	\]
	Next, we calculate,
	\begin{align*}
		\Ex{e\in N(\li)}{\fd_{K_j, e}(g,\alpha)} ~&=~ \Ex{e\in N(\li)}{\dd_{K_j, e}(g,\alpha) \cdot \inparen{1-\dd_{[k],e}(g,\alpha)}} \\
		~&=~ \Ex{e\in N(\li)}{\dd_{K_j, e}(g,\alpha)} - \Ex{e\in N(\li)}{\dd_{[k],e}(g,\alpha)} \\
		~&=~ \Ex{e\in N(\li)}{\indi{g_{e} \neq f_{j,e}}} - s_{\li} \\
		~&=~ \Delta(g_{\li},\fjl) - s_{\li} \\
		~&=~ \Delta(\gl, \fjl).
	\end{align*}
	Applying the $(\delta_0,k,\eps/2)$-average radius list decodability with erasures of inner code to $\gl$ and the set of codewords $\{\fjl\}_{j\in [p]}$, we get,
	\[
		\sum_{j=1}^p \Delta(\gl,\fjl) ~\geq~ (p-1) \parens[\Big]{\delta_0 - s_{\li} - \frac{\eps}{2}}.
	\]
	Substituting $\Delta(\gl, \fjl) = \Ex{e\in N(\li)}{\fd_{K_j, e}(g,\alpha)}$ and $s_{\li} = \Ex{e\in N(\li)}{\dd_{[k],e}(g,\alpha)}$ gives
	\begin{align*}
		\sum_{j=1}^p \Ex{e\in N(\li)}{\fd_{K_j, e}(g,\alpha)} ~\geq~ (p-1) \cdot\parens[\bigg]{\delta_0 -  \Ex{e\in N(\li)}{\dd_{[k],e}(g,\alpha)} - \frac{\eps}{2}}.\qedhere
	\end{align*}
\end{proof}

\subsubsection{Proof of Main Theorem}	

We restate the main theorem and finish the proof using the above lemmas.

\begin{theorem}[Restatement of \cref{thm:sos_main}]
	Let $k_0\geq 1$ be an integer and let $\eps > 0$. Let $\AELC$ be a
code obtained using the AEL construction using $(G,\calC_{\out}, \calC_{\inn})$, where $\calC_{\inn}$ is $(\delta_0, k_0, \eps/2)$ average-radius list decodable with erasures, and the graph $G$ is a $(n,d,\lambda)$-expander. 
	Let $\tildeEx{\cdot}$ be an $\eta$-good $k_0$-tuple of pseudocodewords that satisfies the following pairwise distance property on the left:
	\[
		\forall i,i'\in [k_0] \text{ with } i\neq i', \qquad \tildeEx{\Delta_L(\zee_i, \zee_{i'})} \geq \beta\mper
	\]
	Further, assume that $\lambda \leq \frac{\beta}{12k_0^{k_0}}\cdot \eps$ and $\eta \leq \frac{\beta}{12k_0^{k_0}}\cdot \eps$.	Then, for any $g \in (\F_q^d)^R$,
	\begin{align*}
	\sum_{i\in [k_0]} \tildeEx{\Delta_R(g,\zee_i)} ~&\geq~ (k_0-1)(\delta_0-\eps) + \Ex{\ri \in R}{\tildeEx{ \Pi_{i\in[k_0]} \indi{g_{\ri} \neq \zee_{i,\ri}} }} \\
	~&=~ (k_0-1)(\delta_0 -\eps) + \tildeEx{\Delta_R(g,\zee_{[k_0]})}.
%		\Ex{i\in [k]}{\tildeEx{\Delta_R(g,\zee_i)}} \geq \frac{k-1}{k} \inparen{ \Delta - \Ex{\ri \in R}{\tildeEx{ \Pi_{i\in[k]} \indi{g_{\ri} \neq \zee_{i,\ri}} }} }
	\end{align*}
\end{theorem}

\begin{proof}
The proof, as before, is by induction on $k_0$. The base case, $k_0 =1$ is trivial. So we assume the statement holds upto $k-1$, and the goal is to prove it for $k_0=k$. Fix a partition $\tau^* = (K_1,\cdots, K_p)$ as guaranteed by~\cref{claim:best_partition}.  For each part $K_j$, we may apply the induction hypothesis to the sub-tuple defined by it as the pairwise distance property is already assumed. Using this we get, 
	\begin{align*}
\sum_{i=1}^k \tildeEx{\Delta_R(g,\zee_i)} 
		~&=~ \sum_{j=1}^p \sum_{i\in K_j} \tildeEx{\Delta_R(g,\zee_i)} \\
		~&\geq~ \sum_{j=1}^p \inparen{(|K_j|-1)(\delta_0 -\eps) + \tildeEx{\Delta_R(g,\zee_{K_j})}}.
		\end{align*}

The first term is simply,   
$
\sum_{j=1}^p (|K_j|-1)(\delta_0 -\eps) ~=~ (k-p)(\delta_0 -\eps) . 
$ The goal is now to show that, 
\[\sum_{j=1}^p\tildeEx{\Delta_R(g,\zee_{K_j})}	 ~\geq~ (p-1)(\delta_0 -\eps) + \tildeEx{\Delta_R(g,\zee_{[k]})}.\]
 
 For a fixed $j \in [p]$, we apply \cref{lemma:more_common_errors} to obtain, 

		\[
		\tildeEx{\Delta_R(g,\zee_{K_j})}	~\geq~  \tildeEx{\Delta_R(g,\zee_{[k]})}  + \frac{\Ex{\li \sim \ri}{\;\tildeEx{\Lambda_{\li,\tau^*}(\zee) \cdot \fd_{K_j,\ri}(g,\zee)}\,} }{\Ex{\li}{\tildeEx{\Lambda_{\li,\tau^*}(\zee)}}} - \frac{\eps}{6}.
		\]

The term $c := \tildeEx{\Delta_R(g,\zee_{[k]})}  - \frac{\eps}{6}$ is independent of $j$.
Summing the RHS over $j \in [p]$, 

\begin{align}
\sum_{j=1}^p	\tildeEx{\Delta_R(g,\zee_{K_j})} - p\cdot c ~&\geq~ \sum_{j=1}^p \frac{\Ex{\li \sim \ri}{\tildeEx{\Lambda_{\li,\tau^*}(\zee) \cdot \fd_{K_j,\ri}(g,\zee)}} }{\Ex{\li}{\tildeEx{\Lambda_{\li,\tau^*}(\zee)}}}\\
%~&=~ \sum_{j=1}^p \frac{\Ex{\li}{\tildeEx{\Lambda_{\li,\tau^*}(\zee) \cdot \inparen{ \Ex{e\in N(\li)}{ \fd_{K_j,e}(g,\zee)}}}} }{\Ex{\li}{\tildeEx{\Lambda_{\li,\tau^*}(\zee)}}}\\
~&=~   { \frac{\Ex{\li}{\tildeEx{\Lambda_{\li,\tau^*}(\zee) \cdot \sum_{j=1}^p \inparen{\Ex{e\in N(\li)}{\fd_{K_j,e}(g,\zee)}} }}}{ \Ex{\li}{\tildeEx{\Lambda_{\li,\tau^*}(\zee)}} } }\label{eq:main}
		\end{align}

We can now apply \cref{lemma:local_avg_singleton} to the RHS
\begin{align}		
	~&\geq~ { \frac{\Ex{\li}{\tildeEx{\Lambda_{\li, \tau^*}(\zee) \cdot (p-1) \inparen{ \delta_0 - \Ex{e\in N(\li)}{\dd_{[k],e}(g,\zee)}  - \frac{\eps}{2}} }}}{\Ex{\li}{\tildeEx{\Lambda_{\li,\tau^*}(\zee)}}}} \\
	~&=~	(p-1)(\delta_0 - \frac{\eps}{2}) -(p-1) \frac{\Ex{\li}{\tildeEx{\Lambda_{\li, \tau^*}(\zee) \cdot \inparen{\Ex{e\in N(\li)}{\dd_{[k],e}(g,\zee)} } }}}{\Ex{\li}{\tildeEx{\Lambda_{\li,\tau^*}(\zee)}}}.\label{eq:second}
		 \end{align}

To bound the term on the right, we use~\cref{lemma:local_erasures_upper_bound},			\begin{align*}
				\frac{\Ex{\li}{\tildeEx{\Lambda_{\li, \tau^*}(\zee) \cdot \inparen{\Ex{e\in N(\li)}{\dd_{[k],e}(g,\zee)} } }}}{\Ex{\li}{\tildeEx{\Lambda_{\li,\tau^*}(\zee)}}} ~&=~ \frac{\Ex{\li\sim \ri}{\tildeEx{\Lambda_{\li, \tau^*}(\zee) \cdot \dd_{[k],\ri}(g,\zee)  }}}{\Ex{\li}{\tildeEx{\Lambda_{\li,\tau^*}(\zee)}}}\\[1em]
\text{(\cref{lemma:local_erasures_upper_bound})}\;\;				~&\leq~  \tildeEx{\Ex{\ri}{\dd_{S,\ri}(g,\zee)}} + \frac{\eps}{6}\\
				~&=~	\tildeEx{\Delta_R(g,\zee_{[k]})} + \frac{\eps}{6}.  
				\end{align*}

Plugging this bound in~\cref{eq:second} and then back in~\cref{eq:main},  we get,		
\begin{align*}	
		\tildeEx{\Delta_R(g,\zee_{K_j})}  ~&\geq~ p\cdot c + (p-1)\parens[\Big]{\delta_0 - \frac{\eps}{2}} -  (p-1) \inparen{\tildeEx{\Delta_R(g,\zee_{[k]})} + \frac{\eps}{6}}\\
		~&=~ \tildeEx{\Delta_R(g,\zee_{[k]})}  +(p-1)\parens[\Big]{\delta_0 - \frac{\eps}{2}} - (2p-1)\cdot \frac{\eps}{6} \\
		~&\geq~ \tildeEx{\Delta_R(g,\zee_{[k]})}  + (p-1)(\delta_0 - \eps) \;. \qedhere
		\end{align*}
%	\begin{align*}
%		&= (k-1)(\Delta - \eps) + \Delta_R(g,\zee_{[k]}) +(p-1)\frac{\eps}{2} - (2p-1)\cdot k^k\inparen{\frac{\lambda+\eta}{\beta}} \\
%		&\geq (k-1)(\Delta - \eps) + \Delta_R(g,\zee_{[k]})
%	\end{align*}
\end{proof}

%\paragraph{Proof of the three lemmas}

\subsection{The final algorithm}\label{sec:sos_algo}

In this subsection, we describe the final decoding algorithm. We will mainly rely on the main result of the previous subsection: that for appropriately instantiated AEL codes and for $\eta$ small enough, $\eta$-good $k$-tuple of pseudocodewords satisfy the generalized Singleton bound.

\subsubsection{Decoding from Distributions}
Before describing our main algorithm, we argue that a simple idea based on randomized rounding can be used to extend the unique decoder of $\calC_{\out}$ to decode not only from integral strings close to a codeword, but an ensemble of distributions - one for each coordinate - that is close to a codeword in average sense. It is also standard to derandomize this process via threshold rounding, which is what we describe next. We will be needing this strengthening of the unique decoder of $\calC_{\out}$ as a subroutine in the main algorithm.

It will be helpful to index the codewords in $\calC_{\inn}$ by integers. Let $\abs*{\calC_{\inn}} = M$, and let its codewords be $\alpha_1,\alpha_2,\cdots ,\alpha_M$. Recall that for any $k$-tuple of pseudocodewords $\tildeEx{\cdot}$, for any $j\in [k]$ and $\li \in L$, the set of values
\[
	\inbraces{ \tildeEx{\indi{\zee_{j,\li} = \alpha_i}} }_{i\in [M]}
\]
correspond to a probability distribution over $\calC_{\inn}$. Therefore, the set of intervals
\[
	\inbraces{ \left[\sum_{i=1}^{m-1} \tildeEx{\indi{\zee_{j,\li} = \alpha_i}}, \sum_{i=1}^{m} \tildeEx{\indi{\zee_{j,\li} = \alpha_i}} \right) }_{m\in [M]}
\]
partitions the interval $[0,1)$ into at most $M$ parts.

\begin{lemma}\label{lem:decode_from_distrib}
	Suppose the code $\calC_{\out}$ is unique decodable up to radius $\delta_{\out}^{\dec}$, and let $h$ be a codeword in $\AELC$. Given a collection of distributions $\{\calD_{\li}\}_{\li \in L}$, with each distribution over $\calC_{\inn}$ such that
	\[
		\Ex{\li}{\Ex{f\sim \calD_{\li}}{{\indi{f\neq h_{\li}}}}} \leq \delta_{\out}^{\dec},
	\]
	the \cref{algo:unique-decoding} finds $h$.
\end{lemma}

\begin{figure}[!ht]
\begin{algorithm}{\DECODE}{$\{\calD_{\li}\}_{\li\in L}$}{Codeword $h\in \AELC$ such that $\Ex{\li}{\Ex{f\sim \calD_{\li}}{\indi{f\neq h_{\li}}}} \leq \delta_{\out}^{\dec}$}
\label{algo:unique-decoding}
\begin{itemize}
%
%	\item For $p=1$ to $k$:\label{step2}
%	\begin{enumerate}[(i)]
%%		\item Find a $p$-tuple of pseudocodewords $\pcod{p}{\cdot}$ of degree $t = \cdots$ that respects the following constraints:
%%		\begin{itemize}
%%			\item For every $i,i' \in [p]$ with $i\neq i'$, $~\pcod{p}{\Delta_L(\zee_i,\zee_{i'})} > \delta_{\out}^{\dec}$.
%%			\item For every $i\in [p]$, $\pcod{p}{\Delta_R(g,\zee_i)} \leq \frac{k-1}{k}(1-\rho-\eps)$.
%%		\end{itemize}
%%		\item If no such $\pcod{p}{\cdot}$ exists:
%%		\begin{itemize}
%%		\item $p^* = p-1$
%%		\item Exit loop.
%%		\end{itemize}
%	\end{enumerate}
	\item Let $w_{\li,j}$ be the weight on codeword $\alpha_i$ according to distribution $\calD_{\li}$.
	\item For every threshold $\theta \in [0,1]$:\footnote{As written, this involves trying uncountably many thresholds. However, we only need to try at most $M \cdot |L|$ thresholds since the algorithm only depends on which intervals $\theta$ belongs to. Since the number of endpoints of these intervals is at most $M\cdot |L|$ in total, it suffices to try only $M\cdot |L| = M\cdot n$ many distinct thresholds. This is a standard method called threshold rounding.}
			\begin{enumerate}[(i)]
				\item Construct an $f_{\theta} \in \calC_{\inn}^L$ by assigning 
				\[(f_{\theta})_{\li} = \alpha_m  \iff \theta \in \left[ \sum_{i=1}^{m-1} w_{\li, i}, \sum_{i=1}^{m} w_{\li,i} \right)\]
				for every $\li \in L$.
				\item Let $f_\theta^* \in \Sigma_{\out}^L$ defined as $(f_{\theta}^*)_{\li} = \phi^{-1}((f_{\theta})_{\li})$. 
				\item Use the unique decoder of $\calC_{\out}$ to find an $h^* \in \calC_{\out}$ whose distance from $f_{\theta}^*$ is at most $\delta_{\out}^{\dec}$, if such an $h^*$ exists. That is, $h^* \leftarrow \mathrm{Dec}_{\calC_{\out}}(f_{\theta}^*,\delta_{\out}^{\dec})$.
				\item Let $h\in \AELC$ be the codeword corresponding to $h^* \in \calC_{\out}$. Return $h$.
				\end{enumerate}
		\end{itemize}
%		\item For every threshold $\theta \in [0,1]$:\footnote{As written, this involves trying uncountably many thresholds. However, we only need to try at most $M \cdot |L|$ thresholds since the algorithm only depends on which intervals $\theta$ belongs to. Since the number of thresholds is at most $M\cdot |L|$, it suffices to try only $M\cdot |L|$ many distinct thresholds. This is a standard method called threshold rounding.}
%			\begin{enumerate}[(i)]
%				\item Construct an $f_{\theta} \in \calC_{\inn}^L$ by assigning 
%				\[f_{\li} = \alpha_m  \iff \theta \in \left[ \sum_{i=1}^{m-1} \pcod{p^*}{\indi{\zee_{j,\li} = \alpha_i}}, \sum_{i=1}^{m} \pcod{p^*}{\indi{\zee_{j,\li} = \alpha_i}} \right)\]
%				for every $\li \in L$.
%				\item Let $f_\theta^* \in \Sigma_{\out}^L$ defined as $(f_{\theta}^*)_{\li} = \phi^{-1}((f_{\theta})_{\li})$. 
%				\item Use the unique decoder of $\calC_{\out}$ to find an $h^* \in \calC_{\out}$ whose distance from $f_{\theta}^*$ is at most $\delta_{\out}^{\dec}$, if such an $h^*$ exists. That is, $h^* \leftarrow \mathrm{Dec}(f_{\theta}^*,\delta_{\out}^{\dec})$.
%				\item Let $h$ be the codeword of $\AELC$ corresponding to $h^* \in \calC_{\out}$. If $\Delta_R(g,h) < \frac{k-1}{k} (1-\rho-\eps)$, add $h$ to $\calL$.
%			\end{enumerate}
%
\vspace{5pt}
\end{algorithm}
\end{figure}

\begin{proof}
	It suffices to show that there exists a threshold $\theta \in [0,1)$ for which the distance between $f_{\theta}^*$ constructed by \cref{algo:unique-decoding} and $h^*$ is at most $\delta_{\out}^{\dec}$. In fact, we will show that this is true for an average $\theta$.
	\begin{align*}
		\Ex{\theta \in [0,1)}{ \Delta(f_{\theta}^*,h^*)} &= \Ex{\theta \in [0,1)}{ \Delta_L(f_{\theta},h)} = \Ex{\theta \in [0,1)}{ \Ex{\li}{\indi{(f_{\theta})_{\li} \neq h_{\li}}}}  = \Ex{\theta \in [0,1)}{ \Ex{\li}{\sum_{i\in [M] : \alpha_i \neq h_{\li}}\indi{(f_{\theta})_{\li} = \alpha_i}}} 
	\end{align*}
	We can move the expectation over $\theta$ inside and use the fact that $(f_{\theta})_{\li}$ is $\alpha_i$ with probability exactly $w_{\li,i}$ to get
	\begin{align*}
		\Ex{\li}{\sum_{i\in [M] : \alpha_i \neq h_{\li}} \Ex{\theta \in [0,1)}{\indi{(f_{\theta})_{\li} = \alpha_i}}} =  \Ex{\li}{\sum_{i\in [M] : \alpha_i \neq h_{\li}} w_{\li,i}} = \Ex{\li}{\Ex{f\sim \calD_{\li}}{{\indi{f\neq h_{\li}}}}} \leq \delta_{\out}^{\dec}. \qedhere
	\end{align*}
\end{proof}

\subsubsection{Decoding Algorithm}
\begin{table}[h]
\hrule
\vline
\begin{minipage}[t]{0.99\linewidth}
\vspace{-5 pt}
{\small
\begin{align*}
    &\mbox{find}\quad ~~ \tildeEx{\cdot} ~\text{on}~ \zee_{[p],E} \text{ and alphabet }\Sigma_{\inn}%\tag{List Decoding Program}\label{sos:list_dec}
    \\
&\mbox{subject to}\quad \quad ~\\
	&\qquad \text{(i)}~~ \tildeEx{\cdot} \text{is a }p\text{-tuple of pseudocodewords of SoS-degree } t \\
    &\qquad \text{(ii)}~~ \forall i \in [p],~ \text{the constraint }\Delta_R(g,\zee_i) < \frac{k-1}{k}(1-\rho-\eps) \text{ is respected by }\tildeEx{\cdot}\label{cons:agreement-ld}    \\
&\qquad \text{(iii)}~ \forall F \sub E \text{ such that } |F|\leq \frac{t-2pd}{2}, \sigma \in \Sigma_{\inn}^{[p]\times F} \text{ and }\forall~i,i' \in [p] \text{ with } i\neq i', \\
&\qquad \qquad \qquad \qquad \tildeEx{\inparen{\Delta_L(\zee_i,\zee_{i'}) - \delta_{\out}^{\dec}}\cdot  \indi{\zee_{[p],F} = \sigma}^2 } \geq 0
\end{align*}}
\vspace{-10 pt}
\end{minipage}
\hfill\vline
\hrule
\caption{$\mathrm{SDP}(p, t)$}
\label{tab:SDP_for_feasibility}
\end{table}

\begin{theorem}\label{thm:sos_technical}
Let $k\geq 1$ be an integer and let $\eps > 0$. Let $\AELC$ be a
code obtained using the AEL construction using  $(G, \calC_{\out}, \calC_{\inn})$, where $\calC_{\inn}$ is $(\delta_0, k, \eps/2)$ average-radius list decodable with erasures, and $G$ is a $(n,d,\lambda)$-expander. 
Suppose that $\calC_{\out}$ is unique decodable from radius $\delta_{\out}^{\dec}$ in time $\calT(n)$.

If $\lambda \leq \frac{\delta_{\out}^{\dec}}{12{k}^{k}} \cdot \eps$, then there is a deterministic algorithm that takes as input $g\in (\Sigma_{\inn}^d)^R$, runs in time $\calT(n) + n^{O\inparen{ \frac{d \cdot k^{3k}|\Sigma_{\inn}|^{3kd}}{(\delta_{\out}^{\dec})^2\cdot \eps^2}}}$, and outputs the list $\calL(g,\frac{k-1}{k}(\delta_0-\eps))$.
\end{theorem}

\begin{figure}[!ht]
\begin{algorithm}{List Decoding algorithm up to $\frac{k-1}{k}(\delta_0-\eps)$}{$k$, $g \in (\Sigma_{\inn}^d)^R$}{List of codewords $\calL\inparen{g, \frac{k-1}{k}\inparen{\delta_0-\eps}}$}\label{algo:sos-decoding}
\begin{enumerate}
	\item Initialize $\calL = \{ \}$, $t = 2kd\cdot \inparen{1+\frac{144k^{2k}|\Sigma_{\inn}|^{3kd}}{(\delta_{\out}^{\dec})^2\cdot \eps^2}}$.
	\item Let $p^*$ be the largest $p\in [k]$ such that $\mathrm{SDP}(p,t)$ is feasible.
%	\item For $p=1$ to $k$:\label{step2}
%	\begin{enumerate}[(i)]
%%		\item Find a $p$-tuple of pseudocodewords $\pcod{p}{\cdot}$ of degree $t = \cdots$ that respects the following constraints:
%%		\begin{itemize}
%%			\item For every $i,i' \in [p]$ with $i\neq i'$, $~\pcod{p}{\Delta_L(\zee_i,\zee_{i'})} > \delta_{\out}^{\dec}$.
%%			\item For every $i\in [p]$, $\pcod{p}{\Delta_R(g,\zee_i)} \leq \frac{k-1}{k}(1-\rho-\eps)$.
%%		\end{itemize}
%%		\item If no such $\pcod{p}{\cdot}$ exists:
%%		\begin{itemize}
%%		\item $p^* = p-1$
%%		\item Exit loop.
%%		\end{itemize}
%	\end{enumerate}
	\item For every $F \sub E$ with $|F|\leq \frac{t-2kd}{2}$ and $\sigma\in \Sigma_{\inn}^F$ with $\pcod{p^*}{\indi{\zee_{[p^*],F} = \sigma}} > 0$:\label{step3}
		\begin{itemize}
		\item For $j=1$ to $p^*$:
		\begin{itemize}
			\item Construct an ensemble of distributions $\calD = \{\calD_{\li}\}_{\li\in L}$, with each $\calD_{\li}$ over $\calC_{\inn}$ by using the local distribution induced by the conditioned $\pcod{p^*}{ ~\cdot \given \zee_{[p^*],F} = \sigma}$ over the set $\{j\} \times N(\li)$, or equivalently, over variables $\zee_{j,\li}$.
			\item $h\leftarrow \text{\DECODE} \inparen{\calD}$.
			\item If $\Delta_R(g,h) < \frac{k-1}{k} (\delta_0-\eps)$, add $h$ to $\calL$.
		\end{itemize}
%		\item For every threshold $\theta \in [0,1]$:\footnote{As written, this involves trying uncountably many thresholds. However, we only need to try at most $M \cdot |L|$ thresholds since the algorithm only depends on which intervals $\theta$ belongs to. Since the number of thresholds is at most $M\cdot |L|$, it suffices to try only $M\cdot |L|$ many distinct thresholds. This is a standard method called threshold rounding.}
%			\begin{enumerate}[(i)]
%				\item Construct an $f_{\theta} \in \calC_{\inn}^L$ by assigning 
%				\[f_{\li} = \alpha_m  \iff \theta \in \left[ \sum_{i=1}^{m-1} \pcod{p^*}{\indi{\zee_{j,\li} = \alpha_i}}, \sum_{i=1}^{m} \pcod{p^*}{\indi{\zee_{j,\li} = \alpha_i}} \right)\]
%				for every $\li \in L$.
%				\item Let $f_\theta^* \in \Sigma_{\out}^L$ defined as $(f_{\theta}^*)_{\li} = \phi^{-1}((f_{\theta})_{\li})$. 
%				\item Use the unique decoder of $\calC_{\out}$ to find an $h^* \in \calC_{\out}$ whose distance from $f_{\theta}^*$ is at most $\delta_{\out}^{\dec}$, if such an $h^*$ exists. That is, $h^* \leftarrow \mathrm{Dec}(f_{\theta}^*,\delta_{\out}^{\dec})$.
%				\item Let $h$ be the codeword of $\AELC$ corresponding to $h^* \in \calC_{\out}$. If $\Delta_R(g,h) < \frac{k-1}{k} (1-\rho-\eps)$, add $h$ to $\calL$.
%			\end{enumerate}
		\end{itemize}
	\item Return $\calL$.
\end{enumerate}
\vspace{5pt}
\end{algorithm}
\end{figure}

\begin{proof}
	\cref{algo:sos-decoding} describes this algorithm. In the rest of the proof, we argue the correctness of this algorithm.
	
	We first start with a lemma that readily follows from \cref{thm:sos_main}.
	
	\begin{lemma}\label{lem:pseudocodeword_list_size}
	If $\lambda \leq  \frac{\delta_{\out}^{\dec}}{12k^k} \cdot \eps$, and $t \geq 2kd\cdot \inparen{1+\frac{144k^{2k}|\Sigma_{\inn}|^{3kd}}{(\delta_{\out}^{\dec})^2\cdot \eps^2}}$, then the $\mathrm{SDP}(k,t)$ is infeasible.
	\end{lemma}

	In other words, the lemma says that $p^* < k$, and that $\mathrm{SDP}(p^*+1,t)$ is infeasible. Let $\pcod{p^*}{\cdot}$ be the $p^*$-tuple of pseudocodewords found by $\mathrm{SDP}(p^*,t)$.
	\begin{lemma}\label{lem:covering}
	For any $h \in \calL\inparen{g,\frac{k-1}{k}(1-\rho-\eps)}$,
%	 and its corresponding $h^* \in \calC_{\out}$, 
there exists an $i\in [p^*]$, a set $F\sub E$ with $|F|\leq \frac{t-2kd}{2}$, and a $\sigma \in \Sigma_{\inn}^{[p^*]\times F}$ with $\pcod{p^*}{\indi{\zee_{[p^*],F} = \sigma}} > 0$, such that 
	\[
	\pcod{p^*}{\Delta_L(h, \zee_i) \given \zee_{[p^*],F} = \sigma} \leq \delta_{\out}^{\dec}.
	\] 
	\end{lemma}
	This lemma proves the correctness of the algorithm, since then $h$ will be added to $\calL$ in \hyperref[step3]{Step~\ref*{step3}}, by \cref{lem:decode_from_distrib}.
\end{proof}

\begin{proof}[Proof of \cref{lem:covering}]
	To show this, we will construct a $(p^*+1)$-tuple of pseudocodewords, say $\pcod{p^*+1}{\cdot}$, and use its infeasibility for $\mathrm{SDP}(p^*+1,t)$.
	
	Recall that $\pcod{p^*}{\cdot}$ is an SoS relaxation over the variables $\zee_{[p^*] \times E }$,
% = \{ Z_{i,e,s} \}_{i\in [p^*], e\in E,s \in \Sigma_{\inn}}$
while $\pcod{p^*+1}{\cdot}$ will be an SoS relaxation over the variables $\zee_{[p^*+1]\times E }$.
% = \{ Z_{i,e,s} \}_{i\in [p^*+1], e\in E,s\in \Sigma_{\inn}}$. 
	
	To describe the new $(p^*+1)$-tuple of pseudocodewords, we explicitly specify the corresponding pseudoexpectation operator $\pcod{p^*+1}{\cdot}$.
	Let $P$ be an arbitrary polynomial of degree $t$ over $\zee_{[p^*+1], E}$. We define a new polynomial $P_h$ over $\zee_{[p^*], E}$ by assigning the variables $\zee_{p^*+1,e}$ to be $h_e$ for every $e\in E$.
%	\begin{align*}
%		Z_{p^*+1,e,s} = \begin{cases}
%		1 \quad & h_e = s \\
%		0 \quad & h_e \neq s
%		\end{cases}
%	\end{align*}	
	Then, we define $\pcod{p^*+1}{\cdot}$ using the following
	\[
		\pcod{p^*+1}{P(\zee_{[p^*+1],E})} = \pcod{p^*}{P_h(\zee_{[p^*],E})}
	\]
	This is well-defined since the degree of $P_h$ cannot be more than $t$.
	
	By optimality of $p^*$, $\pcod{p^*+1}{\cdot}$ must be infeasible for $\mathrm{SDP}(p^*+1,t)$. It can be verified that $\pcod{p^*+1}{\cdot}$ is a valid $(p^*+1)$-tuple of pseudocodewords. Further, the constraints 
	\[
		\Delta_R(g,\zee_i) < \frac{k-1}{k}(\delta_0 - \eps)
	\]
	are respected for all $i\in [p^*]$. This is because 
	\[
		\pcod{p^*+1}{P(\zee)^2 \cdot \inparen{\Delta_R(g,\zee_i) - \frac{k-1}{k}(\delta_0 - \eps)}} = \pcod{p^*}{P_h(\zee)^2 \cdot \inparen{\Delta_R(g,\zee_i) - \frac{k-1}{k}(\delta_0 - \eps)}} < 0
	\]
	Further, for $i=p^*+1$,
	\begin{align*}
		\pcod{p^*+1}{P(\zee)^2 \cdot \inparen{\Delta_R(g,\zee_i) - \frac{k-1}{k}(\delta_0 - \eps)}} &= \pcod{p^*}{P_h(\zee)^2 \cdot \inparen{\Delta_R(g,h) - \frac{k-1}{k}(\delta_0 - \eps)}} \\
		&= \inparen{\Delta_R(g,h) - \frac{k-1}{k}(\delta_0 - \eps)} \cdot \pcod{p^*}{P_h(\zee)^2} \\
		&< 0
	\end{align*}
	Therefore, the infeasibility must be due to constraints of type (iii) in the SDP. In other words, there exist
	\begin{enumerate}
	\item a set $F\sub E$ with $|F|\leq \frac{t-2(p^*+1)d}{2}$,
	\item a $\sigma \in \Sigma_{\inn}^{[p^*+1]\times F}$ 
	%with $\pcod{p^*+1}{\zee_{[p^*+1],F,\sigma}} >0$
	, and 
	\item a pair $i,i'\in [p^*+1]$ with $i\neq i'$
	\end{enumerate}
	such that
	\[
		\pcod{p^*+1}{ \inparen{\Delta_L(\zee_i,\zee_{i'}) - \delta_{\out}^{\dec}} \cdot \indi{\zee_{[p^*+1],F} = \sigma}^2} < 0
	\]
	Case I. $i,i' \in [p^*]$:
	
	In this case,
	\begin{align*}
		0~>&~ \pcod{p^*+1}{ \inparen{\Delta_L(\zee_i,\zee_{i'}) - \delta_{\out}^{\dec}} \cdot \indi{\zee_{[p^*+1],F} = \sigma}} \\
		=~~ & \pcod{p^*+1}{ \inparen{\Delta_L(\zee_i,\zee_{i'}) - \delta_{\out}^{\dec}} \cdot \indi{\zee_{[p^*],F} = \sigma_1} \cdot \indi{\zee_{p^*+1,F} = \sigma_2}} \\
		=~~ & \pcod{p^*}{ \inparen{\Delta_L(\zee_i,\zee_{i'}) - \delta_{\out}^{\dec}} \cdot \indi{\zee_{[p^*],F} = \sigma_1} \cdot \indi{h_F = \sigma_2}} \\
		=~~ & \indi{h_F = \sigma_2} \cdot \pcod{p^*}{ \inparen{\Delta_L(\zee_i,\zee_{i'}) - \delta_{\out}^{\dec}} \cdot \indi{\zee_{[p^*],F} = \sigma_1} }
	\end{align*}
Because of the strict inequality with $0$, it must be the case that $\indi{h_F = \sigma_2} =1$, and the resulting equation contradicts the feasibility of $\pcod{p^*}{\cdot}$. So this case cannot happen.

\medskip
\noindent Case II. Without loss of generality, $i'=p^*+1$.
	
%	We may assume $\pcod{p^*+1}{\indi{\zee_{[p^*+1],F} = \sigma}^2} > 0$. If not, 
%	\begin{align*}
%		&\pcod{p^*+1}{\indi{\zee_{[p^*+1],F} = \sigma}^2} = 0 \\
%		\implies & \pcod{p^*+1}{ \Delta_L(\zee_i,\zee_{i'}) \cdot \indi{\zee_{[p^*+1],F} = \sigma}^2} < 0
%	\end{align*}
%	
%	Further, it can be checked that both $i$ and $i'$ cannot be $\leq p^*$, since $\pcod{p^*}{\cdot}$ is feasible for $\mathrm{SDP}(p^*,t)$. Without loss of generality, let $i' = p^*+1$. Then,
In this case,
	\begin{align*}
		0 &~>~  \pcod{p^*+1}{ \inparen{\Delta_L(\zee_i,\zee_{p^*+1}) - \delta_{\out}^{\dec}} \cdot \indi{\zee_{[p^*+1],F} = \sigma}} \\
		&=  \pcod{p^*+1}{ \inparen{\Delta_L(\zee_i,\zee_{p^*+1}) - \delta_{\out}^{\dec}} \cdot \indi{\zee_{[p^*],F} = \sigma_1} \cdot \indi{\zee_{p^*+1,F} = \sigma_2}} \\
		&=\pcod{p^*}{ \inparen{\Delta_L(\zee_i,h) - \delta_{\out}^{\dec}} \cdot \indi{\zee_{[p^*],F} = \sigma_1} \cdot \indi{h_F = \sigma_2}}\\
		&=\indi{h_F = \sigma_2} \cdot \pcod{p^*}{ \inparen{\Delta_L(\zee_i,h) - \delta_{\out}^{\dec}} \cdot \indi{\zee_{[p^*],F} = \sigma_1}}\\
%		&=\frac{\pcod{p^*}{ \inparen{\Delta_L(\zee_i,h) - \delta_{\out}^{\dec}} \cdot \indi{\zee_{[p^*],F} = \sigma_1}}}{\pcod{p^*}{\indi{\zee_{[p^*],F} = \sigma_1}}} \\
%		&=\pcod{p^*}{ \inparen{\Delta_L(\zee_i,h) - \delta_{\out}^{\dec}} \given \zee_{[p^*],F} = \sigma_1}
	\end{align*}
	Again, as before we must have $\indi{h_F = \sigma_2} =1$, so that
	\begin{align}\label{eq:without_h}
		\pcod{p^*}{ \inparen{\Delta_L(\zee_i,h) - \delta_{\out}^{\dec}} \cdot \indi{\zee_{[p^*],F} = \sigma_1}} < 0
	\end{align}
	Finally, we may also assume $\pcod{p^*}{\indi{\zee_{[p^*+1],F} =\sigma}} >0$. If not, then
	\[
		\pcod{p^*}{ \Delta_L(\zee_i,h) \cdot \indi{\zee_{[p^*],F} = \sigma_1}} < 0
	\]
	which is impossible since $\pcod{p^*}{P(\zee)} \geq 0$ for all polynomials that are sum of squares of polynomials. It is easy to verify that $\Delta_L(\zee_i,h) \cdot \indi{\zee_{[p^*],F} = \sigma_1}$ is a sum of squares of polynomials.
	
	Therefore, dividing \cref{eq:without_h} by $\pcod{p^*}{\indi{\zee_{[p^*+1],F} =\sigma}}$, we get that the set $F$ and the assignment $\sigma_1 \in \Sigma_{\inn}^{[p^*] \times F}$ have the property that
	\[
		\pcod{p^*}{\Delta_L(\zee_i,h) \given \zee_{[p^*],F} = \sigma_1} \leq \delta_{\out}^{\dec}
	\]
	as needed.
\end{proof}

\begin{proof}[Proof of \cref{lem:pseudocodeword_list_size}]
	Suppose $\mathrm{SDP}(k,t)$ is feasible with $t = 2kd\cdot \inparen{1+\frac{144k^{2k}|\Sigma_{\inn}|^{3kd}}{(\delta_{\out}^{\dec})^2\cdot \eps^2}}$, so that there exists a $k$-tuple of pseudocodewords $\tildeEx{\cdot}$ that satisfies all the constraints in the SDP. 
	Let $\eta = \frac{\delta_{\out}^{\dec}}{12k^k} \cdot \eps$.
	
	 Then $t\geq 2kd\cdot \inparen{1+\frac{|\Sigma_{\inn}|^{3kd}}{\eta^2}}$, so that \cref{lem:condition_for_eta_good} says that there exists a set $S \sub E$ of size at most $d\cdot \frac{|\Sigma_{\inn}|^{3kd}}{\eta^2}$ such that \[\pcod{S}{\cdot} \defeq \tildeEx{~\cdot \given \zee_{[k],S}} \] is $\eta$-good.
	
	Further, since the constraint $\Delta_R(g,\zee_i) < \frac{k-1}{k}(1-\rho-\eps)$ is \emph{respected} by $\tildeEx{\cdot}$, we can also conclude that for all $i\in [k]$,
	\begin{align}\label{eqn:zee_i_in_ball}
		\pcod{S}{\Delta_R(g,\zee_i)} < \frac{k-1}{k}(1-\rho-\eps)
	\end{align}
	Moreover, for any $i,i'\in [k]$ with $i\neq i'$,
	\[
		\pcod{S}{\inparen{\Delta_L(\zee_i,\zee_{i'}) - \delta_{\out}^{\dec}}} = \Ex{\sigma \sim \calD_{[k]\times S}}{\frac{\tildeEx{ \inparen{\Delta_L(\zee_i,\zee_{i'}) - \delta_{\out}^{\dec}} \cdot \indi{\zee_{[k], S} = \sigma}^2}}{\tildeEx{\indi{\zee_{[k], S}}^2}}} \geq 0
	\]
	where $\calD_{[k]\times S}$ is the local distribution on $\zee_{[k]\times S}$ according to $\tildeEx{\cdot}$. This means that
	\[
		\pcod{S}{\Delta_L(\zee_i,\zee_{i'})} \geq \delta_{\out}^{\dec}
	\]
	Since $\lambda \leq \frac{\delta_{\out}^{\dec}}{12k^k} \cdot \eps$ and $\eta \leq \frac{\delta_{\out}^{\dec}}{12k^k} \cdot \eps$, we can apply \cref{thm:sos_main} to $\pcod{S}{\cdot}$ with $\beta = \delta_{\out}^{\dec}$ and get
	\[
		\sum_{i\in [k]} \pcod{S}{\Delta_R(g,\zee_i)} \geq \frac{k-1}{k}(1-\rho-\eps) + \pcod{S}{\Delta_R(g,\zee_{[k]})} \geq \frac{k-1}{k}(1-\rho-\eps)
	\]
	which is contradicted by \cref{eqn:zee_i_in_ball}.
\end{proof}

Finally, we instantiate \cref{thm:sos_technical} with unique decodable outer codes to obtain codes that can be efficiently decoded up to the list decoding capacity.
\begin{corollary}\label{cor:algo-main}
For every $\rho, \eps \in (0,1)$ and $k \in \N$, there exist explicit inner codes $\calC_{\inn}$ and an infinite family explicit codes $\AELC \subseteq (\F_q^d)^n$ obtained via the AEL construction that satisfy: 
\begin{enumerate}
\item $\rho(\AELC) \geq \rho$
\item For any $g \in  (\F_q^d)^n$ and any $\calH \subseteq \AELC$ with $\abs{\calH} \leq k$ that
\[
\sum_{h \in \calH} \Delta(g,h) ~\geq~ (\abs{\calH}-1) \cdot (1 - \rho - \eps) \mper
\]
\item The alphabet size $q^d$ of the code $\AELC$ can be taken to be $2^{O(k^{3k}/\eps^9)}$.
\item $\AELC$ can be decoded from radius $\frac{k-1}{k}(1-\rho-\eps)$ in time $n^{2^{O(k^{4k}/\eps^{10})}}$ with a list of size at most $k-1$.
\end{enumerate}
\end{corollary}

\begin{proof}
	The first three properties can be ensured by instantiating the AEL amplification procedure as in \cref{cor:ael_instantiation}. We briefly recall the choices in that instantiation.
	
	\begin{itemize}
		\item We picked a graph $G$ with $d = O(k^{2k}/\eps^8)$ and $\lambda \leq \eps^4/(2^{21} k^k)$. 
		\item We picked $\calC_{\inn} \subseteq \Sigma_{\inn}^d$ to be a code with rate $\rho_{\inn} = \rho + \eps/4$, which is $(1 - \rho, k, \eps/2)$ average-radius list decodable with erasures. The alphabet size $|\Sigma_{\inn}|$ for $\calC_{\inn}$ was $2^{O(k + 1/\eps)}$.
		\item Finally, we picked $\calC_{\out} \subseteq (\Sigma_{\inn}^{\rho_{\inn} \cdot d})^n$ be an outer code with rate $\rho_{\out} = 1 - \eps/4$ and distance (say) $\delta_{\out} = \eps^3/2^{15}$. We note that this code can also be decoded from a radius $\delta_{\out}^{\dec} = \frac{\eps^3}{2^{17}}$ \cite{Zemor01, GRS23} in linear time.
	\end{itemize}

The choice of $\lambda$ above is sufficient to ensure $\lambda \leq \frac{\delta_{\out}^{\dec}}{12k^k}\cdot \eps$, and so we can use \cref{thm:sos_technical} to conclude that there is a deterministic algorithm, that takes as input $g\in (\Sigma_{\inn}^d)^R$, runs in time $n^{O\inparen{ \frac{d \cdot k^{3k}|\Sigma_{\inn}|^{3kd}}{(\delta_{\out}^{\dec})^2\cdot \eps^2}}}$, and outputs the list $\calL\inparen{g,\frac{k-1}{k}\inparen{1-\rho-\eps}}$. With the above choice of parameters, the exponent of $n$ in the runtime becomes $2^{O(\frac{k^{4k}}{\eps^{10}})}$.
%\begin{align*}
%	\frac{d \cdot k^{3k}|\Sigma_{\inn}|^{3kd}}{(\delta_{\out}^{\dec})^2\cdot \eps^2} &= \frac{k^{2k}}{6} \cdot k^{3k} \cdot \frac{1}{\eps^2} \frac{1024}{\eps^4} 2^{O(kd(k+\frac{1}{\eps}))} \\
%	&= \frac{k^{2k}}{6} \cdot k^{3k} \cdot \frac{1}{\eps^2} \frac{1024}{\eps^4} 2^{O(\frac{k^{3k}}{\eps^7})} \\
%	 &= 2^{O(\frac{k^{4k}}{\eps^8})}
%\end{align*}
%
%Given the above parameters, we have 
%$\rho(\AELC) ~\geq~ \rho_{\out} \cdot \rho_{\inn} ~=~ (1-\eps/4) \cdot (\rho + \eps/4) ~\geq~ \rho$.
%%
%Since $\lambda \leq \eps \cdot \delta_{\out}/(6k^k)$ and $\calC_{\inn}$ is $(1-\rho, k, \eps/2)$ average-radius list decodable with erasures, we can use \cref{thm:main_technical_avg} to conclude that $\AELC$ is $(1-\rho, k, \eps)$ average-radius list decodable (with erasures) which yields the second condition. Finally, we note that the alphabet size of the code $\AELC$ is $q^d = \exp\inparen{O((k + 1/\eps) \cdot (k^{2k}/\eps^6))} = \exp\inparen{k^{3k}/\eps^7}$, which proves the claim. 
%
%We only need to argue that the codes constructed there can be decoded efficiently.
\end{proof}

\section*{Acknowledgements}
We are grateful to the STOC 2025 reviewers, for a careful reading of the paper, and
for helpful comments to improve the presentation.

%\section{Todo}
%\input{todo}

\bibliographystyle{alphaurl}
\bibliography{macros,madhur}

\appendix
\section{Auxiliary SoS claims}\label{sec:appendix}
%\snote{We prove $\eta$-goodness and EML for pseudocodewords here. May move these to appendix eventually. .}

This appendix is adapted from \cite{JST23} to prove properties about $\eta$-good pseudocodewords and how one can obtain them via random conditioning. The key change is that we need these notions for $k$-tuples of pseudocodewords. Finally, we prove a generalization of the expander mixing lemma for such $k$-tuples of pseudowords.

%\tnote{does the proof use any of this cartesian notation? or is it only for the appendix?}

\subsection{Low average correlation}

\begin{lemma}\label{lem:avg_corr}
    Let $\tildeEx{\cdot}$ be an $\eta$-good $k$-tuple of pseudocodewords.
    
    Let $\{X_{\li}\}_{\li \in L}$ be a collection of $kd$-local functions on $\Sigma_{\inn}^{[k]\times E}$, such that $X_{\li}(f)$ only depends $f_{[k],\li}$. Likewise, let $\{Y_{\ri}\}_{\ri \in R}$ be a collection of $kd$-local functions on $\Sigma_{\inn}^{[k]\times E}$, such that $Y_{\ri}(f)$ only depends $f_{[k],\ri}$.
    
    Then,
    \begin{align*}
         \tildeEx{\Ex{\li,\ri}{X_{\li}(\zee)\cdot Y_{\ri}(\zee) }} - \tildeEx{\Ex{\li}{X_{\li}(\zee)}} \cdot \tildeEx{\Ex{\ri}{Y_{\ri}(\zee)}} \leq \eta \cdot \inparen{ \max_{\li} \norm{X_{\li}}_{\infty}} \cdot \inparen{\max_{\ri} \norm{Y_{\ri}}}_{\infty}
    \end{align*}
\end{lemma}
\begin{proof}
    \begin{align*}
    	& \tildeEx{\Ex{\li,\ri}{X_{\li}(\zee)\cdot Y_{\ri}(\zee) }} - \tildeEx{\Ex{\li}{X_{\li}(\zee)}} \cdot \tildeEx{\Ex{\ri}{Y_{\ri}(\zee)}}  \\
    	=~ & \Ex{\li,\ri}{\tildeEx{X_{\li}(\zee)\cdot Y_{\ri}(\zee) } - \tildeEx{X_{\li}(\zee)} \cdot \tildeEx{Y_{\ri}(\zee)} } \\
    	=~ & \Ex{\li,\ri}{ \sum_{\substack{\alpha \in \Sigma_{\inn}^{[k]\times N(\li)} \\ \beta \in \Sigma_{\inn}^{[k]\times N(\ri)}}} X_{\li}(\alpha) \cdot Y_{\ri}(\beta) \tildeEx{Z_{[k],N(\li),\alpha} \cdot Z_{[k],N(\ri),\beta} } - \sum_{\substack{\alpha \in \Sigma_{\inn}^{[k]\times N(\li)} \\ \beta \in \Sigma_{\inn}^{[k]\times N(\ri)}}} X_{\li}(\alpha) \cdot Y_{\ri}(\beta) \tildeEx{Z_{[k],N(\li),\alpha}} \cdot \tildeEx{Z_{[k],N(\ri),\beta} } } \\
    	\leq~ & \Ex{\li,\ri}{ \norm{X_{\li}}_{\infty} \cdot \norm{Y_{\ri}}_{\infty} \sum_{\substack{\alpha \in \Sigma_{\inn}^{[k]\times N(\li)} \\ \beta \in \Sigma_{\inn}^{[k]\times N(\ri)}}} \abs*{ \tildeEx{Z_{[k],N(\li),\alpha} \cdot Z_{[k],N(\ri),\beta} } -  \tildeEx{Z_{[k],N(\li),\alpha}} \cdot \tildeEx{Z_{[k],N(\ri),\beta} } } } \\
    	=~ & \Ex{\li,\ri}{ \norm{X_{\li}}_{\infty} \cdot \norm{Y_{\ri}}_{\infty} \cdot  \tildeCov{\zee_{[k],N(\li)}}{\zee_{[k],N(\ri)}} } \\
    	\leq~ & \inparen{\max_{\li} \norm{X_{\li}}_{\infty}} \cdot \inparen{\max_{\ri} \norm{Y_{\ri}}_{\infty}} \cdot \Ex{\li,\ri}{\tildeCov{\zee_{[k],N(\li)}}{\zee_{[k],N(\ri)}}} \\
    	\leq~ & \eta \cdot \inparen{\max_{\li} \norm{X_{\li}}_{\infty}} \cdot \inparen{\max_{\ri} \norm{Y_{\ri}}_{\infty}}.\qedhere
    \end{align*}
\end{proof}

\subsection{Conditioning reduces variance}

We next show how to obtain the $\eta$-good property by randomly conditioning the $k$-tuple of pseudocodewords. This is a straightforward generalization of the argument in \cite{JST23}, which can be seen as the $k=1$ case.

Just as in \cite{JST23}, we start with a lemma from \cite{BRS11} that quantifies the decrease in variance of the local distribution corresponding to a set $S$ when conditioning on another small set $T$, in terms of covariances between the two sets $S$ and $T$.

\begin{lemma}[{\cite[Lemma 5.2]{BRS11}}]
	Let $\tildeEx{\cdot}$ be a $k$-tuple of pseudocodewords of degree $t$. Let $S,T$ be subsets of $[k]\times E$ of size at most $t/2$ each. Then,
	\[
		\tildeVar{\zee_S \given \zee_T} \leq \tildeVar{\zee_S} - \frac{1}{|\Sigma_{\inn}|^{|T|}} \cdot \sum_{\substack{\alpha \in \Sigma_{\inn}^S \\ \beta \in \Sigma_{\inn}^T}} \frac{\inparen{\tildeCov{\zee_{S,\alpha}}{\zee_{T,\beta}}}^2}{\tildeVar{\zee_{T,\beta}}}
	\]
\end{lemma}

Now we use this lemma to track the decrease in average (pseudo-)variance $\Ex{\li \in L}{\tildeVar{\zee_{[k],N(\li)}}}$ when conditioning on $\zee_{[k],N(\ri)}$ for a random $\ri\in R$.

\begin{lemma}\label{lem:avg_conditioning}
	\begin{align*}
		\Ex{\ri\in R}{\Ex{\li\in L}{\tildeVar{\zee_{[k],N(\li)} \given \zee_{[k],N(\ri)}}}} \leq \Ex{\li\in L}{\tildeVar{\zee_{[k],N(\li)}}} - \frac{1}{|\Sigma_{\inn}|^{3kd}} \cdot \inparen{\Ex{\li,\ri}{\tildeCov{\zee_{[k],N(\li)}}{\zee_{[k],N(\ri)}}}}^2
	\end{align*}
\end{lemma}

\begin{proof}
	\begin{align*}
		&~~~\Ex{\ri\in R}{\Ex{\li\in L}{\tildeVar{\zee_{[k],N(\li)} \given \zee_{[k],N(\ri)}}}} \\
		~&\leq~ \Ex{\ri\in R}{\Ex{\li\in L}{\tildeVar{\zee_{[k],N(\li)}} - \frac{1}{|\Sigma_{\inn}|^{kd}} \sum_{\substack{\alpha \in \Sigma_{\inn}^{[k]\times N(\li)} \\ \beta \in \Sigma_{\inn}^{[k]\times N(\ri)}}} \frac{ \inparen{ \tildeCov{\zee_{[k],N(\li),\alpha}}{\zee_{[k],N(\ri),\beta}}}^2}{\tildeVar{\zee_{[k],N(\ri),\beta} }}}} \\
		~&\leq~ \Ex{\li\in L}{\tildeVar{\zee_{[k],N(\li)}}} - \frac{1}{|\Sigma_{\inn}|^{kd}} \cdot \Ex{\substack{\li\in L \\ \ri\in R}}{ \sum_{\substack{\alpha \in \Sigma_{\inn}^{[k]\times N(\li)} \\ \beta \in \Sigma_{\inn}^{[k]\times N(\ri)}}} \inparen{ \tildeCov{\zee_{[k],N(\li),\alpha}}{\zee_{[k],N(\ri),\beta}}}^2} \\
		~&\leq~ \Ex{\li\in L}{\tildeVar{\zee_{[k],N(\li)}}} - \frac{1}{|\Sigma_{\inn}|^{3kd}} \cdot \Ex{\substack{\li\in L \\ \ri\in R}}{ \inparen{ \sum_{\substack{\alpha \in \Sigma_{\inn}^{[k]\times N(\li)} \\ \beta \in \Sigma_{\inn}^{[k]\times N(\ri)}}}  \abs*{\tildeCov{\zee_{[k],N(\li),\alpha}}{\zee_{[k],N(\ri),\beta}}}}^2} \\
		~&=~ \Ex{\li\in L}{\tildeVar{\zee_{[k],N(\li)}}} - \frac{1}{|\Sigma_{\inn}|^{3kd}} \cdot \Ex{\substack{\li\in L \\ \ri\in R}}{ \inparen{ \tildeCov{\zee_{[k],N(\li)}}{\zee_{[k],N(\ri)}}}^2} \\
		~&=~ \Ex{\li\in L}{\tildeVar{\zee_{[k],N(\li)}}} - \frac{1}{|\Sigma_{\inn}|^{3kd}} \cdot \inparen{\Ex{\substack{\li\in L \\ \ri\in R}}{  \tildeCov{\zee_{[k],N(\li)}}{\zee_{[k],N(\ri)}}}}^2
	\end{align*}
\end{proof}

Now that we can use average correlation to quantify the decrease in variance, we show that for a $k$-tuple of pseudocodewords with high-enough degree, there always exists a constant-sized subset of $R$ such that randomly conditioning on this set gives an $\eta$-good $k$-tuple of pseudocodewords. 
%If not, we can keep decreasing variance which must remain non-negative.

\begin{lemma}\label{lem:condition_for_eta_good}
	Let $\tildeEx{\cdot}$ be a $k$-tuple of pseudocodewords of degree $t \geq 2kd\inparen{1+\frac{|\Sigma_{\inn}|^{3kd}}{\eta^2}}$. Then there exists a set $S \subseteq E$ of size at most $d\cdot \frac{|\Sigma_{\inn}|^{3kd}}{\eta^2}$ such that the conditioned pseudocodeword $\tildeEx{\cdot \given \zee_{[k],S}}$ is $\eta$-good.
\end{lemma}

\begin{proof}
%	Let $\pcod{0}{\cdot} = \tildeEx{\cdot}$.
	For an integer $C>0$ to be chosen later, consider the following sequence of pseudocodewords, obtained by sequentially conditioning on $\zee_{[k],\ri}$ for a random $\ri$. For $c\in [C+1]$, define
	\[
		\Psi(c) = \Ex{\ri_1,\ri_2, \cdots \ri_c}{ \Ex{\li}{\tildeVar{ \zee_{[k],\li} \given \zee_{[k],\ri_1},\zee_{[k],\ri_2},\cdots ,\zee_{[k],r_c}}}}
	\]
	We have that
	\[
		1 \geq \Psi(0) \geq \Psi(1) \geq \cdots \Psi(C+1) \geq 0
	\]
	Therefore, there must exist a $c^*$, with $0\leq c^* \leq C$, such that
	\[
		\Psi(c^*) - \Psi(c^*+1) \leq \frac{1}{C+1} < \frac{1}{C} = \frac{\eta^2}{|\Sigma_{\inn}|^{3kd}}
	\]
	For an $R'\sub R$, let us denote $\pcod{R'}{\cdot} = \tildeEx{~\cdot~ \given \zee_{[k],N(R')}}$, and its associated variance and covariance operators be $\widetilde{\operatorname{Var}}^{(R')}$ and $\widetilde{\operatorname{Cov}}^{(R')}$. Then,
	\[
		\Psi(c^*) - \Psi(c^*+1) = \Ex{R' \sim R^{c^*}}{\Ex{\li}{\widetilde{\operatorname{Var}}^{(R')}\left[\zee_{[k],\li}\right]}} - \Ex{R'\sim R^{c^*}}{ \Ex{\ri_{c^*+1}}{\Ex{\li}{\widetilde{\operatorname{Var}}^{(R')}\left[\zee_{[k],\li} \given \zee_{[k], \ri_{c^*+1} }\right]}}}
	\]
	The set $R'$ above is sampled by sampling $c^*$ times from $R$ uniformly at random and with replacement. Therefore, there exists a set $R'$ of size at most $c^*$ such that
	\[
		\Ex{\li}{\widetilde{\operatorname{Var}}^{(R')}\left[\zee_{[k],\li}\right]} -  \Ex{\ri_{c^*+1}}{\Ex{\li}{\widetilde{\operatorname{Var}}^{(R')}\left[\zee_{[k],\li} \given \zee_{[k], \ri_{c^*+1} }\right]}} < \frac{\eta^2}{|\Sigma_{\inn}|^{3kd}}
	\]
	Applying the contrapositive of \cref{lem:avg_conditioning} to $\pcod{R'}{\cdot}$, we get that
		\begin{align*}
		&\Ex{\li,\ri}{\widetilde{ \operatorname{Cov}}^{(R')} \left[ \zee_{[k],\li}, \zee_{[k],\ri} \right]}\leq \eta \\
		\implies &\Ex{\li,\ri}{\widetilde{\operatorname{Cov}}\left[\zee_{[k],\li}, \zee_{[k],\ri} \given \zee_{[k],N(R')} \right]} \leq \eta
	\end{align*}
	Therefore, the set $S = N(R')$ proves the lemma.
%	Let us denote $\pcod{S}{\cdot} = \Ex{\ri_1,\cdots ,\ri_{c^*}}{\tildeEx{~\cdot~ \given \zee_{[k],S},\zee_{[k],\ri_2},\cdots ,\zee_{[k],\ri_{c^*}}}}$, and its associated variance and covariance operators be $\widetilde{\operatorname{Var}}^{(c^*)}$ and $\widetilde{\operatorname{Cov}}^{(c^*)}$. Then 
%	\[
%		\Psi(c^*+1) - \Psi(c^*) = \Ex{\li}{\widetilde{\operatorname{Var}}^{(c^*)}\left[\zee_{[k],\li}\right]} - \Ex{\ri_{c^*+1}}{\Ex{\li}{\widetilde{\operatorname{Var}}^{(c^*)}\left[\zee_{[k],\li} \given \zee_{[k],\ri_{c^*+1}}\right]}}
%	\]
%	Applying the contrapositive of \cref{lem:avg_conditioning} to $\pcod{c^*}{\cdot}$, we get that
%	\begin{align*}
%		&\Ex{\li,\ri}{\widetilde{ \operatorname{Cov}}^{(c^*)} \left[ \zee_{[k],\li}, \zee_{[k],\ri} \right]}\leq \eta \\
%		\implies &\Ex{\ri_1,\cdots ,\ri_{c^*}}{\widetilde{\operatorname{Cov}}\left[\zee_{[k],\li}, \zee_{[k],\ri} \given \zee_{[k],\ri_1}, \cdots ,\zee_{[k],\ri_{c^*}} \right]} \leq \eta
%	\end{align*}
%	Therefore, there exists a set of $c^*$ right vertices $\ri_1,\ri_2,\cdots ,\ri_{c^*}$ such that conditioning on $\zee_{[k], \{\ri_1,\ri_2,\cdots ,\ri_{c^*}\}}$ makes the pseudocodeword $\eta$-good.
\end{proof}

\subsection{Expander Mixing Lemma}

\begin{lemma}[EML for pseudoexpectations]\label{lem:eml_sos}
	Let $\tildeEx{\cdot}$ be a $k$-tuple of pseudocodewords of degree at least $2k d$. Let $G = (L,R,E)$ be a bipartite $\lambda$-expander. Recall that $\li \sim \ri$ if and only if $(\li,\ri)$ is an edge in $G$.
	
	Let $\{X_{\li}\}_{\li \in L}$ be a collection of $kd$-local functions on $\Sigma_{\inn}^{[k]\times E}$, such that $X_{\li}(f)$ only depends $f_{[k],\li}$. Likewise, let $\{Y_{\ri}\}_{\ri \in R}$ be a collection of $kd$-local functions on $\Sigma_{\inn}^{[k]\times E}$, such that $Y_{\ri}(f)$ only depends $f_{[k],\ri}$. 
	Then,
	\begin{align*}
		\abs*{\tildeEx{\Ex{\li\sim \ri}{X_{\li}(\zee) \cdot Y_{\ri}(\zee)} - \Ex{\li,\ri}{X_{\li}(\zee) \cdot Y_{\ri}(\zee)}}} \leq \lambda \cdot \sqrt{\tildeEx{\Ex{\li}{X_{\li}(\zee)^2}}} \cdot \sqrt{\tildeEx{\Ex{\ri}{Y_{\ri}(\zee)^2}}}
	\end{align*}
\end{lemma}
\begin{proof}
    Let $A_G$ be the $L \times R$ normalized biadjacency matrix of $G$, so that
    \[
    	A_G(\li,\ri) = \begin{cases}
    		\frac{1}{d}, \qquad (\li,\ri) \in E \\
    		0, \qquad (\li,\ri)\not\in E
    	\end{cases}
    \]
    The singular value decomposition for $A_G$ can be written as
    \[
    	A_G = \sum_{i=1}^n \sigma_i u_i v_i^T
    \]
    with $\sigma_1 \geq \cdots \geq \sigma_n \geq 0$, and $\{u_i\}_{i\in [n]}$, $\{v_i\}_{i \in [n]}$ being two orthonormal bases of $\R^n$. Moreover, since the graph is regular, $\sigma_1 =1$ and $u_1 = v_1 = \frac{1}{\sqrt{n}} \one$.
    
    Let $X(\zee)$ be the vector-valued local function defined as the $n$-dimensional vector with coordinates corresponding to $\li \in L$, and the respective entry being $X_{\li}(\zee)$. Likewise, let $Y(\zee)$ be the vector-valued local function so that coordinates correspond to $\ri\in R$, and the entries being $\{Y_{\ri}(\zee)\}_{\ri\in R}$. Then,
    \begin{align*}
    	&\abs*{\tildeEx{\Ex{\li\sim \ri}{X_{\li}(\zee) \cdot Y_{\ri}(\zee)} - \Ex{\li,\ri}{X_{\li}(\zee) \cdot Y_{\ri}(\zee)}}} \\
    	=~ &\abs*{ \tildeEx{\frac{1}{n} X(\zee)^T A_G Y(\zee)} - \tildeEx{\frac{1}{n} X(\zee)^T (\frac{1}{\sqrt{n}} \one) (\frac{1}{\sqrt{n}} \one)^T Y(\zee)}} \\
    	=~ &\frac{1}{n} \abs*{\tildeEx{X(\zee)^T \inparen{A_G - (\frac{1}{\sqrt{n}} \one) (\frac{1}{\sqrt{n}} \one)^T}Y(\zee)}} \\
    	=~ &\frac{1}{n} \abs*{\tildeEx{X(\zee)^T \inparen{\sum_
    	{i=2}^n \sigma_i u_iv_i^T}Y(\zee)}} \\
    	\leq~ &\frac{1}{n} \abs*{\tildeEx{\sum_{i=2}^n \sigma_i ~ \inparen{X(\zee)^T u_i}\cdot \inparen{Y(\zee)^T v_i}}} \\
    	=~ &\frac{1}{n} \abs*{\tildeEx{\sum_{i=2}^n \inparen{\sqrt{\sigma_i} \cdot X(\zee)^T u_i}\cdot \inparen{\sqrt{\sigma_i}\cdot Y(\zee)^T v_i}}} \\
    	\leq~ &\frac{1}{n} \inparen{\tildeEx{\sum_{i=2}^n \inparen{\sqrt{\sigma_i} \cdot X(\zee)^T u_i}^2}}^{1/2} \cdot \inparen{\tildeEx{\sum_{i=2}^n \inparen{\sqrt{\sigma_i}\cdot Y(\zee)^T v_i}^2}}^{1/2} && (\text{{SoS Cauchy--Schwarz}})\\
    	\leq~ &\frac{1}{n}\cdot \lambda \cdot \inparen{\tildeEx{\sum_{i=2}^n \inparen{ X(\zee)^T u_i}^2}}^{1/2} \cdot \inparen{\tildeEx{\sum_{i=2}^n \inparen{ Y(\zee)^T v_i}^2}}^{1/2} \\
%    	\leq~ &\frac{1}{n} \lambda \cdot \sum_{i=2}^n \abs*{ \tildeEx{\inparen{X(\zee)^T u_i}\cdot \inparen{Y(\zee)^T v_i}}} \\
%    	\leq~ &\frac{1}{n} \lambda \cdot \parens*{ \tildeEx{ \sum_{i=2}^n \inparen{X(\zee)^T u_i}^2}}^{1/2} \cdot \parens*{ \tildeEx{\sum_{i=2}^n \inparen{Y(\zee)^T v_i}^2}}^{1/2} && [\text{\snote{SoS Cauchy Schwarz}}]\\
    	=~ &\frac{1}{n} \lambda \cdot \inparen{ \tildeEx{ X(\zee)^T X(\zee) - (X(\zee)^T u_1)^2}}^{1/2} \cdot \inparen{ \tildeEx{ Y(\zee)^T Y(\zee) - (Y(\zee)^T v_1)^2}}^{1/2} && (\text{{Using orthonormality}})\\
    	\leq~ &\frac{1}{n} \lambda \cdot \inparen{ \tildeEx{ X(\zee)^T X(\zee)}}^{1/2} \cdot \inparen{ \tildeEx{ Y(\zee)^T Y(\zee)}}^{1/2} \\
    	 =~ & \lambda \cdot \inparen{ \tildeEx{\Ex{\li}{ X_{\li}(\zee)^2}}}^{1/2} \cdot \inparen{ \tildeEx{ \Ex{\ri}{Y_{\ri}(\zee)^2}}}^{1/2}. \qedhere
    \end{align*}
\end{proof}

\end{document}